\documentclass[11pt]{article}
\usepackage{amsfonts,amsmath}
\usepackage{hyperref}
\usepackage{latexsym}
\usepackage{amsthm}
\usepackage{amscd}
\usepackage{epsfig}
\usepackage{tikz}
\usepackage{mathtools}
\usepackage{graphicx}
\usepackage{caption}
\usepackage{subcaption}
\usetikzlibrary{decorations.pathreplacing}

\usepackage{float}
\usepackage{soul}
\usepackage{bm}

\numberwithin{equation}{section}

\newtheorem{thm}{Theorem}[section]
\newtheorem{prop}[thm]{Proposition}
\newtheorem{lem}[thm]{Lemma}
\newtheorem{cor}[thm]{Corollary}

\newcommand{\bib}{\bibitem}
\newcommand{\TT}{|T|}

\newcommand{\dm}{\mathsf{D}}
\newcommand{\km}{\mathsf{K}}
\newcommand{\om}{\mathsf{O}}
\newcommand{\tm}{\mathsf{T}}
\newcommand{\upto}{\nearrow}
\newcommand{\downto}{\searrow}

\newcommand{\bb}{\begin{equation}}
\newcommand{\ee}{\end{equation}}

\newcommand{\half}{\frac{1}{2}}

\theoremstyle{definition}

\theoremstyle{remark}

\def\IB{\relax\hbox{$\inbar\kern-.3em{\rm B}$}}
\def\IC{\relax\hbox{$\inbar\kern-.3em{\rm C}$}}
\def\ID{\relax\hbox{$\inbar\kern-.3em{\rm D}$}}
\def\IE{\relax\hbox{$\inbar\kern-.3em{\rm E}$}}
\def\IF{\relax\hbox{$\inbar\kern-.3em{\rm F}$}}
\def\IG{\relax\hbox{$\inbar\kern-.3em{\rm G}$}}
\def\IGa{\relax\hbox{${\rm I}\kern-.18em\Gamma$}}
\def\IH{\relax{\rm I\kern-.18em H}}
\def\IK{\relax{\rm I\kern-.18em K}}
\def\IL{\relax{\rm I\kern-.18em L}}
\def\IP{\relax{\rm I\kern-.18em P}}
\def\IR{\relax{\rm I\kern-.18em R}}
\def\IZ{\relax{\rm Z\kern-.5em Z}}

\definecolor{dblue}{rgb}{.61,.61,1}
\definecolor{lightblue}{rgb}{.61,.61,1}
\definecolor{altblue}{rgb}{.61,.61,1}

\tikzset{cross/.style={cross out, draw=black, minimum size=2*(#1-\pgflinewidth), inner sep=0pt, outer sep=0pt},
	cross/.default={1pt}}

\newcommand{\Gexample}{\begin{tikzpicture}
    \fill[opacity=0.25, lightblue] (0,0) circle [radius=3 cm];

	\draw[line width = 0.03 cm, orange,opacity=0.99] (0,0) -- (90:1);
	\draw[line width=0.035cm] (0,0) -- (30:1 cm);
	\draw[line width=0.035cm] (0,0) -- (220:1);
	\draw[line width=0.035cm] (0,0) -- (300:1);
	\draw[line width=0.035cm] (0,0) -- (170:1);
	
	\draw[line width=0.035cm] (90:1) -- (45:2);
	\draw[line width=0.035cm] (30:1) -- (45:2);
	\draw[line width=0.035cm] (90:1) -- (120:2);
	\draw[line width=0.035cm] (170:1) -- (120:2);
	\draw[line width=0.035cm] (170:1) -- (160:2);
	\draw[line width=0.035cm] (170:1) -- (190:2);
	\draw[line width=0.035cm] (30:1) -- (330:2);
	\draw[line width=0.035cm] (220:1) -- (190:2);
	\draw[line width=0.035cm] (220:1) -- (260:2);
	\draw[line width=0.035cm] (300:1) -- (260:2);
	\draw[line width=0.035cm] (300:1) -- (330:2);
	\draw[line width=0.035cm] (300:1) -- (290:2);
    
	\draw[line width=0.035cm] (45:2) -- (0:3);
	\draw[line width=0.035cm] (45:2) -- (18:3);
	\draw[line width=0.035cm] (45:2) -- (70:3);
	\draw[line width=0.035cm] (120:2) -- (70:3);
	\draw[line width=0.035cm] (120:2) -- (110:3);
	\draw[line width=0.035cm] (120:2) -- (130:3);
	\draw[line width=0.035cm] (160:2)-- (130:3);
	\draw[line width=0.035cm] (160:2)-- (160:3);
	\draw[line width=0.035cm] (160:2)-- (195:3);
	\draw[line width=0.035cm] (190:2) -- (195:3);
	\draw[line width=0.035cm] (190:2) -- (220:3);
	\draw[line width=0.035cm] (190:2) -- (240:3);
	\draw[line width=0.035cm] (260:2)--(240:3);
	\draw[line width=0.035cm] (260:2)--(270:3);
	\draw[line width=0.035cm] (260:2)--(300:3);
	\draw[line width=0.035cm] (330:2)--(0:3);
	\draw[line width=0.035cm] (330:2)--(300:3);
	\draw[line width=0.035cm] (290:2)--(300:3);
	
	\draw[red, line width=0.035cm] (0,0) circle [radius=1 cm];
	\draw[red, line width=0.035cm] (0,0) circle [radius=2 cm];
	\draw[red, line width=0.035cm] (0,0) circle [radius=3 cm];
	
	\node [fill=black,inner sep=1.25pt,label=right:$x$] at (0,0) {};
	\node [fill=black,inner sep=1.25pt,label={[font=\small,rotate=0]above:$v_1$}] at (90:1){};
	\node [fill=black,inner sep=1.25pt,label={[label distance=-1.5mm,font=\small,rotate=0]north east:$v_2$}] at (45:2){};
	\node [fill=black,inner sep=1.25pt,label={[label distance=1mm,font=\small,rotate=0]west:$v_3$}] at (0:3){};
	
	\end{tikzpicture}}

   \newcommand{\densexample}{\begin{tikzpicture}
    \fill[opacity=0.25, lightblue] (0,0) circle [radius=3 cm];
    
	\draw[line width = 0.03 cm, orange,opacity=0.99] (0,0) -- (90:1);
	\draw[line width=0.035cm] (0,0) -- (30:1 cm);
	\draw[line width=0.035cm] (0,0) -- (220:1);
	\draw[line width=0.035cm] (0,0) -- (300:1);
	\draw[line width=0.035cm] (0,0) -- (170:1);
	
	\draw[line width=0.035cm] (90:1) -- (45:2);
	\draw[line width=0.035cm] (30:1) -- (45:2);
	\draw[line width=0.035cm] (90:1) -- (120:2);
	\draw[line width=0.035cm] (170:1) -- (120:2);
	\draw[line width=0.035cm] (170:1) -- (160:2);
	\draw[line width=0.035cm] (170:1) -- (190:2);
	\draw[line width=0.035cm] (30:1) -- (330:2);
	\draw[line width=0.035cm] (220:1) -- (190:2);
	\draw[line width=0.035cm] (220:1) -- (260:2);
	\draw[line width=0.035cm] (300:1) -- (260:2);
	\draw[line width=0.035cm] (300:1) -- (330:2);
	\draw[line width=0.035cm] (300:1) -- (290:2);
    
	\draw[line width=0.035cm] (45:2) -- (0:3);
	\draw[line width=0.035cm] (45:2) -- (18:3);
	\draw[line width=0.035cm] (45:2) -- (70:3);
	\draw[line width=0.035cm] (120:2) -- (70:3);
	\draw[line width=0.035cm] (120:2) -- (110:3);
	\draw[line width=0.035cm] (120:2) -- (130:3);
	\draw[line width=0.035cm] (160:2)-- (130:3);
	\draw[line width=0.035cm] (160:2)-- (160:3);
	\draw[line width=0.035cm] (160:2)-- (195:3);
	\draw[line width=0.035cm] (190:2) -- (195:3);
	\draw[line width=0.035cm] (190:2) -- (220:3);
	\draw[line width=0.035cm] (190:2) -- (240:3);
	\draw[line width=0.035cm] (260:2)--(240:3);
	\draw[line width=0.035cm] (260:2)--(270:3);
	\draw[line width=0.035cm] (260:2)--(300:3);
	\draw[line width=0.035cm] (330:2)--(0:3);
	\draw[line width=0.035cm] (330:2)--(300:3);
	\draw[line width=0.035cm] (290:2)--(300:3);
	
	\draw[line width= 0.07 cm, blue] (0:1) .. controls (5:0.7) .. (30:0.7) ;
	\draw[line width= 0.07 cm, blue] (170:0.7) .. controls (180:0.8) .. (183:1) ;
	\draw[line width= 0.07 cm, blue] (30:0.7) .. controls (80:0.7) and (120:0.6) .. (170:0.7) ;
	\draw[line width= 0.07 cm, blue] (183:1) .. controls  (185:1.3) and (160:1.3)  .. (138:1.3) ;
	\draw[line width= 0.07 cm, blue] (170:3) .. controls  (160:2.5) and (110:2.5) .. (103:3) ;
	\draw[line width= 0.07 cm, blue] (205:1) .. controls (235:3) and (180:2.7) .. (182:3);
	\draw[line width= 0.07 cm, blue] (312:1) .. controls (190:0.7) .. (205:1) ;
	\draw[line width= 0.07 cm, blue] (0:1) .. controls (20:1.8) and (90:2) .. (138:1.3) ;
	\draw[line width= 0.07 cm, blue] (312:1) .. controls (320:2.2) and (220:1.5) .. (250:2) ;
	\draw[line width= 0.07 cm, blue] (250:2) .. controls (260:2.7) and (285:2) .. (307:3) ;
	\draw[line width= 0.07 cm, blue] (85:3) .. controls (80:2.4) and (20:3) .. (350:3) ;

	\draw[red, line width=0.035cm] (0,0) circle [radius=1 cm];
	\draw[red, line width=0.035cm] (0,0) circle [radius=2 cm];
	\draw[red, line width=0.035cm] (0,0) circle [radius=3 cm];
	
	\end{tikzpicture}}
	
	\newcommand{\densexamplecapped}{
	\begin{tikzpicture}
    \fill[opacity=0.25, lightblue]  (0:3)--(9:4)--(18:3)--(44:4)--(70:3)--(90:4)--(110:3)--(120:4)--(130:3)--(145:4)--(160:3)--(177:4)--(195:3)--(207:4)--(220:3)--(230:4)--(240:3)--(255:4)--(270:3)--(285:4)--(300:3)--(330:4)--(0:3);
    
    \draw[line width=0.035cm, dotted] (0:3)--(9:4)--(18:3)--(44:4)--(70:3)--(90:4)--(110:3)--(120:4)--(130:3)--(145:4)--(160:3)--(177:4)--(195:3)--(207:4)--(220:3)--(230:4)--(240:3)--(255:4)--(270:3)--(285:4)--(300:3)--(330:4)--(0:3);
    
    \node [fill=red,inner sep=1.25pt] at (9:4) {};
    \node [fill=red,inner sep=1.25pt] at (44:4) {};
    \node [fill=red,inner sep=1.25pt] at (90:4) {};
    \node [fill=red,inner sep=1.25pt] at (120:4) {};
    \node [fill=red,inner sep=1.25pt] at (145:4) {};
    \node [fill=red,inner sep=1.25pt] at (177:4) {};
    \node [fill=red,inner sep=1.25pt] at (207:4) {};
    \node [fill=red,inner sep=1.25pt] at (230:4) {};
    \node [fill=red,inner sep=1.25pt] at (255:4) {};
    \node [fill=red,inner sep=1.25pt] at (285:4) {};
    \node [fill=red,inner sep=1.25pt] at (330:4) {};
    
	\draw[line width = 0.03 cm, orange,opacity=0.99] (0,0) -- (90:1);
	\draw[line width=0.035cm] (0,0) -- (30:1 cm);
	\draw[line width=0.035cm] (0,0) -- (220:1);
	\draw[line width=0.035cm] (0,0) -- (300:1);
	\draw[line width=0.035cm] (0,0) -- (170:1);
	
	\draw[line width=0.035cm] (90:1) -- (45:2);
	\draw[line width=0.035cm] (30:1) -- (45:2);
	\draw[line width=0.035cm] (90:1) -- (120:2);
	\draw[line width=0.035cm] (170:1) -- (120:2);
	\draw[line width=0.035cm] (170:1) -- (160:2);
	\draw[line width=0.035cm] (170:1) -- (190:2);
	\draw[line width=0.035cm] (30:1) -- (330:2);
	\draw[line width=0.035cm] (220:1) -- (190:2);
	\draw[line width=0.035cm] (220:1) -- (260:2);
	\draw[line width=0.035cm] (300:1) -- (260:2);
	\draw[line width=0.035cm] (300:1) -- (330:2);
	\draw[line width=0.035cm] (300:1) -- (290:2);
    
	\draw[line width=0.035cm] (45:2) -- (0:3);
	\draw[line width=0.035cm] (45:2) -- (18:3);
	\draw[line width=0.035cm] (45:2) -- (70:3);
	\draw[line width=0.035cm] (120:2) -- (70:3);
	\draw[line width=0.035cm] (120:2) -- (110:3);
	\draw[line width=0.035cm] (120:2) -- (130:3);
	\draw[line width=0.035cm] (160:2)-- (130:3);
	\draw[line width=0.035cm] (160:2)-- (160:3);
	\draw[line width=0.035cm] (160:2)-- (195:3);
	\draw[line width=0.035cm] (190:2) -- (195:3);
	\draw[line width=0.035cm] (190:2) -- (220:3);
	\draw[line width=0.035cm] (190:2) -- (240:3);
	\draw[line width=0.035cm] (260:2)--(240:3);
	\draw[line width=0.035cm] (260:2)--(270:3);
	\draw[line width=0.035cm] (260:2)--(300:3);
	\draw[line width=0.035cm] (330:2)--(0:3);
	\draw[line width=0.035cm] (330:2)--(300:3);
	\draw[line width=0.035cm] (290:2)--(300:3);

	\draw[line width= 0.07 cm, blue] (0:1) .. controls (5:0.7) .. (30:0.7) ;
	\draw[line width= 0.07 cm, blue] (170:0.7) .. controls (180:0.8) .. (183:1) ;
	\draw[line width= 0.07 cm, blue] (30:0.7) .. controls (80:0.7) and (120:0.6) .. (170:0.7) ;
	\draw[line width= 0.07 cm, blue] (183:1) .. controls  (185:1.3) and (160:1.3)  .. (138:1.3) ;
	\draw[line width= 0.07 cm, blue] (170:3) .. controls  (160:2.5) and (110:2.5) .. (103:3) ;
	\draw[line width= 0.07 cm, blue] (170:3) .. controls  (167:3.9) and (110:4.2) .. (103:3) ;
	
	\draw[line width= 0.07 cm, blue] (205:1) .. controls (235:3) and (180:2.7) .. (182:3);
	\draw[line width= 0.07 cm, blue] (312:1) .. controls (190:0.7) .. (205:1) ;
	\draw[line width= 0.07 cm, blue] (0:1) .. controls (20:1.8) and (90:2) .. (138:1.3) ;
	\draw[line width= 0.07 cm, blue] (312:1) .. controls (320:2.2) and (220:1.5) .. (250:2) ;
	\draw[line width= 0.07 cm, blue] (250:2) .. controls (260:2.7) and (285:2) .. (307:3) ;
	\draw[line width= 0.07 cm, blue] (307:3) .. controls (310:3.7) and (280:3.7) .. (250:3.4) ;
	\draw[line width= 0.07 cm, blue] (250:3.4) .. controls (230:3.7) and (220:3.5) .. (200:3.4) ;
	\draw[line width= 0.07 cm, blue] (200:3.4) .. controls (190:3.4) and (182:3.3) .. (182:3) ;
	\draw[line width= 0.07 cm, blue] (85:3) .. controls (80:2.4) and (20:3) .. (350:3) ;
	\draw[line width= 0.07 cm, blue] (85:3) .. controls (80:4) and (20:5) .. (350:3.4) ;
	\draw[line width= 0.07 cm, blue] (350:3.4) .. controls (340:3.3) and (340:3.2) .. (350:3) ;

	\draw[red, line width=0.035cm] (0,0) circle [radius=1 cm];
	\draw[red, line width=0.035cm] (0,0) circle [radius=2 cm];
	\draw[red, line width=0.035cm] (0,0) circle [radius=3 cm];
	
\end{tikzpicture}}

\newcommand{\Gexamplecapped}{
	\begin{tikzpicture}
    \fill[opacity=0.25, lightblue]  (0:3)--(9:4)--(18:3)--(44:4)--(70:3)--(90:4)--(110:3)--(120:4)--(130:3)--(145:4)--(160:3)--(177:4)--(195:3)--(207:4)--(220:3)--(230:4)--(240:3)--(255:4)--(270:3)--(285:4)--(300:3)--(330:4)--(0:3);
    
    \draw[line width=0.035cm, dotted] (0:3)--(9:4)--(18:3)--(44:4)--(70:3)--(90:4)--(110:3)--(120:4)--(130:3)--(145:4)--(160:3)--(177:4)--(195:3)--(207:4)--(220:3)--(230:4)--(240:3)--(255:4)--(270:3)--(285:4)--(300:3)--(330:4)--(0:3);
    
    \node [fill=red,inner sep=1.25pt, label={[font=\small, rotate=0]right:$y$}] at (9:4) {};
    \node [fill=red,inner sep=1.25pt] at (44:4) {};
    \node [fill=red,inner sep=1.25pt] at (90:4) {};
    \node [fill=red,inner sep=1.25pt] at (120:4) {};
    \node [fill=red,inner sep=1.25pt] at (145:4) {};
    \node [fill=red,inner sep=1.25pt] at (177:4) {};
    \node [fill=red,inner sep=1.25pt] at (207:4) {};
    \node [fill=red,inner sep=1.25pt] at (230:4) {};
    \node [fill=red,inner sep=1.25pt] at (255:4) {};
    \node [fill=red,inner sep=1.25pt] at (285:4) {};
    \node [fill=red,inner sep=1.25pt] at (330:4) {};
    
	\draw[line width = 0.03 cm, orange,opacity=0.99] (0,0) -- (90:1);
	\draw[line width=0.035cm] (0,0) -- (30:1 cm);
	\draw[line width=0.035cm] (0,0) -- (220:1);
	\draw[line width=0.035cm] (0,0) -- (300:1);
	\draw[line width=0.035cm] (0,0) -- (170:1);
	
	\draw[line width=0.035cm] (90:1) -- (45:2);
	\draw[line width=0.035cm] (30:1) -- (45:2);
	\draw[line width=0.035cm] (90:1) -- (120:2);
	\draw[line width=0.035cm] (170:1) -- (120:2);
	\draw[line width=0.035cm] (170:1) -- (160:2);
	\draw[line width=0.035cm] (170:1) -- (190:2);
	\draw[line width=0.035cm] (30:1) -- (330:2);
	\draw[line width=0.035cm] (220:1) -- (190:2);
	\draw[line width=0.035cm] (220:1) -- (260:2);
	\draw[line width=0.035cm] (300:1) -- (260:2);
	\draw[line width=0.035cm] (300:1) -- (330:2);
	\draw[line width=0.035cm] (300:1) -- (290:2);
    
	\draw[line width=0.035cm] (45:2) -- (0:3);
	\draw[line width=0.035cm] (45:2) -- (18:3);
	\draw[line width=0.035cm] (45:2) -- (70:3);
	\draw[line width=0.035cm] (120:2) -- (70:3);
	\draw[line width=0.035cm] (120:2) -- (110:3);
	\draw[line width=0.035cm] (120:2) -- (130:3);
	\draw[line width=0.035cm] (160:2)-- (130:3);
	\draw[line width=0.035cm] (160:2)-- (160:3);
	\draw[line width=0.035cm] (160:2)-- (195:3);
	\draw[line width=0.035cm] (190:2) -- (195:3);
	\draw[line width=0.035cm] (190:2) -- (220:3);
	\draw[line width=0.035cm] (190:2) -- (240:3);
	\draw[line width=0.035cm] (260:2)--(240:3);
	\draw[line width=0.035cm] (260:2)--(270:3);
	\draw[line width=0.035cm] (260:2)--(300:3);
	\draw[line width=0.035cm] (330:2)--(0:3);
	\draw[line width=0.035cm] (330:2)--(300:3);
	\draw[line width=0.035cm] (290:2)--(300:3);

	\draw[red, line width=0.035cm] (0,0) circle [radius=1 cm];
	\draw[red, line width=0.035cm] (0,0) circle [radius=2 cm];
	\draw[red, line width=0.035cm] (0,0) circle [radius=3 cm];
	
	\node [fill=black,inner sep=1.25pt,label=right:$x$] at (0,0) {};
	\node [fill=black,inner sep=1.25pt,label={[font=\small,rotate=0]above:$v_1$}] at (90:1){};
	\node [fill=black,inner sep=1.25pt,label={[label distance=-1.5mm,font=\small,rotate=0]north east:$v_2$}] at (45:2){};
	\node [fill=black,inner sep=1.25pt,label={[label distance=1mm,font=\small,rotate=0]west:$v_3$}] at (0:3){};
\end{tikzpicture}}

    \newcommand{\dilutexample}{\begin{tikzpicture}
	
	\fill[opacity=0.25, lightblue] (0,0) circle [radius=3 cm];
    
	\draw[line width = 0.03 cm, orange,opacity=0.99] (0,0) -- (90:1);
	\draw[line width=0.035cm] (0,0) -- (30:1 cm);
	\draw[line width=0.035cm] (0,0) -- (220:1);
	\draw[line width=0.035cm] (0,0) -- (300:1);
	\draw[line width=0.035cm] (0,0) -- (170:1);
	
	\draw[line width=0.035cm] (90:1) -- (45:2);
	\draw[line width=0.035cm] (30:1) -- (45:2);
	\draw[line width=0.035cm] (90:1) -- (120:2);
	\draw[line width=0.035cm] (170:1) -- (120:2);
	\draw[line width=0.035cm] (170:1) -- (160:2);
	\draw[line width=0.035cm] (170:1) -- (190:2);
	\draw[line width=0.035cm] (30:1) -- (330:2);
	\draw[line width=0.035cm] (220:1) -- (190:2);
	\draw[line width=0.035cm] (220:1) -- (260:2);
	\draw[line width=0.035cm] (300:1) -- (260:2);
	\draw[line width=0.035cm] (300:1) -- (330:2);
	\draw[line width=0.035cm] (300:1) -- (290:2);
   
	\draw[line width=0.035cm] (45:2) -- (0:3);
	\draw[line width=0.035cm] (45:2) -- (18:3);
	\draw[line width=0.035cm] (45:2) -- (70:3);
	\draw[line width=0.035cm] (120:2) -- (70:3);
	\draw[line width=0.035cm] (120:2) -- (110:3);
	\draw[line width=0.035cm] (120:2) -- (130:3);
	\draw[line width=0.035cm] (160:2)-- (130:3);
	\draw[line width=0.035cm] (160:2)-- (160:3);
	\draw[line width=0.035cm] (160:2)-- (195:3);
	\draw[line width=0.035cm] (190:2) -- (195:3);
	\draw[line width=0.035cm] (190:2) -- (220:3);
	\draw[line width=0.035cm] (190:2) -- (240:3);
	\draw[line width=0.035cm] (260:2)--(240:3);
	\draw[line width=0.035cm] (260:2)--(270:3);
	\draw[line width=0.035cm] (260:2)--(300:3);
	\draw[line width=0.035cm] (330:2)--(0:3);
	\draw[line width=0.035cm] (330:2)--(300:3);
	\draw[line width=0.035cm] (290:2)--(300:3);
	
	\draw[line width= 0.07 cm, blue] (183:1) .. controls (170:0.7) and (140:0.5) .. (130:1) ;
	\draw[line width= 0.07 cm, blue] (130:1) .. controls (130:1.75) and (75:2) .. (82:3) ;
	
	\draw[line width= 0.07 cm, blue] (183:1) .. controls (185:1.4) and (150:1.8) .. (145:2) ;
	\draw[line width= 0.07 cm, blue] (145:2) .. controls (140:2.5) and (160:2.5) .. (167:3) ;
	
	\draw[line width= 0.07 cm, blue] (315:2) .. controls (320:2.5) and (350:2.5) .. (350:2) ;
	\draw[line width= 0.07 cm, blue] (315:2) .. controls (320:1.25) and (350:1.7) .. (350:2) ;
	
	\draw[line width= 0.07 cm, blue] (230:3) .. controls (260:1.2).. (292:3) ;

	\draw[red][line width=0.035cm]  (0,0) circle [radius=1 cm];
	\draw[red][line width=0.035cm]  (0,0) circle [radius=2 cm];
	\draw[red][line width=0.035cm]  (0,0) circle [radius=3 cm];
	
	\end{tikzpicture}}

\newcommand{\dilutexamplecapped}{\begin{tikzpicture}
    \fill[opacity=0.25, lightblue]  (0:3)--(9:4)--(18:3)--(44:4)--(70:3)--(90:4)--(110:3)--(120:4)--(130:3)--(145:4)--(160:3)--(177:4)--(195:3)--(207:4)--(220:3)--(230:4)--(240:3)--(255:4)--(270:3)--(285:4)--(300:3)--(330:4)--(0:3);
    
    \draw[line width=0.035cm, dotted] (0:3)--(9:4)--(18:3)--(44:4)--(70:3)--(90:4)--(110:3)--(120:4)--(130:3)--(145:4)--(160:3)--(177:4)--(195:3)--(207:4)--(220:3)--(230:4)--(240:3)--(255:4)--(270:3)--(285:4)--(300:3)--(330:4)--(0:3);
    
    \node [fill=red,inner sep=1.25pt] at (9:4) {};
    \node [fill=red,inner sep=1.25pt] at (44:4) {};
    \node [fill=red,inner sep=1.25pt] at (90:4) {};
    \node [fill=red,inner sep=1.25pt] at (120:4) {};
    \node [fill=red,inner sep=1.25pt] at (145:4) {};
    \node [fill=red,inner sep=1.25pt] at (177:4) {};
    \node [fill=red,inner sep=1.25pt] at (207:4) {};
    \node [fill=red,inner sep=1.25pt] at (230:4) {};
    \node [fill=red,inner sep=1.25pt] at (255:4) {};
    \node [fill=red,inner sep=1.25pt] at (285:4) {};
    \node [fill=red,inner sep=1.25pt] at (330:4) {};

	\draw[line width = 0.03 cm, orange,opacity=0.99] (0,0) -- (90:1);
	\draw[line width=0.035cm] (0,0) -- (30:1 cm);
	\draw[line width=0.035cm] (0,0) -- (220:1);
	\draw[line width=0.035cm] (0,0) -- (300:1);
	\draw[line width=0.035cm] (0,0) -- (170:1);
	
	\draw[line width=0.035cm] (90:1) -- (45:2);
	\draw[line width=0.035cm] (30:1) -- (45:2);
	\draw[line width=0.035cm] (90:1) -- (120:2);
	\draw[line width=0.035cm] (170:1) -- (120:2);
	\draw[line width=0.035cm] (170:1) -- (160:2);
	\draw[line width=0.035cm] (170:1) -- (190:2);
	\draw[line width=0.035cm] (30:1) -- (330:2);
	\draw[line width=0.035cm] (220:1) -- (190:2);
	\draw[line width=0.035cm] (220:1) -- (260:2);
	\draw[line width=0.035cm] (300:1) -- (260:2);
	\draw[line width=0.035cm] (300:1) -- (330:2);
	\draw[line width=0.035cm] (300:1) -- (290:2);
    
	\draw[line width=0.035cm] (45:2) -- (0:3);
	\draw[line width=0.035cm] (45:2) -- (18:3);
	\draw[line width=0.035cm] (45:2) -- (70:3);
	\draw[line width=0.035cm] (120:2) -- (70:3);
	\draw[line width=0.035cm] (120:2) -- (110:3);
	\draw[line width=0.035cm] (120:2) -- (130:3);
	\draw[line width=0.035cm] (160:2)-- (130:3);
	\draw[line width=0.035cm] (160:2)-- (160:3);
	\draw[line width=0.035cm] (160:2)-- (195:3);
	\draw[line width=0.035cm] (190:2) -- (195:3);
	\draw[line width=0.035cm] (190:2) -- (220:3);
	\draw[line width=0.035cm] (190:2) -- (240:3);
	\draw[line width=0.035cm] (260:2)--(240:3);
	\draw[line width=0.035cm] (260:2)--(270:3);
	\draw[line width=0.035cm] (260:2)--(300:3);
	\draw[line width=0.035cm] (330:2)--(0:3);
	\draw[line width=0.035cm] (330:2)--(300:3);
	\draw[line width=0.035cm] (290:2)--(300:3);

	\draw[line width= 0.07 cm, blue] (183:1) .. controls (170:0.7) and (140:0.5) .. (130:1) ;
	\draw[line width= 0.07 cm, blue] (130:1) .. controls (130:1.75) and (75:2) .. (87:3) ;
	\draw[line width= 0.07 cm, blue] (183:1) .. controls (185:1.4) and (150:1.8) .. (145:2) ;
	\draw[line width= 0.07 cm, blue] (145:2) .. controls (140:2.5) and (160:2.5) .. (167:3) ;
	
	\draw[line width= 0.07 cm, blue] (315:2) .. controls (320:2.5) and (350:2.5) .. (350:2) ;
	\draw[line width= 0.07 cm, blue] (315:2) .. controls (320:1.25) and (350:1.7) .. (350:2) ;
	\draw[line width= 0.07 cm, blue] (230:3) .. controls (260:1.2).. (292:3) ;
	
	\draw[line width= 0.07 cm, blue] (167:3) .. controls (170:3.5) and (140:3.75) .. (110:3.5) ;
	
	\draw[line width= 0.07 cm, blue] (110:3.5) .. controls (100:3.65) and (90:3.45) .. (87:3) ;
	
	\draw[line width= 0.07 cm, blue] (230:3) .. controls (225:3.75) and (270:4) .. (285:3.45) ;
	\draw[line width= 0.07 cm, blue] (285:3.45) .. controls (293:3.3) .. (292:3) ;

	\draw[red][line width=0.035cm]  (0,0) circle [radius=1 cm];
	\draw[red][line width=0.035cm]  (0,0) circle [radius=2 cm];
	\draw[red][line width=0.035cm]  (0,0) circle [radius=3 cm];
	
\end{tikzpicture}
}

\newcommand{\treexampleDashed}{
\begin{tikzpicture}
\fill[opacity=0.25,lightblue] (0,0) circle [radius=3 cm];
\draw[line width = 0.035 cm, orange,opacity=0.99][densely dotted] (122:0.6)--(0,0);
\draw[line width=0.03cm,orange,opacity=0.99] (0,0) -- (90:1); 
\draw[line width=0.03cm] (0,0) -- (30:1 cm);
\draw[line width=0.03cm] (0,0) -- (220:1);
\draw[line width=0.03cm] (0,0) -- (300:1);
\draw[line width=0.03cm] (0,0) -- (170:1);
\draw[line width=0.005cm,dashed] (90:1) -- (45:2);
\draw[line width=0.03cm] (30:1) -- (45:2);
\draw[line width=0.03cm] (90:1) -- (120:2);
\draw[line width=0.005cm,dashed] (170:1) -- (120:2);
\draw[line width=0.03cm] (170:1) -- (160:2);
\draw[line width=0.03cm] (170:1) -- (190:2);
\draw[line width=0.005cm,dashed] (30:1) -- (330:2);
\draw[line width=0.005cm,dashed] (220:1) -- (190:2);
\draw[line width=0.03cm] (220:1) -- (260:2);
\draw[line width=0.005cm,dashed] (300:1) -- (260:2);
\draw[line width=0.03cm] (300:1) -- (330:2);
\draw[line width=0.03cm] (300:1) -- (290:2);
\draw[line width=0.005cm,dashed] (45:2) -- (0:3);
\draw[line width=0.03cm] (45:2) -- (18:3);
\draw[line width=0.03cm] (45:2) -- (70:3);
\draw[line width=0.005cm,dashed] (120:2) -- (70:3);
\draw[line width=0.03cm] (120:2) -- (110:3);
\draw[line width=0.03cm] (120:2) -- (130:3);
\draw[line width=0.005cm,dashed] (160:2)-- (130:3);
\draw[line width=0.03cm] (160:2)-- (160:3);
\draw[line width=0.03cm] (160:2)-- (195:3);
\draw[line width=0.005cm,dashed] (190:2) -- (195:3);
\draw[line width=0.03cm] (190:2) -- (220:3);
\draw[line width=0.03cm] (190:2) -- (240:3);
\draw[line width=0.005cm,dashed] (260:2)--(240:3);
\draw[line width=0.03cm] (260:2)--(270:3);
\draw[line width=0.03cm] (260:2)--(300:3);
\draw[line width=0.03cm] (330:2)--(0:3);
\draw[line width=0.005cm,dashed] (330:2)--(300:3);
\draw[line width=0.005cm,dashed] (290:2)--(300:3);
\node [fill=black,inner sep=1pt] at (0:3.2) {};
\node [fill=black,inner sep=1pt] at (0:3.4) {};
\node [fill=black,inner sep=1pt]  at (0:3.6)  {};
\node [fill=black,inner sep=1pt] at (90:3.2) {};
\node [fill=black,inner sep=1pt] at (90:3.4) {};
\node [fill=black,inner sep=1pt] at (90:3.6) {};
\node [fill=black,inner sep=1pt] at (180:3.2) {};
\node [fill=black,inner sep=1pt] at (180:3.4) {};
\node [fill=black,inner sep=1pt] at (180:3.6) {};
\node [fill=black,inner sep=1pt] at (270:3.2) {};
\node [fill=black,inner sep=1pt] at (270:3.4) {};
\node [fill=black,inner sep=1pt] at (270:3.6) {};
\draw[red, line width=0.005cm,dashed] (0,0) circle [radius=1 cm];
\draw[red, line width=0.005cm,dashed] (0,0) circle [radius=2 cm];
\draw[red, line width=0.005cm,dashed] (0,0) circle [radius=3 cm];
\node [fill=black,inner sep=1.25pt,label=right:$x$] at (0,0) {};
\node [fill=black,inner sep=1.25pt,label={[font=\small,rotate=0]above:$v_1$}] at (90:1){};
\node [fill=black,inner sep=1.25pt,label={[font=\scriptsize, black, rotate=0]left:$x_0$}] at (122:0.6){};
\node [fill=black,inner sep=1.25pt,label={[label distance=-1.5mm,font=\small,rotate=0]north east:$v_2$}] at (45:2){};
	\node [fill=black,inner sep=1.25pt,label={[label distance=1mm,font=\small,rotate=0]west:$v_3$}] at (0:3){};
\end{tikzpicture}}

\newcommand{\flipbosa}{\begin{tikzpicture}[scale=0.5]
    \fill[opacity=0.25, lightblue] (-1,-1) -- (-1,1) -- (1,1) -- (1,-1) -- (-1,-1);
	\draw[line width=0.035cm] (-1,-1) -- (-1,1);
	\draw[line width=0.035cm] (1,-1) -- (1,1);
	\draw[line width=0.035cm] (-1,-1) -- (1,1);
	\draw[line width=0.035cm, red] (-1,1) -- (1,1);
	\draw[line width=0.035cm, red] (-1,-1) -- (1,-1);
	\end{tikzpicture} }

\newcommand{\flipbosb}{\begin{tikzpicture}[scale=0.5]
	\fill[opacity=0.25, lightblue] (-1,-1) -- (-1,1) -- (1,1) -- (1,-1) -- (-1,-1);
	\draw[line width=0.035cm] (-1,-1) -- (-1,1);
	\draw[line width=0.035cm] (1,-1) -- (1,1);
	\draw[line width=0.035cm] (1,-1) -- (-1,1);
	\draw[line width=0.035cm, red] (-1,1) -- (1,1);
	\draw[line width=0.035cm, red] (-1,-1) -- (1,-1);
	\end{tikzpicture} }

\newcommand{\flipa}{\begin{tikzpicture}[scale=0.5]
	\fill[opacity=0.25, lightblue] (-1,-1) -- (-1,1) -- (1,1) -- (1,-1) -- (-1,-1);
	\draw[line width=0.07 cm, blue] (-0.3,-1) arc [start angle=0, end angle=90, radius=7mm];
	\draw[line width=0.07 cm, blue] (1,-0.3) arc [start angle=90, end angle=180, radius=7mm];
	\draw[line width=0.035cm] (-1,-1) -- (-1,1);
	\draw[line width=0.035cm] (1,-1) -- (1,1);
	\draw[line width=0.035cm] (-1,-1) -- (1,1);
	\draw[line width=0.035cm, red] (-1,1) -- (1,1);
	\draw[line width=0.035cm, red] (-1,-1) -- (1,-1);
	\end{tikzpicture} }

\newcommand{\flipb}{\begin{tikzpicture}[scale=0.5]
    \fill[opacity=0.25, lightblue] (-1,-1) -- (-1,1) -- (1,1) -- (1,-1) -- (-1,-1);
	\draw[line width=0.07 cm, blue] (-0.3,-1) arc [start angle=0, end angle=90, radius=7mm];
	\draw[line width=0.07 cm, blue] (1,-0.3) arc [start angle=90, end angle=180, radius=7mm];
	\draw[line width=0.035cm] (-1,-1) -- (-1,1);
	\draw[line width=0.035cm] (1,-1) -- (1,1);
	\draw[line width=0.035cm] (1,-1) -- (-1,1);
	\draw[line width=0.035cm, red] (-1,1) -- (1,1);
	\draw[line width=0.035cm, red] (-1,-1) -- (1,-1);
	\end{tikzpicture} }

\newcommand{\flipc}{\begin{tikzpicture}[scale=0.5]
    \fill[opacity=0.25, lightblue] (-1,-1) -- (-1,1) -- (1,1) -- (1,-1) -- (-1,-1);
	\draw[line width=0.07 cm, blue] (-0.3,1) arc [start angle=0, end angle=-90, radius=7mm];
	\draw[line width=0.07 cm, blue] (0.3,1) arc [start angle=180, end angle=270, radius=7mm];
	\draw[line width=0.035cm] (-1,-1) -- (-1,1);
	\draw[line width=0.035cm] (1,-1) -- (1,1);
	\draw[line width=0.035cm] (-1,-1) -- (1,1);
	\draw[line width=0.035cm, red] (-1,1) -- (1,1);
	\draw[line width=0.035cm, red] (-1,-1) -- (1,-1);
	\end{tikzpicture} }

\newcommand{\flipd}{\begin{tikzpicture}[scale=0.5]
    \fill[opacity=0.25, lightblue] (-1,-1) -- (-1,1) -- (1,1) -- (1,-1) -- (-1,-1);
	\draw[line width=0.07 cm, blue] (-0.3,1) arc [start angle=0, end angle=-90, radius=7mm];
	\draw[line width=0.07 cm, blue] (0.3,1) arc [start angle=180, end angle=270, radius=7mm];
	\draw[line width=0.035cm] (-1,-1) -- (-1,1);
	\draw[line width=0.035cm] (1,-1) -- (1,1);
	\draw[line width=0.035cm] (1,-1) -- (-1,1);
	\draw[line width=0.035cm, red] (-1,1) -- (1,1);
	\draw[line width=0.035cm, red] (-1,-1) -- (1,-1);
	\end{tikzpicture} }

\newcommand{\flipe}{\begin{tikzpicture}[scale=0.5]
    \fill[opacity=0.25, lightblue] (-1,-1) -- (-1,1) -- (1,1) -- (1,-1) -- (-1,-1);
	\draw[line width=0.07 cm, blue] (-1,0) .. controls (-0.5,0.5) and (0.5,-0.5) .. (1,0) ;
	\draw[line width=0.035cm] (-1,-1) -- (-1,1);
	\draw[line width=0.035cm] (1,-1) -- (1,1);
	\draw[line width=0.035cm] (-1,-1) -- (1,1);
	\draw[line width=0.035cm, red] (-1,1) -- (1,1);
	\draw[line width=0.035cm, red] (-1,-1) -- (1,-1);
	\end{tikzpicture} }

\newcommand{\flipf}{\begin{tikzpicture}[scale=0.5]
    \fill[opacity=0.25, lightblue] (-1,-1) -- (-1,1) -- (1,1) -- (1,-1) -- (-1,-1);
	\draw[line width=0.07 cm, blue] (-1,0) .. controls (-0.5,-0.5) and (0.5,0.5) .. (1,0) ;
	\draw[line width=0.035cm] (-1,-1) -- (-1,1);
	\draw[line width=0.035cm] (1,-1) -- (1,1);
	\draw[line width=0.035cm] (1,-1) -- (-1,1);
	\draw[line width=0.035cm, red] (-1,1) -- (1,1);
	\draw[line width=0.035cm, red] (-1,-1) -- (1,-1);
	\end{tikzpicture} }

\newcommand{\flipg}{\begin{tikzpicture}[scale=0.5]
    \fill[opacity=0.25, lightblue] (-1,-1) -- (-1,1) -- (1,1) -- (1,-1) -- (-1,-1);
	\draw[line width=0.07 cm, blue] (-0.3,-1) arc [start angle=0, end angle=90, radius=7mm];
	\draw[line width=0.035cm] (-1,-1) -- (-1,1);
	\draw[line width=0.035cm] (1,-1) -- (1,1);
	\draw[line width=0.035cm] (-1,-1) -- (1,1);
	\draw[line width=0.035cm, red] (-1,1) -- (1,1);
	\draw[line width=0.035cm, red] (-1,-1) -- (1,-1);
	\end{tikzpicture} }

\newcommand{\fliph}{\begin{tikzpicture}[scale=0.5]
    \fill[opacity=0.25, lightblue] (-1,-1) -- (-1,1) -- (1,1) -- (1,-1) -- (-1,-1);
	\draw[line width=0.07 cm, blue] (-0.3,-1) arc [start angle=0, end angle=90, radius=7mm];
	\draw[line width=0.035cm] (-1,-1) -- (-1,1);
	\draw[line width=0.035cm] (1,-1) -- (1,1);
	\draw[line width=0.035cm] (1,-1) -- (-1,1);
	\draw[line width=0.035cm, red] (-1,1) -- (1,1);
	\draw[line width=0.035cm, red] (-1,-1) -- (1,-1);
	\end{tikzpicture} }

\newcommand{\flipi}{\begin{tikzpicture}[scale=0.5]
    \fill[opacity=0.25, lightblue] (-1,-1) -- (-1,1) -- (1,1) -- (1,-1) -- (-1,-1);
	\draw[line width=0.07 cm, blue] (-0.3,1) arc [start angle=0, end angle=-90, radius=7mm];
	\draw[line width=0.035cm] (-1,-1) -- (-1,1);
	\draw[line width=0.035cm] (1,-1) -- (1,1);
	\draw[line width=0.035cm] (-1,-1) -- (1,1);
	\draw[line width=0.035cm, red] (-1,1) -- (1,1);
	\draw[line width=0.035cm, red] (-1,-1) -- (1,-1);
	\end{tikzpicture} }

\newcommand{\flipj}{\begin{tikzpicture}[scale=0.5]
	\fill[opacity=0.25, lightblue] (-1,-1) -- (-1,1) -- (1,1) -- (1,-1) -- (-1,-1);
	\draw[line width=0.07 cm, blue] (-0.3,1) arc [start angle=0, end angle=-90, radius=7mm];
	\draw[line width=0.035cm] (-1,-1) -- (-1,1);
	\draw[line width=0.035cm] (1,-1) -- (1,1);
	\draw[line width=0.035cm] (1,-1) -- (-1,1);
	\draw[line width=0.035cm, red] (-1,1) -- (1,1);
	\draw[line width=0.035cm, red] (-1,-1) -- (1,-1);
	\end{tikzpicture} }

\newcommand{\flipii}{\begin{tikzpicture}[scale=0.5]
    \fill[opacity=0.25, lightblue] (-1,-1) -- (-1,1) -- (1,1) -- (1,-1) -- (-1,-1);
	\draw[line width=0.07 cm, blue] (0.3,1) arc [start angle=180, end angle=270, radius=7mm];
	\draw[line width=0.035cm] (-1,-1) -- (-1,1);
	\draw[line width=0.035cm] (1,-1) -- (1,1);
	\draw[line width=0.035cm] (-1,-1) -- (1,1);
	\draw[line width=0.035cm, red] (-1,1) -- (1,1);
	\draw[line width=0.035cm, red] (-1,-1) -- (1,-1);
	\end{tikzpicture} }

\newcommand{\flipjj}{\begin{tikzpicture}[scale=0.5]
    \fill[opacity=0.25, lightblue] (-1,-1) -- (-1,1) -- (1,1) -- (1,-1) -- (-1,-1);
	\draw[line width=0.07 cm, blue] (0.3,1) arc [start angle=180, end angle=270, radius=7mm];
	\draw[line width=0.035cm] (-1,-1) -- (-1,1);
	\draw[line width=0.035cm] (1,-1) -- (1,1);
	\draw[line width=0.035cm] (1,-1) -- (-1,1);
	\draw[line width=0.035cm, red] (-1,1) -- (1,1);
	\draw[line width=0.035cm, red] (-1,-1) -- (1,-1);
	\end{tikzpicture} }

\newcommand{\flipk}{\begin{tikzpicture}[scale=0.5]
    \fill[opacity=0.25, lightblue] (-1,-1) -- (-1,1) -- (1,1) -- (1,-1) -- (-1,-1);
	\draw[line width=0.07 cm, blue] (-0.3,1) arc [start angle=0, end angle=-90, radius=7mm];
	\draw[line width=0.07 cm, blue] (1,-0.3) arc [start angle=90, end angle=180, radius=7mm];
	\draw[line width=0.035cm] (-1,-1) -- (-1,1);
	\draw[line width=0.035cm] (1,-1) -- (1,1);
	\draw[line width=0.035cm] (-1,-1) -- (1,1);
	\draw[line width=0.035cm, red] (-1,1) -- (1,1);
	\draw[line width=0.035cm, red] (-1,-1) -- (1,-1);
	\end{tikzpicture} }

\newcommand{\flipl}{\begin{tikzpicture}[scale=0.5]
    \fill[opacity=0.25, lightblue] (-1,-1) -- (-1,1) -- (1,1) -- (1,-1) -- (-1,-1);
	\draw[line width=0.07 cm, blue] (0.3,1) arc [start angle=180, end angle=270, radius=7mm];
	\draw[line width=0.07 cm, blue] (-0.3,-1) arc [start angle=0, end angle=90, radius=7mm];
	\draw[line width=0.035cm] (-1,-1) -- (-1,1);
	\draw[line width=0.035cm] (1,-1) -- (1,1);
	\draw[line width=0.035cm] (1,-1) -- (-1,1);
	\draw[line width=0.035cm, red] (-1,1) -- (1,1);
	\draw[line width=0.035cm, red] (-1,-1) -- (1,-1);
	\end{tikzpicture} }

\newcommand{\flipm}{\begin{tikzpicture}[scale=0.5]
    \fill[opacity=0.25, lightblue] (-1,-1) -- (-1,1) -- (1,1) -- (1,-1) -- (-1,-1);
	\draw[line width=0.07 cm, blue] (0,1) .. controls (-0.5,0.5) and (0.5,-0.5) .. (0,-1) ;
	\draw[line width=0.035cm] (-1,-1) -- (-1,1);
	\draw[line width=0.035cm] (1,-1) -- (1,1);
	\draw[line width=0.035cm] (-1,-1) -- (1,1);
	\draw[line width=0.035cm, red] (-1,1) -- (1,1);
	\draw[line width=0.035cm, red] (-1,-1) -- (1,-1);
	\end{tikzpicture} }

\newcommand{\flipn}{\begin{tikzpicture}[scale=0.5]
    \fill[opacity=0.25, lightblue] (-1,-1) -- (-1,1) -- (1,1) -- (1,-1) -- (-1,-1);
	\draw[line width=0.07 cm, blue] (0,1) .. controls (0.5,0.5) and (-0.5,-0.5) .. (0,-1);
	\draw[line width=0.035cm] (-1,-1) -- (-1,1);
	\draw[line width=0.035cm] (1,-1) -- (1,1);
	\draw[line width=0.035cm] (1,-1) -- (-1,1);
	\draw[line width=0.035cm, red] (-1,1) -- (1,1);
	\draw[line width=0.035cm, red] (-1,-1) -- (1,-1);
	\end{tikzpicture} }

\newcommand{\flipo}{\begin{tikzpicture}[scale=0.5]
    \fill[opacity=0.25, lightblue] (-1,-1) -- (-1,1) -- (1,1) -- (1,-1) -- (-1,-1);
	\draw[line width=0.07 cm, blue] (1,-0.3) arc [start angle=90, end angle=180, radius=7mm];
	\draw[line width=0.035cm] (-1,-1) -- (-1,1);
	\draw[line width=0.035cm] (1,-1) -- (1,1);
	\draw[line width=0.035cm] (-1,-1) -- (1,1);
	\draw[line width=0.035cm, red] (-1,1) -- (1,1);
	\draw[line width=0.035cm, red] (-1,-1) -- (1,-1);
	\end{tikzpicture} }

\newcommand{\flipp}{\begin{tikzpicture}[scale=0.5]
    \fill[opacity=0.25, lightblue] (-1,-1) -- (-1,1) -- (1,1) -- (1,-1) -- (-1,-1);
	\draw[line width=0.07 cm, blue] (1,-0.3) arc [start angle=90, end angle=180, radius=7mm];
	\draw[line width=0.035cm] (-1,-1) -- (-1,1);
	\draw[line width=0.035cm] (1,-1) -- (1,1);
	\draw[line width=0.035cm] (1,-1) -- (-1,1);
	\draw[line width=0.035cm, red] (-1,1) -- (1,1);
	\draw[line width=0.035cm, red] (-1,-1) -- (1,-1);
	\end{tikzpicture} }

\newcommand{\flipq}{\begin{tikzpicture}[scale=0.5]
    \fill[opacity=0.25, lightblue] (-1,-1) -- (-1,1) -- (1,1) -- (1,-1) -- (-1,-1);
	\draw[line width=0.07 cm, blue] (0,1) .. controls (-0.5,0.5) and (0.5,-0.5) .. (0,-1) ;
	\draw[line width=0.07 cm, blue] (-0.3,1) arc [start angle=0, end angle=-90, radius=7mm];
	\draw[line width=0.07 cm, blue] (1,-0.3) arc [start angle=90, end angle=180, radius=7mm];
	\draw[line width=0.035cm] (-1,-1) -- (-1,1);
	\draw[line width=0.035cm] (1,-1) -- (1,1);
	\draw[line width=0.035cm] (-1,-1) -- (1,1);
	\draw[line width=0.035cm, red] (-1,1) -- (1,1);
	\draw[line width=0.035cm, red] (-1,-1) -- (1,-1);
	\end{tikzpicture} }

\newcommand{\flipr}{\begin{tikzpicture}[scale=0.5]
    \fill[opacity=0.25, lightblue] (-1,-1) -- (-1,1) -- (1,1) -- (1,-1) -- (-1,-1);
	\draw[line width=0.07 cm, blue] (0,1) .. controls (0.5,0.5) and (-0.5,-0.5) .. (0,-1) ;
	\draw[line width=0.07 cm, blue] (0.3,1) arc [start angle=180, end angle=270, radius=7mm];
	\draw[line width=0.07 cm, blue] (-0.3,-1) arc [start angle=0, end angle=90, radius=7mm];
	\draw[line width=0.035cm] (-1,-1) -- (-1,1);
	\draw[line width=0.035cm] (1,-1) -- (1,1);
	\draw[line width=0.035cm] (1,-1) -- (-1,1);
	\draw[line width=0.035cm, red] (-1,1) -- (1,1);
	\draw[line width=0.035cm, red] (-1,-1) -- (1,-1);
	\end{tikzpicture} }
	
\newcommand{\triangulationexamplebirAlt
	}{\begin{tikzpicture}

    \fill [line width=0.02331cm, altblue, opacity=0.5, yshift=-0.25cm] (-1.25,-.5*3.5-0.5)-- (0,-.5*3.5-1)--(1.25,-.5*3.5-0.5) -- (-1.25,-.5*3.5-0.5);
    
    \fill [altblue,opacity=0.5] plot [smooth, tension=1] coordinates { (-1.25,0)  (-1.25,-2.5)} -- plot [smooth, tension=1] coordinates { (1.25,-2.5) (1.25,0) } -- (-1.25, 0);

    \draw [decorate,decoration={brace,amplitude=5pt},xshift=-4pt,yshift=0pt, line width = 0.02cm] (-1.35,-2.5+0.15) -- (-1.35,0-0.15) node [black,midway,xshift=-0.6cm]  {\scalebox{0.65}{$m-1$}};
	
	\fill [altblue,opacity=0.5] (-1.25,0)-- (1.25,0) -- (0,0.5) -- (-1.25,0);

    \draw[line width=0.02331cm, black, yshift=0cm][densely dotted] (-1.25,0)-- (0,0.5);
    \draw[line width=0.02331cm, black, yshift=0cm] (-1.25+0.25,0)-- (0,0.5);
    \draw[line width=0.02331cm, black, yshift=0cm] (-1.25+0.5,0)-- (0,0.5);
    \draw[line width=0.02331cm, black, yshift=0cm][densely dotted] (1.25,0)-- (0,0.5);
    \draw[line width=0.02331cm, black, yshift=0cm] (1.25-0.25,0)-- (0,0.5);
    \draw[line width=0.02331cm, black, yshift=0cm] (1.25-0.5,0)-- (0,0.5);

	\draw[line width=0.02331cm, black, yshift=-0.25cm][densely dotted] (-1.25,-.5*3.5-0.5)-- (0,-.5*3.5-1);
	\draw[line width=0.02331cm, black, yshift=-0.25cm] (-1.25+0.2+0.3-0.2,-.5*3.5-0.5)-- (0,-.5*3.5-1);
	\draw[line width=0.02331cm, black, yshift=-0.25cm] (-1.25+0.2+0.3+0.1,-.5*3.5-0.5)-- (0,-.5*3.5-1);
	\draw[line width=0.02331cm, black, yshift=-0.25cm][densely dotted] (1.25,-.5*3.5-0.5)-- (0,-.5*3.5-1);
	\draw[line width=0.02331cm, black, yshift=-0.25cm] (1.25-0.25,-.5*3.5-0.5)-- (0,-.5*3.5-1);
	\draw[line width=0.02331cm, black, yshift=-0.25cm] (1.25-0.5,-.5*3.5-0.5)-- (0,-.5*3.5-1);

    \draw[line width=0.02331cm][densely dotted] plot [smooth, tension=2] coordinates { (1.25,0) (1.25,-2.5)};

    \draw[line width=0.02331cm, black, yshift=-0.25cm] (-1.25+0.2+0.3-0.2,-.5*3.5-0.5)-- (-1.25+0.2,-.5*3.5);
    \draw[line width=0.02331cm, black, yshift=-0.25cm] (-1.25+0.2+0.3-0.2,-.5*3.5-0.5) -- (-1.25,-.5*3.5);
    \draw[line width=0.02331cm, black, yshift=-0.25cm] (-1.25+0.2+0.3,-.5*3.5) -- (-1.25+0.2+0.3-0.2,-.5*3.5-0.5);
    \draw[line width=0.02331cm, black, yshift=-0.25cm]     (-1.25+0.2+0.3,-.5*3.5) -- (-1.25+0.2+0.3+0.1,-.5*3.5-0.5);

    \draw[line width=0.02331cm, black, yshift=-0.25cm] (-1.25+0.2,-.5*3.5) -- (-1.25,-.5*3.5+0.5);
    \draw[line width=0.02331cm, black, yshift=-0.25cm] (-1.25+0.2,-.5*3.5) -- (-1.25+0.2,-.5*3.5+0.5);
    \draw[line width=0.02331cm, black, yshift=-0.25cm] (-1.25+0.2,-.5*3.5) -- (-1.25+0.4,-.5*3.5+0.5);
    \draw[line width=0.02331cm, black, yshift=-0.25cm] (-1.25+0.2,-.5*3.5) -- (-1.25+0.6,-.5*3.5+0.5);
    \draw[line width=0.02331cm, black, yshift=-0.25cm] (-1.25+0.6,-.5*3.5+0.5) -- (-1.25+0.2+0.3,-.5*3.5);
    \draw[line width=0.02331cm, black, yshift=-0.25cm] (-1.25+0.2+0.3,-.5*3.5) -- (-1.25+0.2+0.3+0.4,-.5*3.5+0.5);

    \draw[line width=0.02331cm, black, yshift=-0.25cm] (-1.25+0.2,-.5*3.5+1) -- (-1.25,-.5*3.5+1+0.5);
    \draw[line width=0.02331cm, black, yshift=-0.25cm] (-1.25+0.4,-.5*3.5+1) -- (-1.25,-.5*3.5+1+0.5);
    \draw[line width=0.02331cm, black, yshift=-0.25cm] (-1.25+0.25,-.5*3.5+1+0.5) -- (-1.25+0.4,-.5*3.5+1) ;
    \draw[line width=0.02331cm, black, yshift=-0.25cm]     (-1.25+0.5,-.5*3.5+1+0.5) -- (-1.25+0.4,-.5*3.5+1) ;
    \draw[line width=0.02331cm, black, yshift=-0.25cm] (-1.25+0.5,-.5*3.5+1+0.5) -- (-1.25+0.65,-.5*3.5+1) ;
    \draw[line width=0.02331cm, black, yshift=-0.25cm] (-1.25+0.75,-.5*3.5+1+0.5) -- (-1.25+0.65,-.5*3.5+1) ;

    \draw[line width=0.02331cm, black, yshift=-0.25cm] (-1.25,-.5*3.5+2) -- (-1.25+0.25,-.5*3.5+1+0.5);
    \draw[line width=0.02331cm, black, yshift=-0.25cm] (-1.25+0.25,-.5*3.5+2) -- (-1.25+0.25,-.5*3.5+1+0.5);   
    \draw[line width=0.02331cm, black, yshift=-0.25cm] (-1.25+0.5,-.5*3.5+2) -- (-1.25+0.25,-.5*3.5+1+0.5);
    \draw[line width=0.02331cm, black, yshift=-0.25cm] (-1.25+0.5,-.5*3.5+2) -- (-1.25+0.5,-.5*3.5+1+0.5);
    \draw[line width=0.02331cm, black, yshift=-0.25cm] (-1.25+0.5,-.5*3.5+2) -- (-1.25+0.75,-.5*3.5+1+0.5);

    \draw[line width=0.02331cm, black, yshift=-0.25cm] (1.25,-.5*3.5-0.5) -- (1.25-0.25,-.5*3.5);
    \draw[line width=0.02331cm, black, yshift=-0.25cm] (1.25-0.25,-.5*3.5) -- (1.25-0.25,-.5*3.5-0.5);
    \draw[line width=0.02331cm, black, yshift=-0.25cm] (1.25-0.25,-.5*3.5) -- (1.25-0.5,-.5*3.5-0.5);
    \draw[line width=0.02331cm, black, yshift=-0.25cm] (1.25-0.5,-.5*3.5-0.5) -- (1.25-0.55,-.5*3.5);
    \draw[line width=0.02331cm, black, yshift=-0.25cm] (1.25-0.5,-.5*3.5-0.5) -- (1.25-0.55-0.3,-.5*3.5);

    \draw[line width=0.02331cm, black, yshift=-0.25cm] (1.25,-.5*3.5+0.5) -- (1.25-0.25,-.5*3.5);
    \draw[line width=0.02331cm, black, yshift=-0.25cm] (1.25,-.5*3.5+0.5) -- (1.25-0.55,-.5*3.5);
    \draw[line width=0.02331cm, black, yshift=-0.25cm] (1.25-0.3,-.5*3.5+0.5) -- (1.25-0.55,-.5*3.5);
    \draw[line width=0.02331cm, black, yshift=-0.25cm] (1.25-0.6,-.5*3.5+0.5) -- (1.25-0.55,-.5*3.5);
    \draw[line width=0.02331cm, black, yshift=-0.25cm] (1.25-0.6,-.5*3.5+0.5) -- (1.25-0.55-0.3,-.5*3.5);

    \draw[line width=0.02331cm, black, yshift=-0.25cm] (1.25,-.5*3.5+1) -- (1.25-0.2,-.5*3.5+1+0.5);
    \draw[line width=0.02331cm, black, yshift=-0.25cm] (1.25-0.25,-.5*3.5+1) -- (1.25-0.2,-.5*3.5+1+0.5);
    \draw[line width=0.02331cm, black, yshift=-0.25cm] (1.25-0.5,-.5*3.5+1) -- (1.25-0.2,-.5*3.5+1+0.5);
    \draw[line width=0.02331cm, black, yshift=-0.25cm] (1.25-0.5,-.5*3.5+1) -- (1.25-0.5,-.5*3.5+1+0.5);
    \draw[line width=0.02331cm, black, yshift=-0.25cm]     (1.25-0.5-0.25,-.5*3.5+1) -- (1.25-0.5,-.5*3.5+1+0.5);
    \draw[line width=0.02331cm, black, yshift=-0.75cm] (1.25-0.75,-.5*3.5+1+0.5) -- (1.25-0.75-0.1,-.5*3.5+1+1);

    \draw[line width=0.02331cm, black, yshift=-0.25cm]     (1.25-0.2,-.5*3.5+1+0.5) -- (1.25,-.5*3.5+1+1);
    \draw[line width=0.02331cm, black, yshift=-0.25cm] (1.25-0.2,-.5*3.5+1+0.5) -- (1.25-0.25,-.5*3.5+1+1);
    \draw[line width=0.02331cm, black, yshift=-0.25cm] (1.25-0.2,-.5*3.5+1+0.5) -- (1.25-0.5,-.5*3.5+1+1);
    \draw[line width=0.02331cm, black, yshift=-0.25cm] (1.25-0.5,-.5*3.5+1+0.5) -- (1.25-0.5,-.5*3.5+1+1);
    \draw[line width=0.02331cm, black, yshift=-0.25cm] (1.25-0.75-0.1,-.5*3.5+1+0.5) -- (1.25-0.5,-.5*3.5+1+1);

    \draw[line width=0.02331cm, red, yshift=-0.25cm] (-1.25,-.5*3.5) -- (-1.25+2.5,-.5*3.5);
    \draw[line width=0.02331cm, red, yshift=0.25cm] (-1.25,-.5*3.5) -- (-1.25+2.5,-.5*3.5);
    \draw[line width=0.02331cm, red, yshift=0.75cm] (-1.25,-.5*3.5) -- (-1.25+2.5,-.5*3.5);
    \draw[line width=0.02331cm, red, yshift=1.25cm] (-1.25,-.5*3.5) -- (-1.25+2.5,-.5*3.5);

    \draw[line width=0.02331cm, yshift=1.75cm, red] (-1.25,-.5*3.5) -- (-1.25+2.5,-.5*3.5);
    \draw[line width=0.02331cm, yshift=-0.75cm, red] (-1.25,-.5*3.5) -- (-1.25+2.5,-.5*3.5);

    \draw[line width=0.02331cm, black, yshift=-2.5cm, xshift=-1.25cm][densely dotted] (0,2.5) -- (0,0);
    
    \draw[line width=0.02331cm, dotted, yshift=1.95cm] (-0.15+0.025,-.5*3.5) -- (0.15-0.025,-.5*3.5);
    \draw[line width=0.02331cm, dotted, yshift=1.5cm] (-0.15+0.025,-.5*3.5) -- (0.15-0.025,-.5*3.5);
    \draw[line width=0.02331cm, dotted, yshift=1cm] (-0.15+0.025,-.5*3.5) -- (0.15-0.025,-.5*3.5);
    \draw[line width=0.02331cm, dotted, yshift=0cm] (-0.15+0.025,-.5*3.5) -- (0.15-0.025,-.5*3.5);
    \draw[line width=0.02331cm, dotted, yshift=-0.5cm] (-0.15+0.025,-.5*3.5) -- (0.15-0.025,-.5*3.5);
    \draw[line width=0.02331cm, dotted, yshift=-0.95cm] (-0.15+0.025,-.5*3.5) -- (0.15-0.025,-.5*3.5);
    
    \draw[line width=0.02331cm, black, yshift=-3.5cm, xshift=-0.15cm, xshift=-1.25cm][dotted] (0+0.35,2.5-0.1-0.025) -- (0+0.35,2.5-0.5+0.1+0.025);
    \draw[line width=0.02331cm, black, yshift=-3.5cm, xshift=-0.15cm, xshift=0.85cm][dotted] (0+0.35,2.5-0.1-0.025) -- (0+0.35,2.5-0.5+0.1+0.025);
	
	\end{tikzpicture}}

\newcommand{\triangulationexampleikiAlt
	}{\begin{tikzpicture}

	\fill [altblue, opacity=0.5] plot [smooth, tension=1] coordinates { (-1.25,0.5)  (-1.25,-2.5)} -- plot [smooth, tension=1] coordinates { (1.25,-2.5) (1.25,0.5) } -- (-1.25, 0.5);

	\draw [decorate,decoration={brace,amplitude=5pt},xshift=4pt,yshift=0pt, line width = 0.02cm] (1.35,0.5-0.15) -- (1.35,-2.5+0.15) node [black,midway,xshift=0.4cm]  {\scalebox{0.65}{$m$}};
	
    \draw[line width=0.02331cm, black, yshift=0cm] (-1.25+0.25,0)-- (-1.25,0.5);
    \draw[line width=0.02331cm, black, yshift=0cm] (-1.25+0.5,0)-- (-1.25,0.5);
    \draw[line width=0.02331cm, black, yshift=0cm] (1.25,0)-- (-1.25,0.5);
    \draw[line width=0.02331cm, black, yshift=0cm] (1.25-0.5,0)-- (-1.25,0.5);
    \draw[line width=0.02331cm, black, yshift=0cm] (1.25-1,0)-- (-1.25,0.5);
    
	\draw[line width=0.02331cm, black, yshift=0cm] (1.25,0)-- (1.25-0.25,0.5);
	\draw[line width=0.02331cm, black, yshift=0cm] (1.25,0)-- (1.25-0.5,0.5);
	\draw[line width=0.02331cm, black, yshift=0cm] (1.25,0)-- (-1.25+0.5,0.5);
	\draw[line width=0.02331cm, black, yshift=0cm] (1.25,0)-- (-1.25+1,0.5);

    \draw[line width=0.02331cm][densely dotted] plot [smooth, tension=2] coordinates { (1.25,0.5) (1.25,-2.5)};

    \draw[line width=0.02331cm, black, yshift=-0.25cm] (-1.25+0.2+0.3-0.2,-.5*3.5-0.5)-- (-1.25+0.2,-.5*3.5);
    \draw[line width=0.02331cm, black, yshift=-0.25cm] (-1.25+0.2+0.3-0.2,-.5*3.5-0.5) -- (-1.25,-.5*3.5);
    \draw[line width=0.02331cm, black, yshift=-0.25cm] (-1.25+0.2+0.3,-.5*3.5) -- (-1.25+0.2+0.3-0.2,-.5*3.5-0.5);
    \draw[line width=0.02331cm, black, yshift=-0.25cm]     (-1.25+0.2+0.3,-.5*3.5) -- (-1.25+0.2+0.3+0.1,-.5*3.5-0.5);

    \draw[line width=0.02331cm, black, yshift=-0.25cm] (-1.25+0.2,-.5*3.5) -- (-1.25,-.5*3.5+0.5);
    \draw[line width=0.02331cm, black, yshift=-0.25cm] (-1.25+0.2,-.5*3.5) -- (-1.25+0.2,-.5*3.5+0.5);
    \draw[line width=0.02331cm, black, yshift=-0.25cm] (-1.25+0.2,-.5*3.5) -- (-1.25+0.4,-.5*3.5+0.5);
    \draw[line width=0.02331cm, black, yshift=-0.25cm] (-1.25+0.2,-.5*3.5) -- (-1.25+0.6,-.5*3.5+0.5);
    \draw[line width=0.02331cm, black, yshift=-0.25cm] (-1.25+0.6,-.5*3.5+0.5) -- (-1.25+0.2+0.3,-.5*3.5);
    \draw[line width=0.02331cm, black, yshift=-0.25cm] (-1.25+0.2+0.3,-.5*3.5) -- (-1.25+0.2+0.3+0.4,-.5*3.5+0.5);

    \draw[line width=0.02331cm, black, yshift=-0.25cm] (-1.25+0.2,-.5*3.5+1) -- (-1.25,-.5*3.5+1+0.5);
    \draw[line width=0.02331cm, black, yshift=-0.25cm] (-1.25+0.4,-.5*3.5+1) -- (-1.25,-.5*3.5+1+0.5);
    \draw[line width=0.02331cm, black, yshift=-0.25cm] (-1.25+0.25,-.5*3.5+1+0.5) -- (-1.25+0.4,-.5*3.5+1) ;
    \draw[line width=0.02331cm, black, yshift=-0.25cm]     (-1.25+0.5,-.5*3.5+1+0.5) -- (-1.25+0.4,-.5*3.5+1) ;
    \draw[line width=0.02331cm, black, yshift=-0.25cm] (-1.25+0.5,-.5*3.5+1+0.5) -- (-1.25+0.65,-.5*3.5+1) ;
    \draw[line width=0.02331cm, black, yshift=-0.25cm] (-1.25+0.75,-.5*3.5+1+0.5) -- (-1.25+0.65,-.5*3.5+1) ;

    \draw[line width=0.02331cm, black, yshift=-0.25cm] (-1.25,-.5*3.5+2) -- (-1.25+0.25,-.5*3.5+1+0.5);
    \draw[line width=0.02331cm, black, yshift=-0.25cm] (-1.25+0.25,-.5*3.5+2) -- (-1.25+0.25,-.5*3.5+1+0.5);   
    \draw[line width=0.02331cm, black, yshift=-0.25cm] (-1.25+0.5,-.5*3.5+2) -- (-1.25+0.25,-.5*3.5+1+0.5);
    \draw[line width=0.02331cm, black, yshift=-0.25cm] (-1.25+0.5,-.5*3.5+2) -- (-1.25+0.5,-.5*3.5+1+0.5);
    \draw[line width=0.02331cm, black, yshift=-0.25cm] (-1.25+0.5,-.5*3.5+2) -- (-1.25+0.75,-.5*3.5+1+0.5);

    \draw[line width=0.02331cm, black, yshift=-0.25cm] (1.25,-.5*3.5-0.5) -- (1.25-0.25,-.5*3.5);
    \draw[line width=0.02331cm, black, yshift=-0.25cm] (1.25-0.25,-.5*3.5) -- (1.25-0.25,-.5*3.5-0.5);
    \draw[line width=0.02331cm, black, yshift=-0.25cm] (1.25-0.25,-.5*3.5) -- (1.25-0.5,-.5*3.5-0.5);
    \draw[line width=0.02331cm, black, yshift=-0.25cm] (1.25-0.5,-.5*3.5-0.5) -- (1.25-0.55,-.5*3.5);
    \draw[line width=0.02331cm, black, yshift=-0.25cm] (1.25-0.5,-.5*3.5-0.5) -- (1.25-0.55-0.3,-.5*3.5);

    \draw[line width=0.02331cm, black, yshift=-0.25cm] (1.25,-.5*3.5+0.5) -- (1.25-0.25,-.5*3.5);
    \draw[line width=0.02331cm, black, yshift=-0.25cm] (1.25,-.5*3.5+0.5) -- (1.25-0.55,-.5*3.5);
    \draw[line width=0.02331cm, black, yshift=-0.25cm] (1.25-0.3,-.5*3.5+0.5) -- (1.25-0.55,-.5*3.5);
    \draw[line width=0.02331cm, black, yshift=-0.25cm] (1.25-0.6,-.5*3.5+0.5) -- (1.25-0.55,-.5*3.5);
    \draw[line width=0.02331cm, black, yshift=-0.25cm] (1.25-0.6,-.5*3.5+0.5) -- (1.25-0.55-0.3,-.5*3.5);

    \draw[line width=0.02331cm, black, yshift=-0.25cm] (1.25,-.5*3.5+1) -- (1.25-0.2,-.5*3.5+1+0.5);
    \draw[line width=0.02331cm, black, yshift=-0.25cm] (1.25-0.25,-.5*3.5+1) -- (1.25-0.2,-.5*3.5+1+0.5);
    \draw[line width=0.02331cm, black, yshift=-0.25cm] (1.25-0.5,-.5*3.5+1) -- (1.25-0.2,-.5*3.5+1+0.5);
    \draw[line width=0.02331cm, black, yshift=-0.25cm] (1.25-0.5,-.5*3.5+1) -- (1.25-0.5,-.5*3.5+1+0.5);
    \draw[line width=0.02331cm, black, yshift=-0.25cm]     (1.25-0.5-0.25,-.5*3.5+1) -- (1.25-0.5,-.5*3.5+1+0.5);
    \draw[line width=0.02331cm, black, yshift=-0.75cm] (1.25-0.75,-.5*3.5+1+0.5) -- (1.25-0.75-0.1,-.5*3.5+1+1);

    \draw[line width=0.02331cm, black, yshift=-0.25cm]     (1.25-0.2,-.5*3.5+1+0.5) -- (1.25,-.5*3.5+1+1);
    \draw[line width=0.02331cm, black, yshift=-0.25cm] (1.25-0.2,-.5*3.5+1+0.5) -- (1.25-0.5,-.5*3.5+1+1);
    \draw[line width=0.02331cm, black, yshift=-0.25cm] (1.25-0.2,-.5*3.5+1+0.5) -- (1.25-1,-.5*3.5+1+1);
    \draw[line width=0.02331cm, black, yshift=-0.25cm] (1.25-0.5,-.5*3.5+1+0.5) -- (1.25-1,-.5*3.5+1+1);
    \draw[line width=0.02331cm, black, yshift=-0.25cm] (1.25-0.75-0.1,-.5*3.5+1+0.5) -- (1.25-1,-.5*3.5+1+1);

     \draw[line width=0.02331cm, black, yshift=-2.5cm, xshift=-1.25cm][densely dotted] (0,3) -- (0,0);
     
    \draw[line width=0.02331cm, dotted, yshift=1.95cm, xshift=-0.65cm] (-0.15+0.025,-.5*3.5) -- (0.15-0.025,-.5*3.5);
    \draw[line width=0.02331cm, dotted, yshift=2.05cm, xshift=0.7cm] (-0.15+0.025,-.5*3.5) -- (0.15-0.025,-.5*3.5);
    
    \draw[line width=0.02331cm, dotted, yshift=1.5cm] (-0.15+0.025,-.5*3.5) -- (0.15-0.025,-.5*3.5);
    \draw[line width=0.02331cm, dotted, yshift=1cm] (-0.15+0.025,-.5*3.5) -- (0.15-0.025,-.5*3.5);
    \draw[line width=0.02331cm, dotted, yshift=0cm] (-0.15+0.025,-.5*3.5) -- (0.15-0.025,-.5*3.5);
    \draw[line width=0.02331cm, dotted, yshift=-0.5cm] (-0.15+0.025,-.5*3.5) -- (0.15-0.025,-.5*3.5);
    
    \draw[line width=0.02331cm, black, yshift=-3.5cm, xshift=-0.15cm, xshift=-1.25cm][dotted] (0+0.35,2.5-0.1-0.025) -- (0+0.35,2.5-0.5+0.1+0.025);
    \draw[line width=0.02331cm, black, yshift=-3.5cm, xshift=-0.15cm, xshift=0.85cm][dotted] (0+0.35,2.5-0.1-0.025) -- (0+0.35,2.5-0.5+0.1+0.025);
	
	\draw[line width=0.02331cm, red, yshift=-0.75cm] (-1.25,-.5*3.5) -- (-1.25+2.5,-.5*3.5);
	\draw[line width=0.02331cm, red, yshift=-0.25cm] (-1.25,-.5*3.5) -- (-1.25+2.5,-.5*3.5);
    \draw[line width=0.02331cm, red, yshift=0.25cm] (-1.25,-.5*3.5) -- (-1.25+2.5,-.5*3.5);
    \draw[line width=0.02331cm, red, yshift=0.75cm] (-1.25,-.5*3.5) -- (-1.25+2.5,-.5*3.5);
    \draw[line width=0.02331cm, red, yshift=1.25cm] (-1.25,-.5*3.5) -- (-1.25+2.5,-.5*3.5);
    \draw[line width=0.02331cm, red, yshift=1.75cm] (-1.25,-.5*3.5) -- (-1.25+2.5,-.5*3.5);
    \draw[line width=0.02331cm, red, yshift=2.25cm] (-1.25,-.5*3.5) -- (-1.25+2.5,-.5*3.5);

    \end{tikzpicture}}

\newcommand{\planartreeexampleDashed}{\begin{tikzpicture}
    \draw [line width = 0.035 cm, orange,opacity=0.99][densely dotted] (0,0)--(0,1);
    \draw [line width = 0.035 cm, orange,opacity=0.99] (0,1)--(-0.75*2,2);
    \draw [line width = 0.035 cm] (0,1)--(-0.75*1,2);
    \draw [line width = 0.035 cm] (0,1)--(0,2);
    \draw [line width = 0.035 cm] (0,1)--(0.75*1,2);
    \draw [line width = 0.035 cm] (0,1)--(0.75*2,2);
    \draw [line width = 0.035 cm] (-0.75*2,2)--(-0.75*2-0.5,3);
    \draw [line width = 0.035 cm] (-0.75*1,2)--(-0.75*1-0.5,3);
    \draw [line width = 0.035 cm] (0,2)--(-0.25,3);
    \draw [line width = 0.035 cm] (0,2)--(0.25,3);
    \draw [line width = 0.035 cm] (0.75*1,2)--(0.75*1,3);
    \draw [line width = 0.035 cm] (0.75*2,2)--(0.75*2,3);
    \draw [line width = 0.035 cm] (0.75*2,2)--(0.75*2+0.75,3);
    \draw [line width = 0.035 cm] (-0.75*2-0.5,3)--(-0.75*2-0.5-0.25,4);
    \draw [line width = 0.035 cm] (-0.75*2-0.5,3)--(-0.75*2-0.5+0.25,4);
    \draw [line width = 0.035 cm] (-0.75*1-0.5,3)--(-0.75*1-0.25-0.5,4);
    \draw [line width = 0.035 cm] (-0.75*1-0.5,3)--(-0.75*1+0.25-0.5,4);
    \draw [line width = 0.035 cm] (-0.25,3)--(-0.25,4);
    \draw [line width = 0.035 cm] (0.75*1,3)--(0.75*1-0.25,4);
    \draw [line width = 0.035 cm] (0.75*1,3)--(0.75*1+0.25,4);
    \draw [line width = 0.035 cm] (0.75*2,3)--(0.75*2-0.25,4);
    \draw [line width = 0.035 cm] (0.75*2,3)--(0.75*2+0.25,4);
    \draw [line width = 0.035 cm] (0.75*2+0.75,3)--(0.75*2+0.75,4);
    \draw [line width = 0.035 cm] (0.75*2+0.75,3)--(0.75*2+0.75+0.5,4);
    
    \node [fill=black,inner sep=1.25pt,label=left:$x_0$] at (0,0) {};
    \node [fill=black,inner sep=1.25pt,label=left:$x$] at (0,1) {};
    \node [fill=black,inner sep=1.25pt,label={[font=\small, rotate=0]left:$v_1$}] at (-0.75*2,2) {};
    \node [fill=black,inner sep=1.25pt,label={[font=\small, rotate=0]left:$v_2$}] at (-0.75*1-0.5,3) {};
    \node [fill=black,inner sep=1.25pt,label={[font=\small, rotate=0]left:$v_3$}] at (-0.25,4) {};
    
    \node [fill=black,inner sep=1pt, yshift=1.25cm] at (0,3.2) {};
	\node [fill=black,inner sep=1pt, yshift=1.25cm] at (0,3.4) {};
	\node [fill=black,inner sep=1pt, yshift=1.25cm]  at (0,3.6)  {};
\end{tikzpicture}}

\newcommand{\standingTriDi}{\begin{tikzpicture}[scale=1]
	\draw[dashed, xshift=+6 cm, line width=0.0025cm, yshift=1cm] (5.85-0.6,0) -- (-8+0.5,0) node(xline)[left] {$S_{k+1}$};
	\draw[dashed, xshift=6 cm, line width=0.0025cm, yshift=1cm] (5.85-0.6,-1) -- (-8+0.5,-1) node(xline)[left] {$S_{k}$};
	 \fill[opacity=0.25, lightblue] (-1,0) -- (1,0) -- (0,1) -- (-1,0) ;
	\draw[line width=0.035cm] (1,0) -- (0,1) -- (-1,0) ;
	\draw[line width=0.035cm, red] (-1,0) -- (1,0);
	 \fill[opacity=0.25, lightblue, xshift=3.25cm] (-1,0) -- (1,0) -- (0,1) -- (-1,0) ;
	\draw[line width= 0.07 cm, blue, xshift=3.25cm] (-0.5,0.5) arc (225:315:0.7 and 0.5) ;
	\draw[line width= 0.035 cm, xshift=3.25cm] (1,0) -- (0,1) -- (-1,0) ;
	\draw[line width= 0.035 cm, red, xshift=3.25cm] (-1,0) -- (1,0);
	 \fill[opacity=0.25, lightblue, xshift=6.5cm] (-1,0) -- (1,0) -- (0,1) -- (-1,0) ;
	\draw[line width= 0.07 cm, blue, xshift=6.5cm] (-0.3,0) arc (0:45:0.7);
	\draw[xshift=6.5cm, line width=0.035cm] (1,0) -- (0,1) -- (-1,0) ;
	\draw[xshift=6.5cm, red, line width=0.035cm] (-1,0) -- (1,0);
	\fill[opacity=0.25, lightblue, xshift=9.75cm] (-1,0) -- (1,0) -- (0,1) -- (-1,0) ;
	\draw[xshift=9.75 cm, line width= 0.07 cm, blue] (0.3,0) arc (180:135:0.7);
	\draw[line width=0.035cm, xshift=9.75cm] (1,0) -- (0,1) -- (-1,0) ;
	\draw[line width=0.035cm, red, xshift=9.75cm] (-1,0) -- (1,0);
	\end{tikzpicture}}
	
\newcommand{\hangingTriDi}{\begin{tikzpicture}[scale=1]
	\draw[dashed, xshift=+6 cm, line width=0.0025cm] (5.85-0.6,0) -- (-8+0.5,0) node(xline)[left] {$S_{k+1}$};
	\draw[dashed, xshift=6 cm, line width=0.0025cm] (5.85-0.6,-1) -- (-8+0.5,-1) node(xline)[left] {$S_{k}$};
	\fill[opacity=0.25, lightblue] (-1,0) -- (1,0) -- (0,-1) -- (-1,0) ;
	\draw[line width= 0.035 cm] (1,0) -- (0,-1) -- (-1,0) ;
	\draw[red, line width= 0.035 cm] (-1,0) -- (1,0);
	\fill[opacity=0.25, lightblue, xshift=3.25cm] (-1,0) -- (1,0) -- (0,-1) -- (-1,0) ;
	\draw[line width= 0.07 cm, blue, xshift=3.25cm] (-0.5,-0.5) arc (225:315:0.7 and -0.5) ;
	\draw[line width= 0.035 cm, xshift=3.25 cm,] (1,0) -- (0,-1) -- (-1,0) ;
	\draw[red, line width= 0.035 cm, xshift=3.25 cm,] (-1,0) -- (1,0);
	\fill[opacity=0.25, lightblue, xshift=6.5cm] (-1,0) -- (1,0) -- (0,-1) -- (-1,0) ;
	\draw[line width= 0.07 cm, blue, xshift=6.5cm] (-0.3,0) arc (360:315:0.7);
	\draw[line width= 0.035 cm, xshift=6.5cm] (1,0) -- (0,-1) -- (-1,0) ; 
	\draw[red, line width= 0.035 cm, xshift=6.5cm] (-1,0) -- (1,0);
	\fill[opacity=0.25, lightblue, xshift=9.75cm] (-1,0) -- (1,0) -- (0,-1) -- (-1,0) ;
	\draw[xshift=9.75 cm, line width= 0.07 cm, blue] (0.3,0) arc (180:225:0.7);
	\draw[line width= 0.035 cm, xshift=9.75 cm,] (1,0) -- (0,-1) -- (-1,0) ; 
	\draw[red, line width= 0.035 cm, xshift=9.75 cm,] (-1,0) -- (1,0);
	\end{tikzpicture}}

\newcommand{\bothTriDe}{\begin{tikzpicture}[scale=1]
	\draw[dashed, xshift=+6 cm, line width=0.0025cm] (5.85-0.6,0) -- (-8+0.5,0) node(xline)[left] {$S_{k+1}$};
	\draw[dashed, xshift=6 cm, line width=0.0025cm] (5.85-0.6,-1) -- (-8+0.5,-1) node(xline)[left] {$S_{k}$};
	\fill[opacity=0.25, lightblue] (-1,0) -- (1,0) -- (0,-1) -- (-1,0) ;
	\draw[line width= 0.07 cm, blue] (-0.5,-0.5) arc (135:45:0.7 and 0.5) ;
	\draw[line width=0.035cm] (1,0) -- (0,-1) -- (-1,0) ;
	\draw[red, line width=0.035cm] (-1,0) -- (1,0);
	\fill[opacity=0.25, lightblue, xshift=3.25cm] (-1,0) -- (1,0) -- (0,-1) -- (-1,0) ;
	\draw[xshift=3.25 cm, line width= 0.07 cm, blue] (-0.3,0) arc (360:315:0.7);
	\draw[xshift=3.25 cm, line width= 0.07 cm, blue] (0.3,0) arc (180:225:0.7);
	\draw[xshift=3.25cm, line width=0.035cm] (1,0) -- (0,-1) -- (-1,0) ;
	\draw[red, xshift=3.25cm, line width=0.035cm] (-1,0) -- (1,0);
	\fill[opacity=0.25, lightblue, xshift=6.5cm, yshift=-1cm] (-1,0) -- (1,0) -- (0,1) -- (-1,0) ;
	\draw[line width= 0.07 cm, blue, xshift=6.5cm, yshift=-1cm] (-0.5,0.5) arc (225:315:0.7 and 0.5) ;
	\draw[line width=0.035cm, xshift=6.5cm, yshift=-1cm] (1,0) -- (0,1) -- (-1,0) ;
	\draw[line width=0.035cm, red, xshift=6.5cm, yshift=-1cm] (-1,0) -- (1,0);
	\fill[opacity=0.25, lightblue, xshift=9.75cm, yshift=-1cm] (-1,0) -- (1,0) -- (0,1) -- (-1,0) ;
	\draw[xshift=9.75 cm, line width= 0.07 cm, blue, yshift=-1cm] (-0.3,0) arc (0:45:0.7);
	\draw[xshift=9.75 cm, line width= 0.07 cm, blue, yshift=-1cm] (0.3,0) arc (180:135:0.7);
	\draw[xshift=9.75cm, line width=0.035cm, yshift=-1cm] (1,0) -- (0,1) -- (-1,0) ;
	\draw[line width=0.035cm, red, xshift=9.75cm, yshift=-1cm] (-1,0) -- (1,0);
	\end{tikzpicture}}

\newcommand{\planartreeexampleDashedDense}{
\begin{tikzpicture}
    \draw [line width = 0.035 cm, orange,opacity=0.99][densely dotted] (0,0)--(0,1);
    \draw [line width = 0.035 cm, orange,opacity=0.99] (0,1)--(-0.75*2,2);
    \draw [line width = 0.035 cm] (0,1)--(-0.75*1,2);
    \draw [line width = 0.035 cm] (0,1)--(0,2);
    \draw [line width = 0.035 cm] (0,1)--(0.75*1,2);
    \draw [line width = 0.035 cm] (0,1)--(0.75*2,2);

    \draw [line width = 0.035 cm] (-0.75*2,2)--(-0.75*2-0.5,3);
    \draw [line width = 0.035 cm] (-0.75*1,2)--(-0.75*1-0.5,3);
    \draw [line width = 0.035 cm] (0,2)--(-0.25,3);
    \draw [line width = 0.035 cm] (0,2)--(0.25,3);
    \draw [line width = 0.035 cm] (0.75*1,2)--(0.75*1,3);
    \draw [line width = 0.035 cm] (0.75*2,2)--(0.75*2,3);
    \draw [line width = 0.035 cm] (0.75*2,2)--(0.75*2+0.75,3);

    \draw [line width = 0.035 cm] (-0.75*2-0.5,3)--(-0.75*2-0.5-0.25,4);
    \draw [line width = 0.035 cm] (-0.75*2-0.5,3)--(-0.75*2-0.5+0.25,4);
    \draw [line width = 0.035 cm] (-0.75*1-0.5,3)--(-0.75*1-0.25-0.5,4);
    \draw [line width = 0.035 cm] (-0.75*1-0.5,3)--(-0.75*1+0.25-0.5,4);
    \draw [line width = 0.035 cm] (-0.25,3)--(-0.25,4);
    \draw [line width = 0.035 cm] (0.75*1,3)--(0.75*1-0.25,4);
    \draw [line width = 0.035 cm] (0.75*1,3)--(0.75*1+0.25,4);
    \draw [line width = 0.035 cm] (0.75*2,3)--(0.75*2-0.25,4);
    \draw [line width = 0.035 cm] (0.75*2,3)--(0.75*2+0.25,4);
    \draw [line width = 0.035 cm] (0.75*2+0.75,3)--(0.75*2+0.75,4);
    \draw [line width = 0.035 cm] (0.75*2+0.75,3)--(0.75*2+0.75+0.5,4);
    
    \node [fill=black,inner sep=1.25pt] at (0,0) {};
    \node [fill=black,inner sep=1.25pt] at (0,1) {};
    
    \node [fill=black,inner sep=1pt, yshift=1.25cm] at (0,3.2) {};
	\node [fill=black,inner sep=1pt, yshift=1.25cm] at (0,3.4) {};
	\node [fill=black,inner sep=1pt, yshift=1.25cm]  at (0,3.6)  {};
	
	\fill[blue, line width=0.05cm]  (-0.75*1,2) circle [radius=3pt];
    \fill[blue, line width=0.05cm]  (0.75*1,2) circle [radius=3pt];
	
	\fill[blue, line width=0.05cm]  (0.75*1,3) circle [radius=3pt];
	
	\fill[blue, line width=0.05cm]  (-0.75*2-0.5+0.25,4) circle [radius=3pt];
    \fill[blue, line width=0.05cm]  (-0.25,4) circle [radius=3pt];
    \fill[blue, line width=0.05cm]  (0.75*2+0.75,4) circle [radius=3pt];
    
\end{tikzpicture}}

\newcommand{\planartreeexampleDashedDilute}{
\begin{tikzpicture}
    \draw [line width = 0.035 cm, orange,opacity=0.99][densely dotted] (0,0)--(0,1);
    \draw [line width = 0.035 cm, orange,opacity=0.99] (0,1)--(-0.75*2,2);
    \draw [line width = 0.035 cm] (0,1)--(-0.75*1,2);
    \draw [line width = 0.035 cm] (0,1)--(0,2);
    \draw [line width = 0.035 cm] (0,1)--(0.75*1,2);
    \draw [line width = 0.035 cm] (0,1)--(0.75*2,2);

    \draw [line width = 0.035 cm] (-0.75*2,2)--(-0.75*2-0.5,3);
    \draw [line width = 0.035 cm] (-0.75*1,2)--(-0.75*1-0.5,3);
    \draw [line width = 0.035 cm] (0,2)--(-0.25,3);
    \draw [line width = 0.035 cm] (0,2)--(0.25,3);
    \draw [line width = 0.035 cm] (0.75*1,2)--(0.75*1,3);
    \draw [line width = 0.035 cm] (0.75*2,2)--(0.75*2,3);
    \draw [line width = 0.035 cm] (0.75*2,2)--(0.75*2+0.75,3);

    \draw [line width = 0.035 cm] (-0.75*2-0.5,3)--(-0.75*2-0.5-0.25,4);
    \draw [line width = 0.035 cm] (-0.75*2-0.5,3)--(-0.75*2-0.5+0.25,4);
    \draw [line width = 0.035 cm] (-0.75*1-0.5,3)--(-0.75*1-0.25-0.5,4);
    \draw [line width = 0.035 cm] (-0.75*1-0.5,3)--(-0.75*1+0.25-0.5,4);
    \draw [line width = 0.035 cm] (-0.25,3)--(-0.25,4);
    \draw [line width = 0.035 cm] (0.75*1,3)--(0.75*1-0.25,4);
    \draw [line width = 0.035 cm] (0.75*1,3)--(0.75*1+0.25,4);
    \draw [line width = 0.035 cm] (0.75*2,3)--(0.75*2-0.25,4);
    \draw [line width = 0.035 cm] (0.75*2,3)--(0.75*2+0.25,4);
    \draw [line width = 0.035 cm] (0.75*2+0.75,3)--(0.75*2+0.75,4);
    \draw [line width = 0.035 cm] (0.75*2+0.75,3)--(0.75*2+0.75+0.5,4);
    
    \node [fill=black,inner sep=1.25pt] at (0,0) {};
    \node [fill=black,inner sep=1.25pt] at (0,1) {};
    
    \node [fill=black,inner sep=1pt, yshift=1.25cm] at (0,3.2) {};
	\node [fill=black,inner sep=1pt, yshift=1.25cm] at (0,3.4) {};
	\node [fill=black,inner sep=1pt, yshift=1.25cm]  at (0,3.6)  {};
	
	\fill[blue, line width=0.05cm]  (0.75*1,2) circle [radius=3pt];
    \fill[blue, line width=0.05cm]  (0.75*2,2) circle [radius=3pt];
	
	\fill[blue, line width=0.05cm]  (0.75*1,3) circle [radius=3pt];
	\fill[blue, line width=0.05cm]  (-0.75*2-0.5,3) circle [radius=3pt];
	\fill[blue, line width=0.05cm]  (-0.75*2+0.25,3) circle [radius=3pt];
	\fill[blue, line width=0.05cm]  (-0.25,3) circle [radius=3pt];
	\fill[blue, line width=0.05cm]  (0.25,3) circle [radius=3pt];
	\fill[blue, line width=0.05cm]  (0.75*3,3) circle [radius=3pt];
	
	\fill[blue, line width=0.05cm]  (-0.75*3+0.5,4) circle [radius=3pt];
	\fill[blue, line width=0.05cm]  (0.75*1-0.25,4) circle [radius=3pt];
    \fill[blue, line width=0.05cm]  (0.75*2-0.25,4) circle [radius=3pt];
    \fill[blue, line width=0.05cm]  (0.75*3,4) circle [radius=3pt];
    
\end{tikzpicture}}

\newcommand{\treexampleDashedDense}{
\begin{tikzpicture}
\fill[opacity=0.25,lightblue] (0,0) circle [radius=3 cm];

	\draw[line width= 0.07 cm, blue] (0:1) .. controls (5:0.7) .. (30:0.7) ;
	\draw[line width= 0.07 cm, blue] (170:0.7) .. controls (180:0.8) .. (183:1) ;
	\draw[line width= 0.07 cm, blue] (30:0.7) .. controls (80:0.7) and (120:0.6) .. (170:0.7) ;
	\draw[line width= 0.07 cm, blue] (183:1) .. controls  (185:1.3) and (160:1.3)  .. (138:1.3) ;
	\draw[line width= 0.07 cm, blue] (170:3) .. controls  (160:2.5) and (110:2.5) .. (103:3) ;
	\draw[line width= 0.07 cm, blue] (205:1) .. controls (235:3) and (180:2.7) .. (182:3);
	\draw[line width= 0.07 cm, blue] (312:1) .. controls (190:0.7) .. (205:1) ;
	\draw[line width= 0.07 cm, blue] (0:1) .. controls (20:1.8) and (90:2) .. (138:1.3) ;
	\draw[line width= 0.07 cm, blue] (312:1) .. controls (320:2.2) and (220:1.5) .. (250:2) ;
	\draw[line width= 0.07 cm, blue] (250:2) .. controls (260:2.7) and (285:2) .. (307:3) ;
	\draw[line width= 0.07 cm, blue] (85:3) .. controls (80:2.4) and (20:3) .. (350:3) ;

\draw[line width = 0.035 cm, orange,opacity=0.99][densely dotted] (122:0.6)--(0,0);
\draw[line width=0.03cm,orange,opacity=0.99] (0,0) -- (90:1); 
\draw[line width=0.03cm] (0,0) -- (30:1 cm);
\draw[line width=0.03cm] (0,0) -- (220:1);
\draw[line width=0.03cm] (0,0) -- (300:1);
\draw[line width=0.03cm] (0,0) -- (170:1);
\draw[line width=0.005cm,dashed] (90:1) -- (45:2);
\draw[line width=0.03cm] (30:1) -- (45:2);
\draw[line width=0.03cm] (90:1) -- (120:2);
\draw[line width=0.005cm,dashed] (170:1) -- (120:2);
\draw[line width=0.03cm] (170:1) -- (160:2);
\draw[line width=0.03cm] (170:1) -- (190:2);
\draw[line width=0.005cm,dashed] (30:1) -- (330:2);
\draw[line width=0.005cm,dashed] (220:1) -- (190:2);
\draw[line width=0.03cm] (220:1) -- (260:2);
\draw[line width=0.005cm,dashed] (300:1) -- (260:2);
\draw[line width=0.03cm] (300:1) -- (330:2);
\draw[line width=0.03cm] (300:1) -- (290:2);
\draw[line width=0.005cm,dashed] (45:2) -- (0:3);
\draw[line width=0.03cm] (45:2) -- (18:3);
\draw[line width=0.03cm] (45:2) -- (70:3);
\draw[line width=0.005cm,dashed] (120:2) -- (70:3);
\draw[line width=0.03cm] (120:2) -- (110:3);
\draw[line width=0.03cm] (120:2) -- (130:3);
\draw[line width=0.005cm,dashed] (160:2)-- (130:3);
\draw[line width=0.03cm] (160:2)-- (160:3);
\draw[line width=0.03cm] (160:2)-- (195:3);
\draw[line width=0.005cm,dashed] (190:2) -- (195:3);
\draw[line width=0.03cm] (190:2) -- (220:3);
\draw[line width=0.03cm] (190:2) -- (240:3);
\draw[line width=0.005cm,dashed] (260:2)--(240:3);
\draw[line width=0.03cm] (260:2)--(270:3);
\draw[line width=0.03cm] (260:2)--(300:3);
\draw[line width=0.03cm] (330:2)--(0:3);
\draw[line width=0.005cm,dashed] (330:2)--(300:3);
\draw[line width=0.005cm,dashed] (290:2)--(300:3);
\node [fill=black,inner sep=1pt] at (0:3.2) {};
\node [fill=black,inner sep=1pt] at (0:3.4) {};
\node [fill=black,inner sep=1pt]  at (0:3.6)  {};
\node [fill=black,inner sep=1pt] at (90:3.2) {};
\node [fill=black,inner sep=1pt] at (90:3.4) {};
\node [fill=black,inner sep=1pt] at (90:3.6) {};
\node [fill=black,inner sep=1pt] at (180:3.2) {};
\node [fill=black,inner sep=1pt] at (180:3.4) {};
\node [fill=black,inner sep=1pt] at (180:3.6) {};
\node [fill=black,inner sep=1pt] at (270:3.2) {};
\node [fill=black,inner sep=1pt] at (270:3.4) {};
\node [fill=black,inner sep=1pt] at (270:3.6) {};
\draw[red, line width=0.005cm,dashed] (0,0) circle [radius=1 cm];
\draw[red, line width=0.005cm,dashed] (0,0) circle [radius=2 cm];
\draw[red, line width=0.005cm,dashed] (0,0) circle [radius=3 cm];
\node [fill=black,inner sep=1.25pt] at (0,0) {};
\node [fill=black,inner sep=1.25pt] at (122:0.6){};
\end{tikzpicture}}

\newcommand{\treexampleDashedDilute}{
\begin{tikzpicture}
\fill[opacity=0.25,lightblue] (0,0) circle [radius=3 cm];

\draw[line width= 0.07 cm, blue] (183:1) .. controls (170:0.7) and (140:0.5) .. (130:1) ;
	\draw[line width= 0.07 cm, blue] (130:1) .. controls (130:1.75) and (75:2) .. (82:3) ;

	\draw[line width= 0.07 cm, blue] (183:1) .. controls (185:1.4) and (150:1.8) .. (145:2) ;
	\draw[line width= 0.07 cm, blue] (145:2) .. controls (140:2.5) and (160:2.5) .. (167:3) ;
	
	\draw[line width= 0.07 cm, blue] (315:2) .. controls (320:2.5) and (350:2.5) .. (350:2) ;
	\draw[line width= 0.07 cm, blue] (315:2) .. controls (320:1.25) and (350:1.7) .. (350:2) ;
	
	\draw[line width= 0.07 cm, blue] (230:3) .. controls (260:1.2).. (292:3) ;

\draw[line width = 0.035 cm, orange,opacity=0.99][densely dotted] (122:0.6)--(0,0);
\draw[line width=0.03cm,orange,opacity=0.99] (0,0) -- (90:1); 
\draw[line width=0.03cm] (0,0) -- (30:1 cm);
\draw[line width=0.03cm] (0,0) -- (220:1);
\draw[line width=0.03cm] (0,0) -- (300:1);
\draw[line width=0.03cm] (0,0) -- (170:1);
\draw[line width=0.005cm,dashed] (90:1) -- (45:2);
\draw[line width=0.03cm] (30:1) -- (45:2);
\draw[line width=0.03cm] (90:1) -- (120:2);
\draw[line width=0.005cm,dashed] (170:1) -- (120:2);
\draw[line width=0.03cm] (170:1) -- (160:2);
\draw[line width=0.03cm] (170:1) -- (190:2);
\draw[line width=0.005cm,dashed] (30:1) -- (330:2);
\draw[line width=0.005cm,dashed] (220:1) -- (190:2);
\draw[line width=0.03cm] (220:1) -- (260:2);
\draw[line width=0.005cm,dashed] (300:1) -- (260:2);
\draw[line width=0.03cm] (300:1) -- (330:2);
\draw[line width=0.03cm] (300:1) -- (290:2);
\draw[line width=0.005cm,dashed] (45:2) -- (0:3);
\draw[line width=0.03cm] (45:2) -- (18:3);
\draw[line width=0.03cm] (45:2) -- (70:3);
\draw[line width=0.005cm,dashed] (120:2) -- (70:3);
\draw[line width=0.03cm] (120:2) -- (110:3);
\draw[line width=0.03cm] (120:2) -- (130:3);
\draw[line width=0.005cm,dashed] (160:2)-- (130:3);
\draw[line width=0.03cm] (160:2)-- (160:3);
\draw[line width=0.03cm] (160:2)-- (195:3);
\draw[line width=0.005cm,dashed] (190:2) -- (195:3);
\draw[line width=0.03cm] (190:2) -- (220:3);
\draw[line width=0.03cm] (190:2) -- (240:3);
\draw[line width=0.005cm,dashed] (260:2)--(240:3);
\draw[line width=0.03cm] (260:2)--(270:3);
\draw[line width=0.03cm] (260:2)--(300:3);
\draw[line width=0.03cm] (330:2)--(0:3);
\draw[line width=0.005cm,dashed] (330:2)--(300:3);
\draw[line width=0.005cm,dashed] (290:2)--(300:3);
\node [fill=black,inner sep=1pt] at (0:3.2) {};
\node [fill=black,inner sep=1pt] at (0:3.4) {};
\node [fill=black,inner sep=1pt]  at (0:3.6)  {};
\node [fill=black,inner sep=1pt] at (90:3.2) {};
\node [fill=black,inner sep=1pt] at (90:3.4) {};
\node [fill=black,inner sep=1pt] at (90:3.6) {};
\node [fill=black,inner sep=1pt] at (180:3.2) {};
\node [fill=black,inner sep=1pt] at (180:3.4) {};
\node [fill=black,inner sep=1pt] at (180:3.6) {};
\node [fill=black,inner sep=1pt] at (270:3.2) {};
\node [fill=black,inner sep=1pt] at (270:3.4) {};
\node [fill=black,inner sep=1pt] at (270:3.6) {};
\draw[red, line width=0.005cm,dashed] (0,0) circle [radius=1 cm];
\draw[red, line width=0.005cm,dashed] (0,0) circle [radius=2 cm];
\draw[red, line width=0.005cm,dashed] (0,0) circle [radius=3 cm];
\node [fill=black,inner sep=1.25pt] at (0,0) {};
\node [fill=black,inner sep=1.25pt] at (122:0.6){};
\end{tikzpicture}}

\newcommand{\idhangexpan}{
\begin{tikzpicture}
    \fill [lightblue, opacity=0.35] (0,0) -- (1.5,0) -- (0.75,-1.3)--(0,0);
    
    \draw[line width  = 0.15cm, blue,xshift=0.025cm] (0.35,-0.5*1.3) arc (135:45:0.535);
    
    \draw [line width=0.075cm] (1.5,0) -- (0.75,-1.3)--(0,0);
    \draw [line width=0.075cm, red] (0,0) -- (1.5,0);
\end{tikzpicture}}

\newcommand{\ehangexpan}{
\begin{tikzpicture}
    \fill [lightblue, opacity=0.35] (0,0) -- (1.5,0) -- (0.75,-1.3)--(0,0);
    
    \draw[line width  = 0.15cm, blue] (0.75-0.2,0) arc (360:315:0.8);
    \draw[line width  = 0.15cm, blue] (0.75+0.2,0) arc (180:225:0.8);
    
    \draw [line width=0.075cm] (1.5,0) -- (0.75,-1.3)--(0,0);
    \draw [line width=0.075cm, red] (0,0) -- (1.5,0);
\end{tikzpicture}}

\newcommand{\idstandexpan}{
\begin{tikzpicture}
    \fill [lightblue, opacity=0.45] (0,0) -- (1.5,0) -- (0.75,1.3)--(0,0);
    
    \draw[line width  = 0.15cm, blue,xshift=0.025cm] (0.35,0.5*1.3) arc (225:315:0.535);
    
    \draw [line width=0.075cm] (1.5,0) -- (0.75,1.3)--(0,0);
    \draw [line width=0.075cm, red] (0,0) -- (1.5,0);
\end{tikzpicture}}

\newcommand{\estandexpan}{
\begin{tikzpicture}
    \fill [lightblue, opacity=0.45] (0,0) -- (1.5,0) -- (0.75,1.3)--(0,0);
    
    \draw[line width  = 0.15cm, blue] (0.75-0.2,0) arc (0:45:0.8);
    \draw[line width  = 0.15cm, blue] (0.75+0.2,0) arc (180:135:0.8);
    
    \draw [line width=0.075cm] (1.5,0) -- (0.75,1.3)--(0,0);
    \draw [line width=0.075cm, red] (0,0) -- (1.5,0);
\end{tikzpicture}}

\newcommand{\idhangMD}{
\begin{tikzpicture}
    \fill [lightblue, opacity=0.35] (0,0) -- (1.5,0) -- (0.75,-1.3)--(0,0);
    
    \draw [line width=0.075cm] (1.5,0) -- (0.75,-1.3)--(0,0);
    \draw [line width=0.075cm, red] (0,0) -- (1.5,0);
 
\end{tikzpicture}}

\newcommand{\idstandMD}{
\begin{tikzpicture}
    \fill [lightblue, opacity=0.35] (0,0) -- (1.5,0) -- (0.75,1.3)--(0,0);
    
    \draw [line width=0.075cm] (1.5,0) -- (0.75,1.3)--(0,0);
    \draw [line width=0.075cm, red] (0,0) -- (1.5,0);

\end{tikzpicture}}

\newcommand{\eLhangexpan}{
\begin{tikzpicture}
    \fill [lightblue, opacity=0.35] (0,0) -- (1.5,0) -- (0.75,-1.3)--(0,0);
    
    \draw[line width  = 0.15cm, blue] (0.75,0) arc (180:225:1.05);
    
    \draw [line width=0.075cm] (1.5,0) -- (0.75,-1.3)--(0,0);
    \draw [line width=0.075cm, red] (0,0) -- (1.5,0);
    
\end{tikzpicture}}

\newcommand{\eRhangexpan}{
\begin{tikzpicture}
    \fill [lightblue, opacity=0.35] (0,0) -- (1.5,0) -- (0.75,-1.3)--(0,0);
    
    \draw[line width  = 0.15cm, blue] (0.75,0) arc (360:315:1.05);
    
    \draw [line width=0.075cm] (1.5,0) -- (0.75,-1.3)--(0,0);
    \draw [line width=0.075cm, red] (0,0) -- (1.5,0);

\end{tikzpicture}}

\newcommand{\eLstandexpan}{
\begin{tikzpicture}
    \fill [lightblue, opacity=0.35] (0,0) -- (1.5,0) -- (0.75,1.3)--(0,0);
    
    \draw[line width  = 0.15cm, blue] (0.75,0) arc (180:135:1.05);
    
    \draw [line width=0.075cm] (1.5,0) -- (0.75,1.3)--(0,0);
    \draw [line width=0.075cm, red] (0,0) -- (1.5,0);
    
\end{tikzpicture}}

\newcommand{\eRstandexpan}{
\begin{tikzpicture}
    \fill [lightblue, opacity=0.35] (0,0) -- (1.5,0) -- (0.75,1.3)--(0,0);
    
    \draw[line width  = 0.15cm, blue] (0.75,0) arc (0:45:1.05);
    
    \draw [line width=0.075cm] (1.5,0) -- (0.75,1.3)--(0,0);
    \draw [line width=0.075cm, red] (0,0) -- (1.5,0);

\end{tikzpicture}}

\begin{document}
	
\tikzset{
	position label/.style={
		below = 3pt,
		text height = 1.5ex,
		text depth = 1ex
	},
	brace/.style={
		decoration={brace, mirror},
		decorate
	}
}

\thispagestyle{empty}
\quad

\vspace{2cm}
\begin{center}

\textbf{\huge Critical behaviour of loop models}\\
\vspace{0.25cm}
\textbf{{\huge on causal triangulations}}
\vspace{1cm}

{\Large Bergfinnur Durhuus$^{a}$ ~~~~ Xavier Poncini$^{b}$}

\vspace{.2cm}

{\Large J{\o}rgen Rasmussen$^{b}$ ~~~~ Meltem Ünel$^{a}$}

\vspace{0.5cm}

$^a$ {\it Department of Mathematical Sciences, Copenhagen University\\
Universitetsparken 5, DK-2100 Copenhagen {\O}, Denmark}

$^b$ {\it School of Mathematics and Physics, University of Queensland\\ St Lucia, Brisbane, Queensland 4072, Australia}

\vspace{0.5cm} {\sf durhuus@math.ku.dk\quad\ x.poncini@uq.edu.au\quad\ j.rasmussen@uq.edu.au\quad\
 meltem@math.ku.dk}

\vspace{1.4cm}

{\large\textbf{Abstract}}\end{center}

We introduce a dense and a dilute loop model on causal dynamical triangulations. Both models are characterised by a geometric coupling constant $g$ and a loop parameter $\alpha$ in such a way that the purely geometric causal triangulation model is recovered for $\alpha=1$. We show that the dense loop model can be mapped to a solvable planar tree model, whose partition function we compute explicitly and use to determine the critical behaviour of the loop model. The dilute loop model can likewise be mapped to a planar tree model; however, a closed-form expression for the corresponding partition function is not obtainable using the standard methods employed in the dense case. Instead, we derive bounds on the critical coupling $g_c$ and apply transfer matrix techniques to examine the critical behaviour for $\alpha$ small.

\newpage

\tableofcontents

\section{Introduction}\label{sec:1}

If a two-dimensional statistical mechanical model with a second-order phase transition is coupled to a random background, its critical 
exponents may change and there may be a back-reaction on the background geometry changing its Hausdorff dimension. A prominent 
example of this phenomenon is the Ising model on a random two-dimensional triangulation (or quadrangulation), as demonstrated 
in \cite{kazakov1986ising}. Other examples are dimer models \cite{staudacher1990yang}, Potts models \cite{daul1995q,zinn2000dilute},
and multicritical models \cite{kostov1992multicritical}, see also \cite{ambjorn1997quantum} for an overview. 
The relation between the critical exponents, or scaling dimensions, of a matter field on a flat background and on a random curved  
background is given quite generally by the KPZ-formula of Liouville quantum gravity \cite{knizhnik1988fractal}. A conjectured formula for 
the Hausdorff dimension of the background geometry as a function of the central charge of the matter fields can be found 
in \cite{watabiki1993fractal}, although recent mathematical results in Liouville quantum gravity \cite{goswami, ding2020fractal} imply restrictions on the possible range of validity of this formula.

It is natural to ask how universal this so-called dressing of critical exponents is with respect to the ensembles of background geometries 
considered. In particular, it is natural to compare the ensemble of unrestricted dynamical triangulations (DT) 
(see e.g.~\cite{ambjorn1997quantum}) considered in the references above, with the ensemble of causal dynamical 
triangulations (CDT) \cite{ambjorn1998non}. 
Without coupling to a matter system, these ensembles, which we shall call \emph{pure DT} and \emph{pure CDT} in the following, 
exhibit different critical behaviours, the former having Hausdorff dimension 4 \cite{ambjorn1995scaling, chassaing2004random}, while the 
latter has Hausdorff dimension 2 \cite{durhuus2010spectral}. Very few analytical results are available concerning matter systems coupled 
to CDT. A number of numerical studies have been carried out, notably for Ising type and Potts type models
\cite{ambjorn1999numI, ambjorn2009numP, ambjorn2011second}, but there is no clear indication of a change in the
critical exponents. In \cite{atkin2012analytical,ambjorn2014restricted}, a class of restricted dimer models are mapped to
certain labelled tree models. Using this, the corresponding Hausdorff dimension is found to be affected by the dimer system, 
although the underlying mechanism remains unclear. 
In the work \cite{di2000integrable}, a class of CDT models 
with curvature-dependent weights are found to exhibit the same scaling behaviour as pure CDT.

In statistical mechanics, one usually works with local degrees of freedom, such as spins or heights, as in the Ising and Potts models above.
However, percolation and polymer systems, for example, require that one keeps track of connectivities or some other
inherently nonlocal degrees of freedom, and this shift in paradigm has a profound effect on the physical properties of the models.
Critical fully-packed loop models on regular square lattices have thus been found to give rise to {\em logarithmic} conformal field theories in 
the continuum scaling limit \cite{PRZ06}. Using underlying Temperley-Lieb algebraic structures \cite{TL,Jones}, these loop models
are found to be Yang-Baxter integrable and amenable to exact solutions. One of these models describes critical
dense polymers and has been solved exactly on the strip \cite{PRpol07}, the cylinder \cite{PRVpol10} and the torus \cite{MDPRpol13},
confirming predictions about scaling dimensions made in \cite{Duplantier86,Saleur87,SaleurDuplantier87}.
Other types of loop models have also been constructed, including dilute loop models associated with the $\mathcal{O}(n)$ 
models \cite{BN89,Nienhuis90} where the configurations may contain spaces of variable sizes in between the loop segments.
Loop models have also been coupled to random surfaces \cite{kostov1992multicritical,EK95}, including random triangulations \cite{DiFGK99,BBG12}.
However, to the best of our knowledge, loop models have yet to be coupled to CDT.

In this paper, we introduce and study two models of loop configurations on two-dimensional causal dynamical triangulations: a {\em dense loop model} and a {\em dilute loop model}, 
reminiscent of the familiar fully-packed and dilute loop models, respectively.
Both models are characterised by a geometric coupling constant $g$ associated with the underlying triangulations, as in pure CDT,
and a loop parameter $\alpha$ that encodes the relative weights of the admissible loop configurations on individual elementary triangles. No weight is associated with the number of loops in the models considered here, effectively setting the corresponding loop fugacities to $1$. 
We show that the known correspondence between pure CDT and planar trees
\cite{malyshev2001two,durhuus2010spectral} 
extends to each of the loop models and a corresponding class of labelled trees. This implies simple
relations between the partition functions of the loop models and those of the associated labelled tree models.
In the case of the dense loop model, we solve the corresponding tree model exactly
and find that its Hausdorff dimension equals that of pure CDT. The critical behaviour of the loop model is readily extracted
from the closed-form expression we obtain for the partition function following our tree analysis.
Although the dilute model can likewise be mapped to a planar tree model, a closed-form expression for the corresponding 
partition function is not obtainable using the techniques employed in the dense case. 
Instead, we apply a transfer matrix formalism and use analyticity arguments to examine the critical behaviour of the loop model for $\alpha$ close to $0$.
We conclude that the critical 
behaviour for $\alpha$ small is different from that of pure CDT, and provide an explanation for this difference. Based on results in \cite{UnelInPrep1}, we argue that the Hausdorff dimension equals 1 in this phase. 
While it is possible that these results hold for more general values of $\alpha$, our analysis has not been able to confirm this.
In fact, it is consistent with our findings that there exists a transition point $\alpha_0\in(0,1)$ at which 
the scaling behaviour changes.

This paper is organised as follows. 
In Section \ref{sec:2}, we introduce the dense and dilute loop models and define the associated partition functions.
In Section \ref{sec:3}, for each of the loop models, we establish a correspondence between the set of allowed loop configurations and 
a set of labelled planar trees. 
The critical properties of the labelled tree models are investigated in Section \ref{sec:4}.
In Section \ref{sec:5}, we use the tree results to obtain an explicit expression for the dense loop model partition function from
which we determine the critical behaviour of the model.
In our analysis of the dilute loop model, we supplement the relevant tree results with insight obtained using
a transfer matrix formalism to establish a relation between the critical coupling constant and the loop parameter $\alpha$. 
For $\alpha$ small, we determine the critical behaviour of the partition function and 
identify two possible scenarios for larger values of $\alpha$. 
Finally, Section \ref{sec:6} contains some concluding remarks, including a brief discussion of possible future research directions.

\section{Loop models}
\label{sec:2}

Before defining the loop models in Sections \ref{subsec:DLM} and \ref{subsec:DiLM}, 
we recall properties of causal dynamical triangulations in Section \ref{subsec:CT}.
There and in subsequent sections, we refer to two-dimensional causal dynamical triangulations simply as {\em causal triangulations}. We also adopt the convention that $\mathbb{N}$ and $\mathbb{N}_0$ denote the set of positive integers and non-negative integers, respectively.

\subsection{Causal triangulations}\label{subsec:CT}

A causal triangulation of a disk is defined by a central vertex $x$ and a sequence of cycles (circular graphs)
$S_0\equiv\{x\},S_1,\dots, S_m$, where $m$ is the {\em height}, such that, for each $k=0,1,\ldots,m-1$, 
the annulus $A_k$, bounded by $S_k$ and $S_{k+1}$, is 
triangulated. Note that $A_0$ is merely the disk with boundary $S_1$ and vertex $x$ at its centre.
As illustrated in Figure \ref{figure: pure}, each \emph{elementary triangle} within $A_k$ is either \emph{forward-directed}, 
i.e.~has one edge in $S_k$ and the opposite vertex in $S_{k+1}$, or \emph{backward-directed}, i.e.~has one edge in $S_{k+1}$ and the 
opposite vertex in $S_k$.
Edges within a cycle are called \emph{space-like} and are coloured red, 
while edges between neighbouring cycles are called \emph{time-like} and are coloured black.
The number of space-like edges in $S_k$ is denoted by $|S_k|$, where we note that $|S_0|=0$ and $|S_k|>0$ for $k\geq1$.

A vertex $v_1$ in ${S}_1$, or equivalently the edge $\{x,v_1\}$ (coloured orange in Figure \ref{figure: pure}), is distinguished. 
From $v_1$, the rightmost outwardly emanating edge connects $v_1$ to a vertex $v_2$ in $S_2$. Likewise,
$v_2$ is linked to a vertex $v_3$ in $S_3$ via its rightmost outwardly emanating edge. Continuing this procedure till we reach
the outermost cycle $S_m$ produces a sequence of vertices $v_0\equiv x,v_1,\ldots,v_m$, where each neighbour pair is connected
by a time-like edge: $\{v_k,v_{k+1}\}$, $k=0,\ldots,m-1$.
Relative to the distinguished vertex $v_k$, a clockwise order is assigned to the space-like edges in $S_k$.
Accordingly, the first (leftmost) triangle in the annulus $A_k$, $k=1,\ldots,m-1$, is forward-directed. Note that one can apply the above prescription to construct a causal triangulation of a cylinder by omitting $S_0$ and letting $S_1$ correspond to the first instant of time.

\begin{figure}
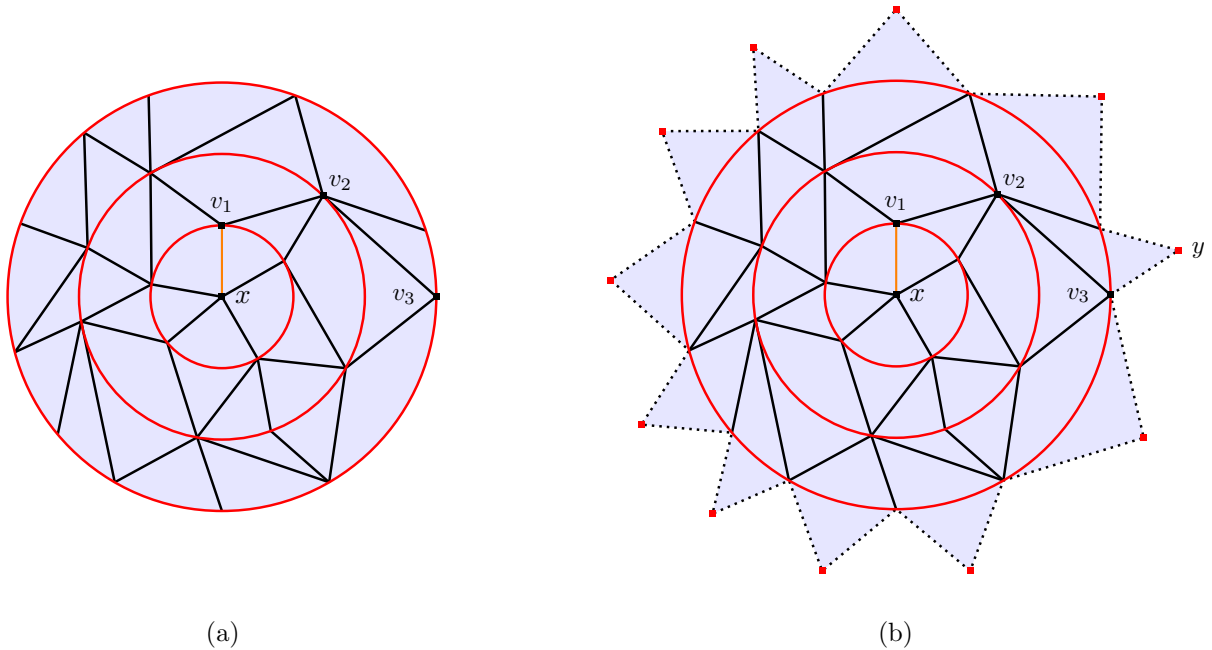

     \centering
     \begin{subfigure}[b]{0.5\textwidth}
         \centering
         \begin{align*}
            \scalebox{0.95}{\raisebox{0.6cm}{\Gexample}}\\[-13pt]
        \end{align*}
         \caption{}
         \label{figure: pure}
     \end{subfigure}
     \hfill
     \begin{subfigure}[b]{0.45\textwidth}
         \centering
         \begin{align*}
            \scalebox{0.95}{\raisebox{-5cm}{\Gexamplecapped}}
        \end{align*}
         \caption{}
         \label{fig: pureCap}
     \end{subfigure}
        \caption{(a) Causal triangulation of the disk, with distinguished edge $\{x,v_1\}$. (b) Extension of the disk triangulation in the left figure to a triangulation of the sphere, the latter being obtained by identifying neighbouring dotted time-like edges, including the outer red nodes.}
        \label{figure: disk}
\end{figure}

Let ${\mathcal C}_m(N)$ denote the set of causal disk triangulations of height $m$ with $N$ vertices,
and similarly ${\mathcal C}_m$ the set without constraints on the number of vertices, and define
$$
 |C|:=\sum_{k=0}^m|S_k|,\qquad C\in\mathcal{C}_m.
$$
That is, $|C|$ counts the number of space-like edges in $C$
and equals the number of vertices in $C$, excluding the central vertex.
We find it convenient to include the {\em degenerate case} $m=0$, where $C\equiv S_0$ and $|C|=0$ for the unique `triangulation' $C\in\mathcal{C}_0$.
The pure CDT partition functions are then defined as
\begin{equation}\label{Part1}
 Z(g):=\sum_{m=0}^{\infty} Z_m(g),\qquad 
 Z_m(g):=\sum_{C\in {\mathcal C}_m} g^{|C|}.
\end{equation}
It is well known \cite{ambjorn1998non}, and will also be shown below, that there exists a critical coupling $g_c>0$ such that $Z(g)$ is finite 
and analytic for $|g|< g_c$, while $Z(g)$ is divergent for $g>g_c$ and hence has a singularity at $g=g_c$.
We note that $Z_0(g)=1$.

To provide a geometric interpretation in terms of triangulation area, it is convenient to extend the triangulated disk $C$ to a triangulated sphere $\bar{C}$
by adjoining forward-directed triangles to each space-like edge in $S_m$, identifying the new time-like edges that are
incident on the same vertex in $S_m$. Such an extension is illustrated in Figure $\ref{fig: pureCap}$.
Note that the `outer' vertices of the added triangles are all identified, suggesting the introduction of $S_{m+1}\equiv\{y\}$,
where $y$ is the pole opposite to $x$.
In the degenerate case $C\in\mathcal{C}_0$, we set $\bar C\equiv C$ and accordingly have $y=x$.
We define the {\em area} $a(\bar{C})$ as the number of elementary triangles in $\bar{C}$.
Since $a(\bar{C})=2|C|$, it follows that
$$
Z_m(g) = \sum_{C\in {\mathcal C}_m} g^{a(\bar{C})/2}.
$$

\subsection{Dense loop model}\label{subsec:DLM}

In the \emph{dense loop model}, the set ${\mathcal L}_m^{de}(N)$ of allowed {\em loop configurations} of height $m$ 
with $N$ vertices is obtained from the 
set ${\mathcal C}_m(N)$ by replacing each elementary triangle in every configuration $C\in{\mathcal C}_m(N)$ with one of the
two similarly directed triangles in Figure \ref{figure: dense1}. These triangles are constructed by decorating the original empty triangles
with blue \textit{arc(s)}.
\begin{figure}[ht]
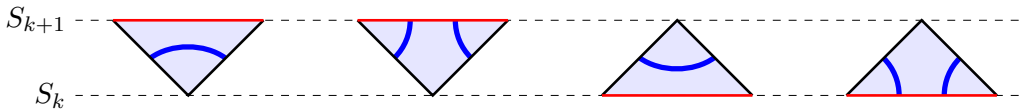

	\centering
	\bothTriDe
    \caption{The possible decorations of elementary triangles in the dense loop model.}
    \label{figure: dense1}
\end{figure}
Arcs within elementary triangles are non-intersecting and defined up to regular isotopy, that is, they can be bent and stretched but 
cannot be cut. Loop segments may form (closed) loops but are only allowed to terminate on the {\em boundary} of the 
triangulation (the cycle $S_m$).
This enforces a compatibility condition on shared space-like edges. 
The set of dense loop configurations of height $m$, without constraints on
the number of vertices, is denoted by ${\mathcal L}_m^{de}$.
An example of a dense loop configuration is depicted in Figure \ref{figure: densex}.

Each crossing of a space-like edge by a loop segment is called an \textit{intersection}. We denote by $s(L)$, respectively $|L|$, the number of intersections and space-like edges in $L\in{\mathcal L}_m^{de}$,
and define the partition functions for the dense loop model by 
\begin{equation*}
 Z^{de}(g,\alpha):= \sum_{m=0}^{\infty}Z_m^{de}(g,\alpha),\qquad
 Z_m^{de}(g,\alpha):= \sum_{L\in {\mathcal L}_m^{de}} g^{|L|}\alpha^{s(L)}.
\end{equation*}
Here $\alpha\in[0,1]$, where we take $\alpha^n$ to mean $\delta_{n,0}$ for $\alpha=0$ and $n\in\mathbb{N}_0$. Note that these partition functions reflect that the loop fugacity has been set equal to $1$.
It will be shown below that, for fixed $\alpha$, there exists $g_c^{de}(\alpha)>0$ 
such that $Z^{de}(g,\alpha)$ is analytic as a function of $g$ for $|g|<g_c^{de}(\alpha)$, 
while it is divergent for $g>g_c^{de}(\alpha)$. 
We note that $Z^{de}_0(g,\alpha)=1$.

\begin{figure}
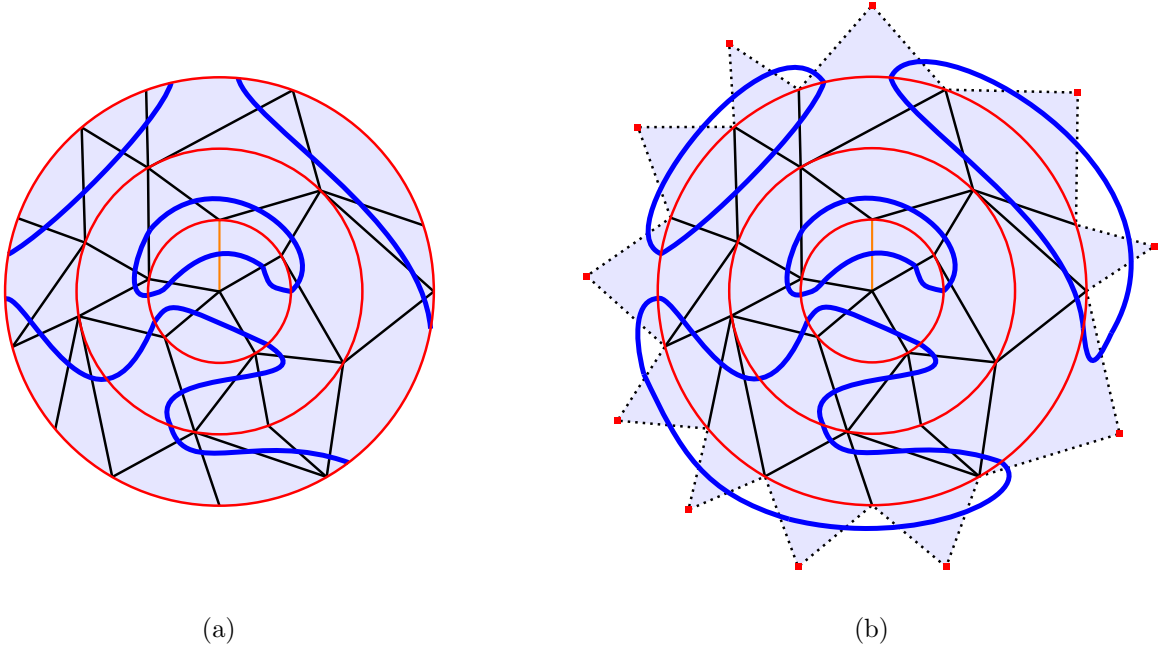

     \centering
     \begin{subfigure}[b]{0.5\textwidth}
         \centering
         \begin{align*}
            \scalebox{0.95}{\raisebox{0.6cm}{\densexample}}\\[-13pt]
        \end{align*}
         \caption{}
         \label{figure: densex}
     \end{subfigure}
     \hfill
     \begin{subfigure}[b]{0.45\textwidth}
         \centering
         \begin{align*}
            \scalebox{0.95}{\raisebox{-5cm}{\densexamplecapped}}
        \end{align*}
         \caption{}
         \label{fig: densexCap}
     \end{subfigure}
        \caption{(a) Dense loop configuration on a causal triangulation of the disk, with a distinguished edge incident to the central vertex. 
         (b) Extension of the loop configuration in the left figure to a loop configuration on the triangulated sphere, the latter being obtained by identifying neighbouring dotted time-like edges, including the outer red nodes. With the given prescription for constructing the sphere, the extension of the loop configuration is unique.}
        \label{fig:three graphs}
\end{figure}

As in the pure CDT model, we may extend the disk triangulations underlying the loop configurations to sphere triangulations
(defined as in Section \ref{subsec:CT}).
Correspondingly, we may extend the loop configuration $L$ defined on the triangulated disk $C$ to a 
loop configuration $\bar{L}$ on the triangulated sphere $\bar{C}$.
In fact, there exists a unique assignment of arcs to the additional forward-directed triangles in $\bar C$; it is obtained by mirroring their backward-directed
counterpart with which they share a space-like edge in $S_m$.
Such an extension is illustrated in Figure \ref{fig: densexCap}.
As before, the area $a(\bar L)$ is defined as the number of elementary triangles in $\bar L$, while we now also
define the {\em length} $l(\bar L)$ of the loop configuration $\bar{L}$ as the number of elementary arcs.
Noting the relations $a(\bar L)=2|L|$ and $l(\bar L)=a(\bar L)+s(L)$, we may re-express the fixed-height partition function as 
\begin{equation*}
 Z_m^{de}(g,\alpha) = \sum_{L\in {\mathcal L}_m^{de}} g^{a(\bar L)/2}\alpha^{l(\bar L)-a(\bar L)}.
\end{equation*}  

\subsection{Dilute loop model}\label{subsec:DiLM}

In the \emph{dilute loop model}, the set ${\mathcal L}_m^{di}(N)$ of allowed loop configurations of height $m$ 
with $N$ vertices is obtained from the 
set ${\mathcal C}_m(N)$ by replacing each elementary triangle in every configuration $C\in{\mathcal C}_m(N)$ with one of the
four similarly directed triangles in Figure \ref{figure:dilute1}. 
\begin{figure}[ht]
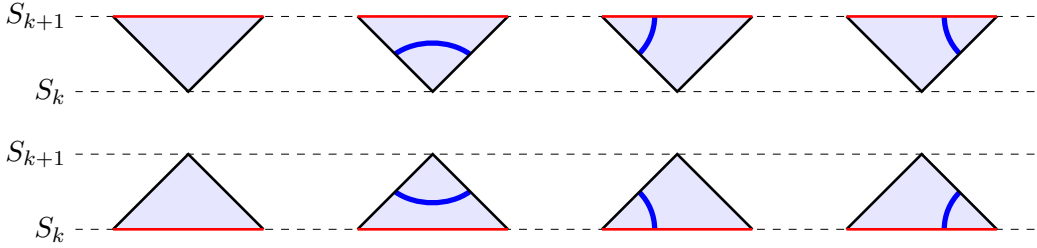

    \centering
	\scalebox{1}{\hangingTriDi}
\\[.2cm]
	\scalebox{1}{\standingTriDi}
    \caption{The possible decorations of elementary triangles in the dilute loop model.}
    \label{figure:dilute1}
\end{figure}
\noindent
These triangles are constructed by decorating the original empty triangles
with at most one blue \textit{arc}.
As in the dense loop model, an arc within an elementary triangle is defined up to regular isotopy, the crossing of a space-like edge by a loop segment is called an {\em intersection}, and loop segments may form (closed) loops but are only allowed to terminate on the boundary of the triangulation.
Contrary to the dense loop model, a compatibility condition is thereby enforced on {\em both} types of shared edges: space-like as well as time-like.
An example of a dilute loop configuration is depicted in Figure \ref{figure: dilutex}. 
We let ${\mathcal L}_m^{di}$ denote the set of dilute loop configurations of height $m$, without constraints on
the number of vertices.
The partition functions for the dilute loop model (with loop fugacity $1$) are then defined by
\begin{equation}\label{equ:diPf}
 Z^{di}(g,\alpha):= \sum_{m=0}^{\infty}Z_m^{di}(g,\alpha),\qquad
 Z_m^{di}(g,\alpha):= \sum_{L\in {\mathcal L}_m^{di}} g^{|L|}\alpha^{s(L)},
\end{equation}
where
$s(L)$, respectively $|L|$, denotes the number of intersections and space-like edges in the loop configuration $L$.
As in the dense case, $\alpha\in[0,1]$, and it will be shown that, for fixed $\alpha$, there exists $g_c^{di}(\alpha)>0$ 
such that $Z^{di}(g,\alpha)$ is analytic as a function of $g$ for $|g|<g_c^{di}(\alpha)$, while it is divergent for 
$g>g_c^{di}(\alpha)$.
We note that $Z^{di}_0(g,\alpha)=1$.

\begin{figure}
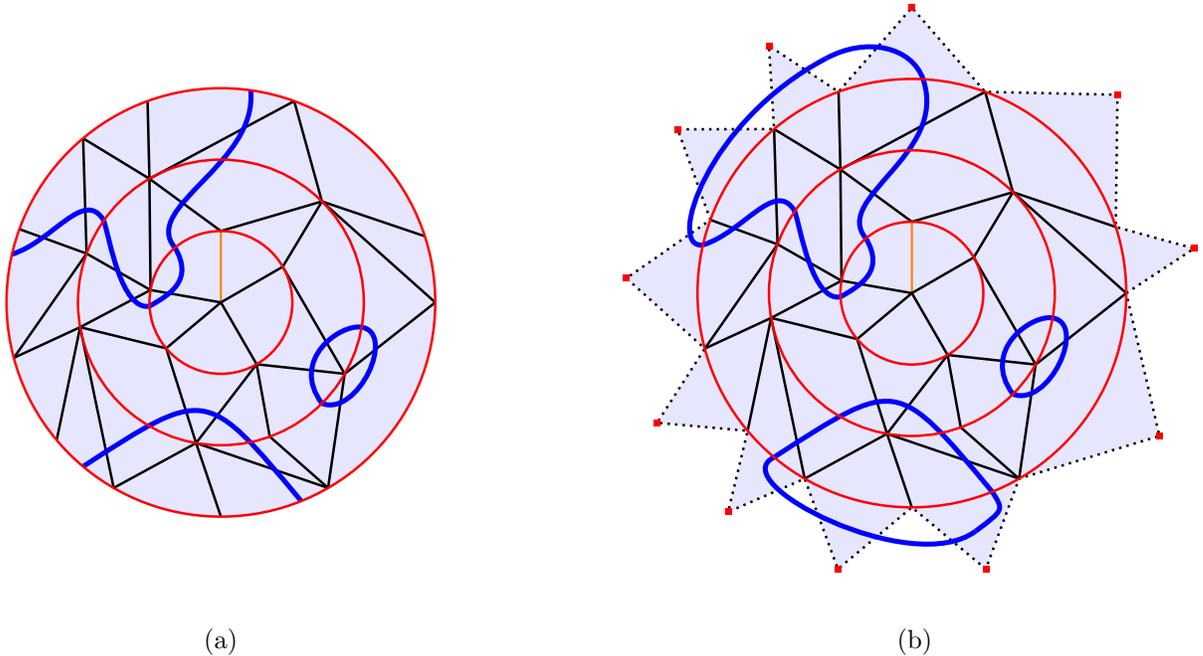

     \centering
     \begin{subfigure}[b]{0.5\textwidth}
         \centering
         \begin{align*}
            \scalebox{0.95}{\raisebox{0.6cm}{\dilutexample}}\\[-13pt]
        \end{align*}
         \caption{}
         \label{figure: dilutex}
     \end{subfigure}
     \hfill
     \begin{subfigure}[b]{0.45\textwidth}
         \centering
         \begin{align*}
            \scalebox{0.95}{\raisebox{-5cm}{\dilutexamplecapped}}
        \end{align*}
         \caption{}
         \label{fig: dilutexCap}
     \end{subfigure}
        \caption{(a) Dilute loop configuration on a causal triangulation of the disk with a distinguished edge incident to the central vertex. 
        (b) Extension of the loop configuration in the left figure to a loop configuration of the triangulated sphere, the latter being obtained by identifying neighbouring dotted time-like edges, including the outer red nodes. With the given prescription for constructing the sphere, there are two possible extensions of the loop configuration.}
        \label{fig:five graphs}
\end{figure}

In both loop models, let $s_k(L)$ denote the number of intersections of $S_k$, where we observe that $s_0(L)=0$. In the dilute loop model, the compatibility conditions along shared time-like edges,
following from the requirement that loop segments can only terminate on the boundary, imply that the parity (even or odd) of $s_k(L)$ is the same for all $k$:
\begin{equation}\label{ell}
 s_k(L)\equiv s_{k'}(L)\quad (\mathrm{mod}\ 2),\qquad \forall\,k,k'.
\end{equation}
Since the dilute loop configurations here are defined on disks (for which $s_0(L)=0$), this parity is {\em even}. 
As an immediate consequence, we have
$$
 s(L)=\sum_{k=0}^ms_k(L)\in2\mathbb{N}_0.
$$

As in the pure CDT model, we may extend the disk triangulations underlying the dilute loop configurations to sphere triangulations, see Figure \ref{fig: dilutexCap}.
As in the dense loop model, the dilute loop configuration $L$ defined on the disk triangulation $C$ is extended correspondingly to a dilute
loop configuration on the sphere triangulation $\bar{C}$, by
assigning arcs to the additional forward-directed triangles in $\bar C$. However, {\em unlike} the situation in the dense loop model,
there are exactly {\em two} possible extended dilute loop configurations, here denoted by $\bar{L}_1$ and $\bar{L}_2$.
Since $|\bar{L}_1| = |\bar{L}_2|$ and $s(\bar{L}_1) = s(\bar{L}_2)$, we can use either extension when re-expressing the partition
function (\ref{equ:diPf}) in terms of data encoded in the corresponding sphere configurations:
$$
 Z_m^{di}(g,\alpha) = \sum_{L\in {\mathcal L}_m^{di}} g^{a(\bar{L})/2}\alpha^{s(\bar{L})}.
$$  
However, contrary to the dense case, this rewriting does not admit an interpretation in terms of area and length of the loop configuration. This is re-addressed in Section \ref{sec:6}.

\section{Tree correspondences}
\label{sec:3}

In this section, we develop correspondences between models defined on causal triangulations and models on planar trees.
The partition functions of the matched constructions are closely related, facilitating the use of tree methods in the analysis of 
the CDT model partition functions.

\subsection{Pure CDT model}

Let ${\mathcal T}_m(N)$ denote the set of planar trees of height $m+1$ with $N$ edges and a root of degree $1$,
and ${\mathcal T}_m$ the set of planar trees of height $m+1$ and a root of degree $1$.
As noted in \cite{malyshev2001two,durhuus2010spectral}, there exists a bijective correspondence
\begin{equation}\label{psiCT}
    \psi: \mathcal{C}_m(N)\to\mathcal{T}_m(N).
\end{equation}
To construct it, let $C\in\mathcal{C}_m(N)$. For each vertex in $S_k$, $1\leq k<m$, remove the rightmost 
outward-pointing time-like edge, as well as all space-like edges, and add a new vertex $x_0$ and a corresponding edge $\{x,x_0\}$ 
immediately to the left (viewed outwardly) of the distinguished edge $\{x,v_1\}$. 
The resulting graph, $T=\psi(C)$, is a planar tree with root $x_0$ of degree $1$, and the map $\psi$ is readily seen to be invertible.
Note that the vertices of $C$ and $T$ are the same, except for the vertex $x_0$, and that the {\em graph distance} from $x$ to any vertex $v$ in $C$ and the one between the corresponding vertices in $T$ are the same.
An example of the construction of $\psi(C)$ is given in Figure \ref{figure:trex}.

\begin{figure}
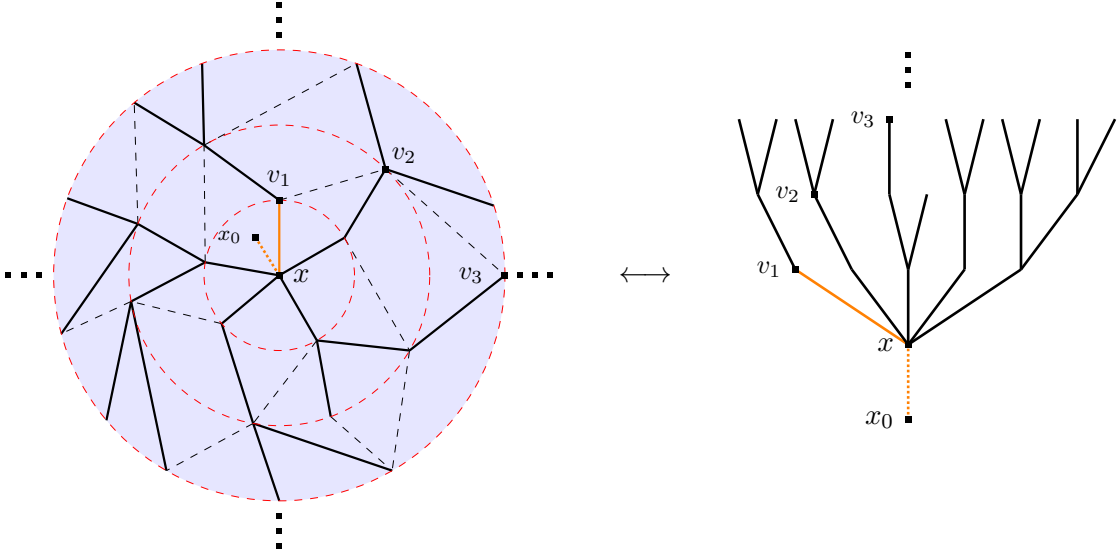

    \centering
    \begin{align*}
        \raisebox{-3.55cm}{\treexampleDashed} \qquad\longleftrightarrow\qquad \raisebox{-2.075cm}{\planartreeexampleDashed}
    \end{align*}
    \captionof{figure}{Causal triangulation of the disk $C$ and the corresponding tree $\psi(C)$.}
\label{figure:trex}
\end{figure}

We introduce
$$
    V_k(T) := \big\{v\in T\,|\,d_T(x_0,v)=k+1\big\}
$$
as the set of vertices of $T\in{\mathcal T}_m$ with graph distance $k+1$ 
from $x_0$, writing
$$
    V_k(T)= \big\{v_{k,i}\,|\,i=1,\ldots,|V_k(T)|\big\},
$$
where $v_{k,1}\equiv v_k$ is the distinguished vertex, defined in Section \ref{subsec:CT}, and the remaining $v_{k,i}$ are labelled clockwise from $v_{k,1}$, with $i=2,\ldots,|V_k(T)|$.
We also define
\begin{align}
    V(T):=\bigcup_{k=1}^{m}V_k(T),
\end{align}
which is the vertex set of $T$, excluding the vertices $x_0$ and $x$.

\subsection{Loop models}

We now proceed to extend the correspondence (\ref{psiCT}) to a pair of correspondences from the set of triangulations ${\mathcal C}_m$ to the 
respective sets ${\mathcal L}_m^{de}$ and ${\mathcal L}_m^{di}$ of 
loop configurations. To this end, let $\widetilde{\mathcal T}_m$ denote the set of planar trees of height $m+1$ whose vertices, 
except for the root $x_0$ and vertex $x$, are labelled $0$ or $1$, that is,
$$
 \widetilde{\mathcal T}_m:= \big\{(T,\delta)\,|\,T\in {\mathcal T}_m,\;\delta: V(T)\to \{0,1\}\big\}.
$$
For $(T,\delta)\in\widetilde{\mathcal T}_m$, we define the {\em labelling characteristics} at height $k+1$
\begin{align*}
    \bm{\delta}_k:= \big(\delta(v_{k,1}),\ldots,\delta(v_{k,|V_k(T)|})\big)
\end{align*}
and set
$$
 \delta_k:=\sum_{v\in V_k(T)}\delta(v) 
$$
for $k=1,\ldots,m\,$. It follows that
$$
|\delta|:=\sum_{k=1}^m\delta_k =\sum_{v\in V(T)} \delta(v)
$$
counts the number of $1$-labels in $(T,\delta)$.
For later convenience, we also introduce
\begin{align*}
 \widetilde{\mathcal T}_m(N)&:=\big\{(T,\delta)\in\widetilde{\mathcal T}_m\,|\,T\in\mathcal{T}_m(N)\big\},
 \\[.2cm]
 \widetilde{\mathcal T}_m^{ev}(N)&:=\big\{(T,\delta)\in\widetilde{\mathcal T}_m\,|\,T\in\mathcal{T}_m(N),\,
   \delta_k\in2\mathbb{N}_0,\,k=1,\ldots,m\big\},
\end{align*}
and $\widetilde{\mathcal T}_m^{ev}$ for the similar set without constraints 
on the number of vertices.
Proposition \ref{denseloop-trees} below establishes a bijective correspondence between the elements 
of ${\mathcal L}_m^{de}(N)$ and those of $\widetilde{\mathcal T}_m(N)$,
while Proposition \ref{diluteloop-trees} establishes a $2^m$ to $1$ correspondence 
between the elements of ${\mathcal L}_m^{di}(N)$ and those of $\widetilde{\mathcal T}_m^{ev}(N)$. 

For a triangulated disk $C$, we write $l_k\equiv |S_k|$, $0\leq k\leq m$, 
and refer to $l_k$ and $l_{k+1}$ as the {\em boundary lengths} of the annulus $A_k$. Any triangulation of $A_k$ can be transformed into one of the form depicted in Figure \ref{fig:standard}, called a \textit{standard triangulation}, with the same boundary lengths, $l_k$ and $l_{k+1}$,
where $v_k$ denotes the distinguished vertex in $S_k$. This can be achieved by a sequence of local flips of the form
\begin{equation}\label{fig:flip}
	\scalebox{1}{\raisebox{-0.45cm}{\begin{tikzpicture}[scale=0.5]
	\fill[opacity=0.25, lightblue] (-1,-1) -- (-1,1) -- (1,1) -- (1,-1) -- (-1,-1);
	\draw[line width=0.035cm] (-1,-1) -- (-1,1);
	\draw[line width=0.035cm] (1,-1) -- (1,1);
	\draw[line width=0.035cm] (-1,-1) -- (1,1);
	\draw[line width=0.035cm, red] (-1,1) -- (1,1);
	\draw[line width=0.035cm, red] (-1,-1) -- (1,-1);
	\end{tikzpicture}}}
	\quad
\longleftrightarrow
	\quad
	\scalebox{1}{\raisebox{-0.45cm}{\begin{tikzpicture}[scale=0.5]
	\fill[opacity=0.25, lightblue] (-1,-1) -- (-1,1) -- (1,1) -- (1,-1) -- (-1,-1);
	\draw[line width=0.035cm] (-1,-1) -- (-1,1);
	\draw[line width=0.035cm] (1,-1) -- (1,1);
	\draw[line width=0.035cm] (-1,1) -- (1,-1);
	\draw[line width=0.035cm, red] (-1,1) -- (1,1);
	\draw[line width=0.035cm, red] (-1,-1) -- (1,-1);
	\end{tikzpicture}}}
\end{equation}
as follows. An arbitrary triangulation $T_{A_k}$ of $A_k$ can be described as a sequence of forward- and backward-directed triangles, here denoted by $f$'s and $b$'s, where, by construction, the first entry is a forward-directed triangle. Suppose $T_{A_k}$ is not standard; then there exists at least one $b$ followed by a $f$, i.e. a $bf$ in the sequence. Apply the flip operation \eqref{fig:flip} to the first instance of $bf$, thereby transforming it into an $fb$.  Iteratively applying this procedure, one arrives at the standard triangulation of $A_k$.

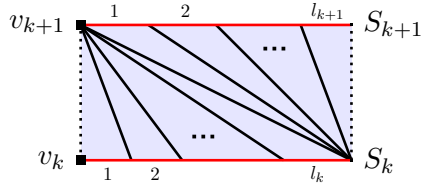
\begin{figure}[ht]
	\centering
	\begin{tikzpicture}[scale=0.9]
	\fill[opacity=0.25, lightblue] (0,-1) -- (0,1) -- (4,1) -- (4,-1) -- (0,-1);
	
	\draw[line width=0.035cm][dotted] (0,-1) -- (0,1);
	\draw[line width=0.035cm][dotted] (4,-1) -- (4,1);
	
	\draw[line width=0.035cm] (0,1) -- (0.75,-1);
	\draw[line width=0.035cm] (0,1) -- (1.5,-1);
	\draw[line width=0.035cm] (0,1) -- (3,-1);
	\draw[line width=0.035cm] (0,1) -- (4,-1);
	\draw[line width=0.035cm] (1,1) -- (4,-1);
	\draw[line width=0.035cm] (2,1) -- (4,-1);
	\draw[line width=0.035cm] (4-0.75,1) -- (4,-1);
	
	\draw[line width=0.035cm, red] (0,1) -- (4,1);
	\draw[line width=0.035cm, red] (0,-1) -- (4,-1);
	
	\node [fill=black,inner sep=0.01pt,label=east:$S_k$]  at (4,-1) {};
	\node [fill=black,inner sep=1.75pt,label=west:$v_k$]  at (0,-1) {};
	
	\node [fill=black,inner sep=0.01pt,label=east:$S_{k+1}$]  at (4,1) {};
	\node [fill=black,inner sep=1.75pt,label=west:$v_{k+1}$]  at (0,1) {};
	
	\draw[line width=0.05cm, xshift=0.9cm, yshift=-0.65cm] [dotted] (0.75,0) -- (1.1,0);
	\draw[line width=0.05cm, xshift=1.95cm, yshift=0.65cm] [dotted] (0.75,0) -- (1.1,0);
	
	\node at (0.4,-1.2) {\scalebox{0.65}{$1$}};
	\node at (1.1,-1.2) {\scalebox{0.65}{$2$}};
	\node at (3.5,-1.2) {\scalebox{0.65}{$l_{k}$}};
	\node at (0.5,1.2) {\scalebox{0.65}{$1$}};
	\node at (1.55,1.2) {\scalebox{0.65}{$2$}};
	\node at (3.65,1.2) {\scalebox{0.65}{$l_{k+1}$}};
	\end{tikzpicture}
	\caption{Standard triangulation of annulus $A_k$. The leftmost and rightmost time-like edges are dashed to indicate that they are identified.} 
	\label{fig:standard}
\end{figure}

The flip operation in (\ref{fig:flip}) is readily extended to a flip operation on the similar local components of a loop configuration
on $A_k$. In the dense loop model, the extension is given by
\begin{equation}\label{flipdense}
	\scalebox{1}{\raisebox{-0.45 cm}{\flipa}} \longleftrightarrow\, \scalebox{1}{\raisebox{-0.45 cm}{\flipb}}\qquad\; 
	\scalebox{1}{\raisebox{-0.45 cm}{\flipc}} \longleftrightarrow\, \scalebox{1}{\raisebox{-0.45 cm}{\flipd}}\qquad\; 
	\scalebox{1}{\raisebox{-0.45 cm}{\flipe}} \longleftrightarrow\, \scalebox{1}{\raisebox{-0.45 cm}{\flipf}}\qquad\;
	\scalebox{1}{\raisebox{-0.45 cm}{\flipq}} \longleftrightarrow\, \scalebox{1}{\raisebox{-0.45 cm}{\flipr}} 
\end{equation}
In the dilute loop model, the extension is given by
\begin{equation}\label{flipdilute}
 \begin{array}{cccc}
     \scalebox{1}{\raisebox{-0.45 cm}{\flipg}} \longleftrightarrow\, \scalebox{1}{\raisebox{-0.45 cm}{\fliph}}\quad\ & 
     \scalebox{1}{\raisebox{-0.45 cm}{\flipe}} \longleftrightarrow\, \scalebox{1}{\raisebox{-0.45 cm}{\flipf}}\quad\ &
     \scalebox{1}{\raisebox{-0.45 cm}{\flipi}} \longleftrightarrow\, \scalebox{1}{\raisebox{-0.45 cm}{\flipj}}\quad\ &
     \scalebox{1}{\raisebox{-0.45 cm}{\flipk}} \longleftrightarrow\, \scalebox{1}{\raisebox{-0.45 cm}{\flipl}}
      \\[0.7cm]
     \scalebox{1}{\raisebox{-0.45 cm}{\flipii}} \longleftrightarrow\, \scalebox{1}{\raisebox{-0.45 cm}{\flipjj}}\quad\ &
     \scalebox{1}{\raisebox{-0.45 cm}{\flipm}} \longleftrightarrow\, \scalebox{1}{\raisebox{-0.45 cm}{\flipn}}\quad\ &
     \scalebox{1}{\raisebox{-0.45 cm}{\flipbosa}} \longleftrightarrow\, \scalebox{1}{\raisebox{-0.45 cm}{\flipbosb}}\quad\ &
     \scalebox{1}{\raisebox{-0.45 cm}{\flipo}} \longleftrightarrow\, \scalebox{1}{\raisebox{-0.45 cm}{\flipp}}
\end{array}   
\end{equation}
We immediately have the following result.
\begin{lem}\label{lem:3.1}
\mbox{}
\begin{itemize}
\item[(i)]
The number of possible dense, respectively dilute, loop configurations on a triangulated annulus $A_k$ only depends on the 
boundary lengths $l_k$ and $l_{k+1}$, not the details of the triangulation.
\item[(ii)]
The flip operations (\ref{flipdense}) and (\ref{flipdilute}) applied to a dense, respectively dilute, loop configuration, $L$, on a 
triangulated disk leave $|L|$ and the positions of intersections, hence also $s(L)$, invariant.
\end{itemize}
\end{lem}

\noindent
A loop configuration on a triangulated annulus $A_k$ is characterised not only by the boundary lengths $l_k$ and $l_{k+1}$, 
but also by the number and location of the space-like edges which are intersected by loop segments. 
Correspondingly, for the cycle $S_k$, $k=1,\ldots,m$, the {\em intersection characteristics} are given by the $l_k$-tuple
$$
 \mathbf{n}_k=(n_{k,1},\ldots,n_{k,l_k})\in\{0,1\}^{l_k},\qquad k=1,\ldots, m,
$$
where $n_{k,i}$ takes the value $1$ if the $i$'th space-like edge (labelled clockwise from $v_k$) is intersected, and $0$ otherwise. 
Associated to the intersection characteristics $\mathbf{n}_k$, we set
$$
 n_k:=\sum_{i=1}^{l_k}n_{k,i},\qquad k=1,\ldots, m.
$$
\noindent
\textbf{Remark}. 
With notation as in Section \ref{sec:2}, $s_k(L)=2n_{k}$ for $L\in{\mathcal L}_m^{de}$, while $s_k(L)=n_{k}$ for $L\in{\mathcal L}_m^{di}$.

\begin{lem}\label{lem:IntCharDeDi}
Let $C_k$ be a triangulation of the annulus $A_k$ with boundary lengths $l_k$ and $l_{k+1}$, and let $\mathbf{n}_k\in\{0,1\}^{l_k}$\! and\, $\mathbf{n}_{k+1}\in\{0,1\}^{l_{k+1}}$.
Then, 
\begin{itemize}
\item[(i)] $C_k$ admits exactly one dense loop configuration with intersection characteristics $\mathbf{n}_k$\! and\, $\mathbf{n}_{k+1}$;
\item[(ii)] $C_k$ admits exactly two, respectively zero, dilute loop configurations with intersection characteristics $\mathbf{n}_k$\! and\, $\mathbf{n}_{k+1}$ 
if the sum $n_k+n_{k+1}$ is even, respectively odd.
\end{itemize}
\end{lem}

\begin{proof}
By Lemma \ref{lem:3.1}, it suffices to consider standard triangulations, which we will do in the following.
In the dense loop model, the entries of the intersection characteristics are in one-to-one correspondence with the possible decorations of the 
associated elementary triangles, as follows:
$$ 
        \scalebox{0.5}{\raisebox{-0.5cm}{\idhangexpan}}   \leftrightarrow\;  n_{k+1,i} = 0,\qquad 
        \scalebox{0.5}{\raisebox{-0.5cm}{\ehangexpan}}    \leftrightarrow\;  n_{k+1,i} = 1,\qquad
        \scalebox{0.5}{\raisebox{-0.5cm}{\idstandexpan}}   \leftrightarrow\;  n_{k,i} = 0 ,\qquad 
        \scalebox{0.5}{\raisebox{-0.5cm}{\estandexpan}}    \leftrightarrow\;  n_{k,i} = 1.
$$ 
Since there are no compatibility constraints along time-like edges in the dense loop model, this establishes the claim in this case.
In the dilute loop model, the entries $n_{k,i}$ and $n_{k+1,i}$ corresponding to the elementary triangles in Figure \ref{figure:dilute1} are given as follows:
    \begin{align*}
        \left\{\scalebox{0.5}{\raisebox{-0.55cm}{\idhangMD}},\;\scalebox{0.5}{\raisebox{-0.55cm}{\idhangexpan}}\right\} &\,\rightarrow\; n_{k+1,i} = 0,\qquad 
        \left\{\scalebox{0.5}{\raisebox{-0.55cm}{\eLhangexpan}},\;\scalebox{0.5}{\raisebox{-0.55cm}{\eRhangexpan}} \right\}\,\rightarrow\; n_{k+1,i} = 1,\\[.15cm]
        \left\{\scalebox{0.5}{\raisebox{-0.55cm}{\idstandMD}},\;\scalebox{0.5}{\raisebox{-0.55cm}{\idstandexpan}}\right\}&\,\rightarrow\; n_{k,i} = 0,\qquad\quad 
        \left\{\scalebox{0.5}{\raisebox{-0.55cm}{\eLstandexpan}},\;\scalebox{0.5}{\raisebox{-0.55cm}{\eRstandexpan}} \right\}\,\rightarrow\; n_{k,i} = 1.
    \end{align*}
Despite the $2$ to $1$ nature of these correspondences,
there is a unique decorated triangle associated to a given value of $n_{k,i}$ or $n_{k+1,i}$ if the `intersectedness' of the left time-like edge is known.
By successive applications of this argument, given the intersectedness of the first time-like edge in the standard triangulation of $A_k$, 
the decorations of all the forward-directed triangles are determined by the intersection characteristics $\mathbf{n}_k$. 
The intersectedness of the right edge of the last of these triangles
is the same as (respectively opposite to) that of the left edge of the first triangle if $n_k$ is even (respectively odd).
Repeating the arguments for the backward-directed triangles in $A_k$, we see that, given the intersectedness of the left time-like edge of the first
of these triangles (whose space-like edge carries the label $n_{k+1,1}$), all the 
backward-directed triangles are determined by the intersection characteristics $\mathbf{n}_{k+1}$, and the intersectedness of the right edge of the last backward-directed triangle
is the same as (respectively opposite to) that of the left edge of the first backward-directed triangle if $n_{k+1}$ is even (respectively odd).
Due to the nontrivial compatibility constraints along the time-like edges in the dilute loop model, $n_k+n_{k+1}$ must be even, as observed in (\ref{ell}), 
so the intersectedness of the right edge of the last backward-directed triangle
is the same as that of the left edge of the first forward-directed triangle in all cases, in accordance with the periodicity of the annulus.
Since the intersectedness of the left time-like edge of the first forward-directed triangle is undetermined by $\mathbf{n}_{k}$ and $\mathbf{n}_{k+1}$, 
this establishes the claim in the dilute loop model case.
\end{proof}

\begin{figure}[ht]
	\centering
	\begin{tikzpicture}[scale=0.9]
	\fill[opacity=0.25, lightblue] (0,-1) -- (0,1) -- (4,1) -- (4,-1) -- (0,-1);
	\draw[line width=0.035cm][dotted] (0,-1) -- (0,1);
	\draw[line width=0.035cm][dotted] (4,-1) -- (4,1);	
	\draw[line width= 0.07 cm, blue] (1.2,-1) arc (360:180:0.35 and -0.3);
	\draw[line width= 0.07 cm, blue] (0.3,-1) arc (360:270:0.3 and -1);
	\draw[line width=0.07 cm, blue] (1.4,-1).. controls (1.4,-0.1) and (1.4+0.8,0.1)  .. (1.4+0.8, 1);
	\draw[line width=0.07 cm, blue] (1.4+0.8+0.4, 1) arc (180:90:1.4 and -1);	
	\draw[line width=0.035cm] (0,1) -- (0.85,-1);
	\draw[line width=0.035cm] (0,1) -- (1.95,-1);
	\draw[line width=0.035cm] (0,1) -- (3,-1);
	\draw[line width=0.035cm] (0,1) -- (4,-1);
	\draw[line width=0.035cm] (1.4,1) -- (4,-1);
	\draw[line width=0.035cm] (3,1) -- (4,-1);		
	\draw[line width=0.035cm, red] (0,1) -- (4,1);
	\draw[line width=0.035cm, red] (0,-1) -- (4,-1);	
	\node [fill=black,inner sep=0.01pt,label=east:$S_k$]  at (4,-1) {};
	\node [fill=black,inner sep=1.75pt,label=west:$v_k$]  at (0,-1) {};
	
	\node [fill=black,inner sep=0.01pt,label=east:$S_{k+1}$]  at (4,1) {};
	\node [fill=black,inner sep=1.75pt,label=west:$v_{k+1}$]  at (0,1) {};
	
	\node at (0.4,-1.2) {\scalebox{0.65}{$1$}};
	\node at (1.3,-1.2) {\scalebox{0.65}{$2$}};
	\node at (2.4,-1.2) {\scalebox{0.65}{$3$}};
	\node at (3.5,-1.2) {\scalebox{0.65}{$4$}};
	\node at (0.85,1.2) {\scalebox{0.65}{$1$}};
	\node at (2.4,1.2) {\scalebox{0.65}{$2$}};
	\node at (3.55,1.2) {\scalebox{0.65}{$3$}};	
	\end{tikzpicture}
	\caption{The unique dense loop configuration on the standard triangulation of $A_k$ with intersection characteristics $\mathbf{n}_k=(1,1,0,0)$ and $\mathbf{n}_{k+1}=(0,1,0)$.}
	\label{fig:conflabel1}
\end{figure}
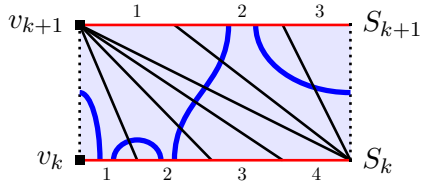

\noindent
To illustrate Lemma \ref{lem:IntCharDeDi}, Figure \ref{fig:conflabel1} depicts the unique dense loop configuration on the standard triangulation of $A_k$ 
with intersection characteristics $\mathbf{n}_k=(1,1,0,0)$ and $\mathbf{n}_{k+1}=(0,1,0)$,
while Figure \ref{fig:conflabel2} depicts the two possible dilute loop configurations on the standard triangulation of $A_k$ 
with intersection characteristics $\mathbf{n}_k=(1,1,1,1)$ and $\mathbf{n}_{k+1}=(0,1,1)$.

Intersection characteristics of loop configurations on triangulated annuli are readily extended to loop configurations on disk triangulations of 
height $m$ by specifying $\mathbf{n}_{k}$ for each $k=1,\ldots,m$.
Conversely, we say that the $m$ tuples $\mathbf{n}_k\in\{0,1\}^{l_k}$, $k=1,\ldots,m$, form {\em admissible intersection characteristics} for a disk triangulation
if there exists a corresponding loop configuration. Following Lemma \ref{lem:IntCharDeDi} and the Remark immediately preceding it, as well as (\ref{ell}),
such a set of tuples is always admissible in the dense case, and it is admissible in the dilute case if and only if $n_k\in2\mathbb{N}$, $k=1,\ldots,m$.
Moreover, given a triangulated disk with admissible intersection characteristics, the choice of loop configuration on any one of the triangulated annuli is 
independent of the choices made for the other triangulated annuli.
Thus, by Lemma \ref{lem:IntCharDeDi} and the proof thereof, we have the following two results.
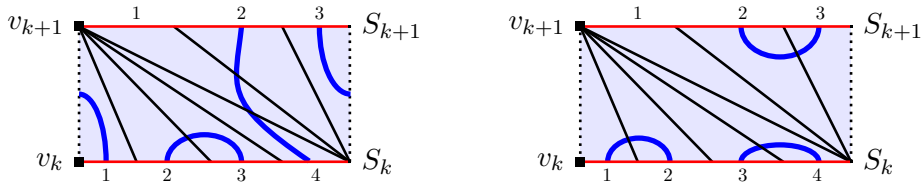
\begin{figure}[ht]
	\centering
    \begin{tikzpicture}[scale=0.9]
	\fill[opacity=0.25, lightblue] (0,-1) -- (0,1) -- (4,1) -- (4,-1) -- (0,-1);
	\draw[line width=0.035cm][dotted] (0,-1) -- (0,1);
	\draw[line width=0.035cm][dotted] (4,-1) -- (4,1);
	\draw[line width= 0.07 cm, blue] (0.4,-1) arc (360:270:0.4 and -1);
	\draw[line width= 0.07 cm, blue] (2.4,-1) arc (360:180:0.55 and -0.4);
	\draw[line width=0.07 cm, blue] (2.4+1,-0.99).. controls (2.1,-0.1) and (2.3,0.1)  .. (2.4, 1);
	\draw[line width= 0.07 cm, blue] (3.55,1) arc (180:90:0.46 and -1);
	\draw[line width=0.035cm] (0,1) -- (0.85,-1);
	\draw[line width=0.035cm] (0,1) -- (1.95,-1);
	\draw[line width=0.035cm] (0,1) -- (3,-1);
	\draw[line width=0.035cm] (0,1) -- (4,-1);
	\draw[line width=0.035cm] (1.4,1) -- (4,-1);
	\draw[line width=0.035cm] (3,1) -- (4,-1);
	\draw[line width=0.035cm, red] (0,1) -- (4,1);
	\draw[line width=0.035cm, red] (0,-1) -- (4,-1);
	\node [fill=black,inner sep=0.01pt,label=east:$S_k$]  at (4,-1) {};
	\node [fill=black,inner sep=1.75pt,label=west:$v_k$]  at (0,-1) {};
	
	\node [fill=black,inner sep=0.01pt,label=east:$S_{k+1}$]  at (4,1) {};
	\node [fill=black,inner sep=1.75pt,label=west:$v_{k+1}$]  at (0,1) {};
	\node at (0.4,-1.2) {\scalebox{0.65}{$1$}};
	\node at (1.3,-1.2) {\scalebox{0.65}{$2$}};
	\node at (2.4,-1.2) {\scalebox{0.65}{$3$}};
	\node at (3.5,-1.2) {\scalebox{0.65}{$4$}};
	\node at (0.85,1.2) {\scalebox{0.65}{$1$}};
	\node at (2.4,1.2) {\scalebox{0.65}{$2$}};
	\node at (3.55,1.2) {\scalebox{0.65}{$3$}};
	\end{tikzpicture}
	\qquad
	\begin{tikzpicture}[scale=0.9]
	\fill[opacity=0.25, lightblue] (0,-1) -- (0,1) -- (4,1) -- (4,-1) -- (0,-1);
	\draw[line width=0.035cm][dotted] (0,-1) -- (0,1);
	\draw[line width=0.035cm][dotted] (4,-1) -- (4,1);
	\draw[line width= 0.07 cm, blue] (1.325,-1) arc (360:180:0.46 and -0.35);
	\draw[line width= 0.07 cm, blue] (3.525,-1) arc (360:180:0.575 and -0.25);
	\draw[line width= 0.07 cm, blue] (3.525,1) arc (0:180:0.575 and -0.45);
	\draw[line width=0.035cm] (0,1) -- (0.85,-1);
	\draw[line width=0.035cm] (0,1) -- (1.95,-1);
	\draw[line width=0.035cm] (0,1) -- (3,-1);
	\draw[line width=0.035cm] (0,1) -- (4,-1);
	\draw[line width=0.035cm] (1.4,1) -- (4,-1);
	\draw[line width=0.035cm] (3,1) -- (4,-1);
	\draw[line width=0.035cm, red] (0,1) -- (4,1);
	\draw[line width=0.035cm, red] (0,-1) -- (4,-1);
	\node [fill=black,inner sep=0.01pt,label=east:$S_k$]  at (4,-1) {};
	\node [fill=black,inner sep=1.75pt,label=west:$v_k$]  at (0,-1) {};
	
	\node [fill=black,inner sep=0.01pt,label=east:$S_{k+1}$]  at (4,1) {};
	\node [fill=black,inner sep=1.75pt,label=west:$v_{k+1}$]  at (0,1) {};
	\node at (0.4,-1.2) {\scalebox{0.65}{$1$}};
	\node at (1.3,-1.2) {\scalebox{0.65}{$2$}};
	\node at (2.4,-1.2) {\scalebox{0.65}{$3$}};
	\node at (3.5,-1.2) {\scalebox{0.65}{$4$}};
	\node at (0.85,1.2) {\scalebox{0.65}{$1$}};
	\node at (2.4,1.2) {\scalebox{0.65}{$2$}};
	\node at (3.55,1.2) {\scalebox{0.65}{$3$}};
	\end{tikzpicture}
	\caption{The two possible dilute loop configurations on the standard triangulation of $A_k$ with intersection characteristics 
	$\mathbf{n}_k=(1,1,1,1)$ and $\mathbf{n}_{k+1}=(0,1,1)$.}
	\label{fig:conflabel2}
\end{figure}
\begin{lem}\label{lem:2m1}
A triangulated disk of height $m$ with admissible intersection characteristics $\mathbf{n}_k$, $k=1,\ldots,m$,
admits exactly one dense loop configuration or exactly \,$2^m$\! dilute loop configurations.
\end{lem}

\begin{cor}\label{dilute-dense}
\mbox{}
\begin{itemize}
\item[(i)] A triangulated cylinder with $N$ vertices admits \,$2^N$\! dense loop configurations or \,$2^N$\! dilute loop configurations.
\item[(ii)] A triangulated disk with $N$ vertices admits \,$2^{N-1}$\! dense loop configurations or \,$2^{N-1}$\! dilute loop configurations. Hence,
$$
 \big|{\mathcal L}_m^{de}(N)\big|=\big|{\mathcal L}_m^{di}(N)\big|,\qquad m,N\in\mathbb{N}.
$$
\end{itemize}
\end{cor}

\noindent The parity constraints of the dilute loop configurations compensate for the additional factor of $2$ arising from each layer.
The next proposition establishes the advertised correspondence between dense loop configurations and labelled trees. 
The correspondence is illustrated in Figure \ref{figure:trexDe}, where tree vertices labelled $1$ are indicated with blue circles.
\begin{figure}
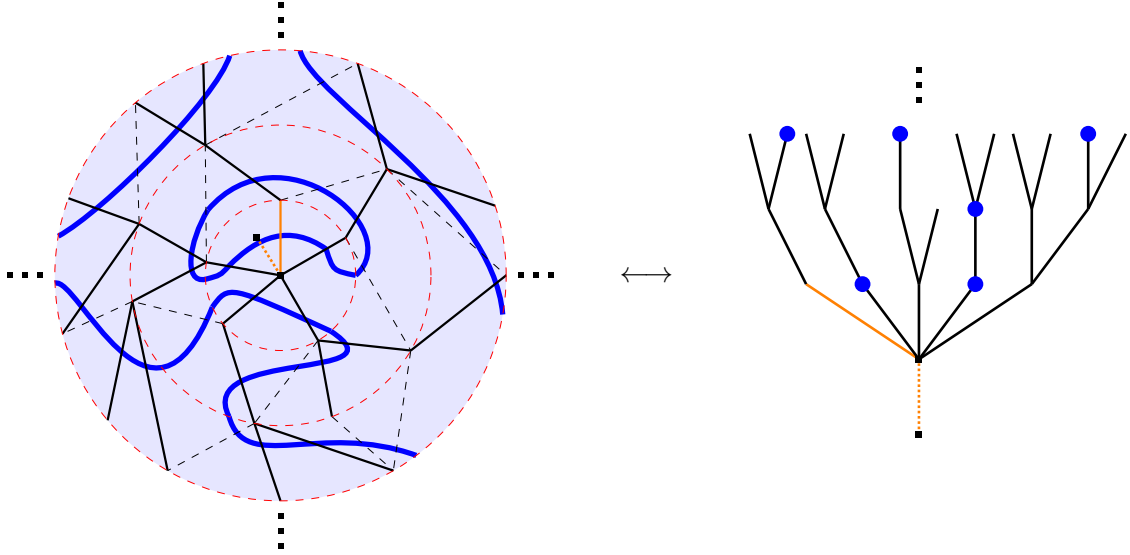

    \centering
    \begin{align*}
        \raisebox{-3.55cm}{\treexampleDashedDense}
        \qquad\longleftrightarrow\qquad
        \raisebox{-2.075cm}{\planartreeexampleDashedDense}
    \end{align*}
    \captionof{figure}{Dense loop configuration $L$ and the corresponding labelled planar tree $\tilde\psi(L)$. Vertices labelled $1$ are indicated with blue circles.}
\label{figure:trexDe}
\end{figure}
\begin{prop}
\label{denseloop-trees} 
For every $m\in\mathbb{N}_0$ and $N\in\mathbb{N}$, there is a bijective correspondence
$$
 \tilde\psi: {\mathcal L}_m^{de}(N)\to\widetilde{\mathcal T}_m(N)
$$
such that if $(T,\delta)=\tilde\psi(L)$ then $T=\psi(C)$, where $C$ is the triangulation underlying $L$. 
Moreover, $|\delta|=s(L)/2$.
\end{prop}

\begin{proof}
The map $\delta$ in the image $\tilde\psi(L)=(T,\delta)$ labels with a $1$ any vertex that appears to the left
(viewed outwardly) of a space-like edge intersected by two loop arcs, and any other vertex in $V(T)$ with a $0$.
The vertices $x_0$ and $x$ are not labelled. As $(T,\delta)\in\widetilde{\mathcal T}_m(N)$,
this defines the correspondence $\tilde\psi$.
The bijectivity of $\tilde\psi$ is a consequence of the following facts: $\psi$ is bijective; 
disregarding $x_0$, the vertices of $L\in{\mathcal L}_m^{de}(N)$ and $T=\psi(C)$ coincide;
by Lemma \ref{lem:2m1}, $L$ is uniquely described by the intersection characteristics $\mathbf{n}_k$, $k=1,\ldots,m$;
$(T,\delta)$ is uniquely described by the labelling characteristics $\bm{\delta}_k$, $k=1,\ldots,m$; and
$\bm{\delta}_k=\mathbf{n}_k$ for $k=1,\ldots,m$.
The relation $|\delta|=s(L)/2$ readily follows.
\end{proof}

\noindent
The following proposition is the dilute counterpart to Proposition \ref{denseloop-trees}, see Figure \ref{figure:trexDi}.
\begin{prop}\label{diluteloop-trees} 
For every $m\in\mathbb{N}_0$ and $N\in\mathbb{N}$, there is a $2^m$ to $1$ correspondence
$$
 \hat\psi: {\mathcal L}_m^{di}(N)\to\widetilde{\mathcal T}_m^{ev}(N)
$$
such that if $(T,\delta)=\hat\psi(L)$ then $T=\psi(C)$, where $C$ is the triangulation underlying $L$. 
Moreover, $|\delta|=s(L)$.
\end{prop}
\begin{proof}
The map $\delta$ in the image $\hat\psi(L)=(T,\delta)$ labels with a $1$ any vertex that appears to the left
(viewed outwardly) of a space-like edge intersected by an arc, and any other vertex in $V(T)$ with a $0$.
The vertices $x_0$ and $x$ are not labelled. As $(T,\delta)\in\widetilde{\mathcal T}_m^{ev}(N)$,
this defines the correspondence $\hat\psi$.
The $2^m$ to $1$ property of $\hat\psi$ is a consequence of the following facts: $\psi$ is bijective; 
disregarding $x_0$, the vertices of $L\in{\mathcal L}_m^{di}(N)$ and $T=\psi(C)$ coincide;
by Lemma \ref{lem:2m1}, $L$ is one of $2^m$ possible loop configurations admitted by the admissible intersection characteristics $\mathbf{n}_k$, $k=1,\ldots,m$;
$(T,\delta)$ is uniquely described by the labelling characteristics $\bm{\delta}_k$, $k=1,\ldots,m$; and
$\bm{\delta}_k=\mathbf{n}_k$ for $k=1,\ldots,m$.
The relation $|\delta|=s(L)$ readily follows.
\end{proof}

\begin{figure}
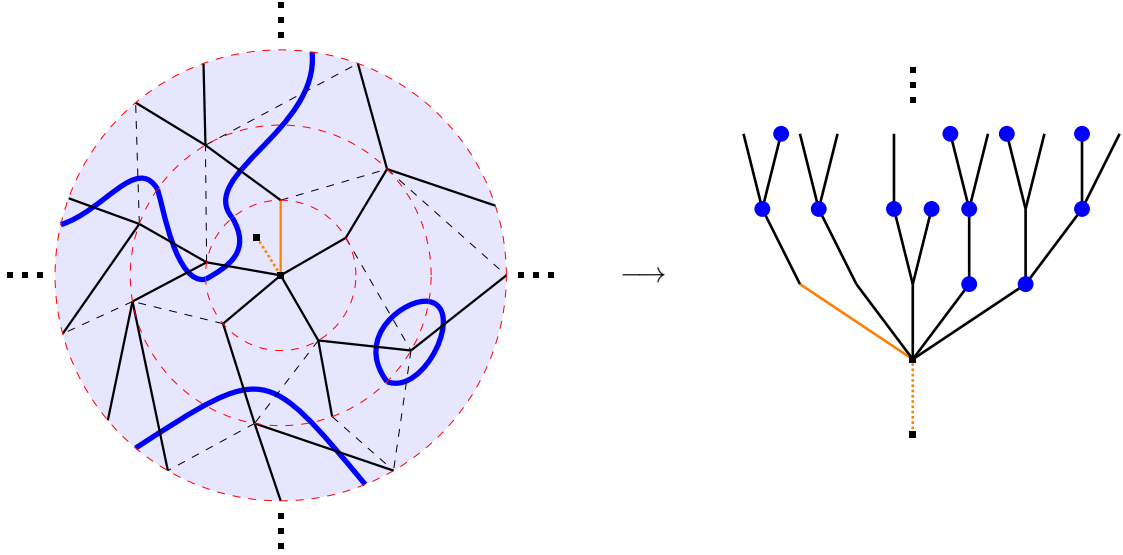

    \centering
    \begin{align*}
        \raisebox{-3.55cm}{\treexampleDashedDilute} 
        \qquad\longrightarrow\qquad \raisebox{-2.075cm}{\planartreeexampleDashedDilute}
    \end{align*}
    \captionof{figure}{Dilute loop configuration $L$ and the corresponding labelled planar tree $\hat\psi(L)$. Vertices labelled $1$ are indicated with blue circles.}
\label{figure:trexDi}
\end{figure}

\subsection{Partition functions}\label{sub:PF}

Following Propositions \ref{denseloop-trees} and \ref{diluteloop-trees}, we are interested in the tree ensembles
$\widetilde{\mathcal T}_m$ and $\widetilde{\mathcal T}_m^{ev}$, $m\in\mathbb{N}_0$. 
The corresponding partition functions are defined as
\begin{equation}
\label{LabTreePart1}
 W(g,\alpha):= \sum_{m=0}^{\infty} W_m(g,\alpha),\qquad
 W_m(g,\alpha):= \sum_{(T,\delta)\in\widetilde{\mathcal T}_m} g^{\TT-1}\alpha^{|\delta|},
\end{equation}

\begin{equation}
\label{LabTreePart2}
 W^{ev}(g,\alpha):= \sum_{m=0}^{\infty} W_m^{ev}(g,\alpha),\qquad
 W_m^{ev}(g,\alpha):= \sum_{(T,\delta)\in\widetilde{\mathcal T}_m^{ev}} g^{\TT-1}\alpha^{|\delta|},
\end{equation}

\noindent where $\TT$ denotes the number of edges in $T$.
The subtraction of $1$ from $\TT$ in the exponent of $g$ disregards the contribution from the edge $\{x_0,x\}$,
and we note that $W_0(g,\alpha) =W_0^{ev}(g,\alpha)=1$.
It follows from Proposition \ref{denseloop-trees} that 
\begin{equation*}
 Z_m^{de}(g,\alpha) = W_m(g,\alpha^2),
\end{equation*}
and consequently
\begin{equation}\label{LoopTree1}
 Z^{de}(g,\alpha) = W(g,\alpha^2).
\end{equation}
It likewise follows from Proposition \ref{diluteloop-trees} that
\begin{equation}\label{LoopTree2m}
   Z_m^{di}(g,\alpha) = 2^m W_m^{ev}(g,\alpha),
\end{equation}
and consequently
\begin{equation*}
   Z^{di}(g,\alpha) =\sum_{m=0}^\infty Z_m^{di}(g,\alpha)=\sum_{m=0}^\infty 2^m W_m^{ev}(g,\alpha).
\end{equation*}

\section{Tree partition function analysis}
\label{sec:4}

We now turn to an analysis of the tree partition functions $W(g,\alpha)$ and $W^{ev}(g,\alpha)$, 
with a view to determine the critical behaviour
of the dense and dilute loop models.
Regarding the dilute case, we consider the {\em height-coupled partition function} 
\begin{align}\label{equ:HCpfk}
    W^{ev} (g,\alpha, k) := \sum_{m=0}^{\infty}k^m W_m^{ev}(g,\alpha),
\end{align}
where $k > 0$ is an arbitrary height coupling, such that $Z^{di}(g,\alpha)$ is recovered for $k=2$:
$$
 Z^{di}(g,\alpha)=W^{ev}(g,\alpha, 2).
$$
In Section \ref{subsec:PT}, we analyse the partition 
functions $W(g,\alpha),~ W^{ev}(g,\alpha)$ and $W^{ev}(g,\alpha, k)$ for $\alpha=0$ corresponding to unlabelled planar trees, while in Section \ref{Sec:Lab} we consider the case where $\alpha \in (0,1]$.

\subsection{Planar trees}\label{subsec:PT}

Let $\mathcal{T}$ denote the set of planar trees with root of degree $1$ (hence of height at least $1$), that is,
$$
 {\mathcal T}:= \bigcup_{m=0}^\infty{\mathcal T}_m.
$$
The partition function, $W(g)$, for trees in $\mathcal{T}$, again disregarding the
contribution from the single edge emanating from the root, is defined by
\begin{align}\label{equ:zeroCoupWz}
    W(g):=\sum_{T\in{\mathcal T}} g^{\TT-1}=\sum_{m=0}^\infty W_m(g),\qquad
    W_m(g):=\sum_{T\in{\mathcal T}_m} g^{\TT-1},
\end{align}
where $W_m(g)$ is the similar partition function restricted to trees of height $m+1$.
It readily follows that
\begin{align*} 
    W(g,0)=W^{ev}(g,0)=W^{ev}(g,0,1)=W(g).
\end{align*}
As $W(g)$ is the generating function of rooted planar trees up to a factor of $\frac{1}{g}$,
it follows (see e.g.~\cite{drmota2009random})
that
\begin{equation} \label{gW2}
 W(g)=\frac{1}{1-gW(g)},
\end{equation}
from which one obtains
\begin{equation}\label{W}
 W(g) = \frac{1-\sqrt{1-4g}}{2g},
\end{equation}
which is analytic on the disk 
$$
 {\mathbb D} = \big\{g\in\mathbb C\,|\, |g|<\tfrac{1}{4}\big\},
$$ 
with a square-root singularity at $g_c=\frac14$. We note that $W(0)=1$ and $W(\frac{1}{4})=2$.

Specialising the labelling parameter in \eqref{equ:HCpfk} to $\alpha=0$, the expression reduces to a height-coupled partition function
generalising the tree partition function $W(g)$:
$$
    W^{ev}(g,0, k) = \sum_{m=0}^{\infty}k^m W_m(g),\qquad k>0.
$$

For $m\in\mathbb{N}_0$, let $X_m(g)$ be the partition function of trees in $\mathcal{T}$ with height at most $m+1$.
It follows that
$$
 W_m(g) = X_m(g)-X_{m-1}(g),
$$
where $X_{-1}(g)\equiv0$, and that
\begin{equation}\label{XX}
 X_m(g)=\frac{1}{1-gX_{m-1}(g)},\qquad m\geq0.
\end{equation}
This recursion relation is solved by
$$
 X_m(g)=\frac{U_m\big(\tfrac{1}{2\sqrt{g}}\big)}{\sqrt{g}\,U_{m+1}\big(\tfrac{1}{2\sqrt{g}}\big)},
$$
where $U_n(x)$ is the $n$'th Chebyshev polynomial of the second kind (with $U_{-1}(x)\equiv0$). It follows that
$$
 W_m(g)=\frac{1}{\sqrt{g}\,U_m\big(\tfrac{1}{2\sqrt{g}}\big)U_{m+1}\big(\tfrac{1}{2\sqrt{g}}\big)},
$$
and since
$$
 U_n(x)=2^n\prod_{j=1}^n\Big(x-\cos\frac{j\pi}{n+1}\Big),\qquad n\in\mathbb{N}_0,
$$
$X_m(g)$ is analytic on the disk
\begin{equation*} 
 \mathbb{D}_m = \{g\in\mathbb C\,|\,|g|< g_m \},\qquad
 g_m = \frac14\Big(1 + \tan^2\!\frac{\pi}{m+2}\Big),
\end{equation*}
with a simple pole at $g_m$ (noting that $g_0=\infty$).

Because the coefficients defining $X_m(g)$ as a power series in $g$ are nonnegative, it follows that $X_m(g)$ is divergent 
for $g>g_m$. Since $g_m\to \frac 14$ as $m\to\infty$, this implies that $W(g,0, k)$ is divergent for $g>\frac 14$, so the radius of 
convergence $g_c(k)$ of $W(g,0, k)$ for fixed $k$ is at most $\frac14$ for any $k>0$. 
The following lemma is a step towards the determination of 
$g_c(k)$ for all $k>0$, completed in Proposition \ref{critk} below.

\begin{lem}\label{TreeCrit}
For $g\in(0,\frac14)$ and $m\in\mathbb{N}$, we have 
$$
    \frac{\phi^2(W(g)-1)}{W(g)}\big(W(g)-1\big)^{m} < W_m(g) <\big(W(g)-1\big)^{m},
$$
where $\phi$ is Euler's function.
\end{lem}
\begin{proof} 
From the recursion relation \eqref{XX}, we get 
$$
	W_m = X_m - X_{m-1}  
	= \frac{g(X_{m-1}-X_{m-2})}{(1-gX_{m-1})(1-gX_{m-2})}
	=g(X_{m-1}-X_{m-2})X_{m}X_{m-1},
$$
and iteratively applying this, we obtain (for $m\geq1$)
\begin{equation}\label{eq:FinalBoundExpr} 
	W_m=g^mW^{2m-2}X_m\prod_{i=1}^{m-1}\Big(\frac{X_i}{W}\Big)^2.
\end{equation}
Since $X_i$ is increasing and converges to $W$ as $i\to\infty$, it follows that 
$1<X_i<W$ for $i\geq1$, so $\frac{X_i}{W} <1$, allowing us to use (\ref{gW2}) to conclude that
$$
   W_m 
   <\big(W-1\big)^{m}.
$$
This establishes the upper bound.

For the lower bound, we have
$$
 W-X_m=g(W-X_{m-1})WX_m
 =(W-1)g^mW^{2m}\prod_{i=1}^{m}\frac{X_i}{W}
 <(W-1)^m,
$$
where we have used  (\ref{gW2}) and that $1<W<2$.
It follows that
$$
 \frac{X_i}{W}= 1- \frac{W-X_i}{W}>1- \frac{(W-1)^{i}}{W} >1-\big(W-1\big)^i.
$$
Since $0<W-1<1$, we have
$$
 \prod_{i=1}^m\big(1- (W-1)^i\big)>\prod_{i=1}^\infty \big(1-(W-1)^i\big) = \phi(W-1).
$$
Returning to (\ref{eq:FinalBoundExpr}), we thus find that
$$
 W_m=\frac{g^mW^{2m}}{X_m}\prod_{i=1}^{m}\Big(\frac{X_i}{W}\Big)^2>\frac{\phi^2(W-1)}{W}(W-1)^{m},
$$
thereby establishing the lower bound and completing the proof.
\end{proof}

\noindent The following corollary immediately follows.
\begin{cor}\label{kTreeCrit} 
For $k>0$ and $g\in(0,\frac14)$, we have
$$
    \frac{\phi^2(W(g)-1)}{W(g)}\sum_{m=0}^{\infty}\big(k(W(g)-1)\big)^m 
    \leq W(g,0,k)
    \leq \sum_{m=0}^{\infty}\big(k(W(g)-1)\big)^m.
$$
\end{cor}

\begin{prop}
\label{critk} 
The critical coupling for $W(g,0,k)$ is given by
$$
g_c(k) = \begin{cases} \frac14,\quad &k\in(0,1],\\[.2cm] 
 \frac{k}{(k+1)^2},\quad &k\in(1,\infty).\end{cases}
$$
\end{prop}
\begin{proof}
As mentioned previously, we have $g_c(k) \leq \frac 14$. For $k\leq 1$, we clearly have $W(g,k) \leq W(g)$, so $g_c(k) = \frac 14$.
For $k>1$, Corollary \ref{kTreeCrit} implies that $W(g,0,k)$ is finite if and only if 
$$
 k\big(W(g)-1\big)<1.
$$ 
It follows that $g_c(k)=\frac{k}{(k+1)^2}$.
\end{proof}

\noindent
For $k<1$, it can be shown that all $g$-derivatives of $W(g,0,k)$ from 
the left at $g=\frac 14$ are finite, but this fact is not essential for the discussion in Section \ref{sec:5}. 
For $k=1$, we have $W(g,0,1)=W(g)$, whose singular nature was recalled above. 
For $k>1$, it follows from Corollary \ref{kTreeCrit} and the analyticity of $W(g)$ at $g_c(k)<\frac14$ that 
$W(g,0,k)$ has a simple pole at $g_c(k)$. 
We refrain from giving the detailed arguments here (see \cite{UnelInPrep1}) 
but will assume their validity in the following.

\subsection{Labelled planar trees}
\label{Sec:Lab}

We now consider the presence of a labelling with the associated coupling $\alpha\in(0,1]$. 
Summing over the labels $\delta$ in the expression \eqref{LabTreePart1} for $W_m(g,\alpha)$ yields
\begin{equation}\label{equ:W1Expans}
    W_m(g,\alpha) = \sum_{T\in{\mathcal T}_m}\big(g(1+\alpha)\big)^{\TT-1}=W_m(g(1+\alpha)),
\end{equation}
where we have used \eqref{equ:zeroCoupWz}. It follows that
\begin{equation}\label{Wde}
 W(g,\alpha)=W(g(1+\alpha)),
\end{equation}
whose critical $g$-coupling is given by $g=\frac{1}{4(1+\alpha)}$.

Likewise, summing over the admissible labels in the expression \eqref{LabTreePart2} for $W_m^{ev}(g,\alpha)$ yields
\begin{align}
\label{Wmeven}
 W_m^{ev}(g,\alpha) &= \sum_{T\in{\mathcal T}_m} g^{\TT-1} \prod_{i=1}^{m}\tfrac12 \big[(1+\alpha)^{n_i} + (1-\alpha)^{n_i}\big],
\end{align}   
where $n_i=|V_i(T)|$ is the number of vertices in $T$ at height $i+1$. Note that $W^{ev}(g,\alpha)$ 
naturally extends to an even function of $\alpha\in[-1,1]$.
Contrary to the case $W_m(g,\alpha)$ in (\ref{equ:W1Expans}),
we do not have a closed-form expression for $W_m^{ev}(g,\alpha)$. Instead, we may gain insight into the critical 
behaviour of $W_m^{ev}(g,\alpha)$ by bounding the function by (partition) functions whose critical behaviour is understood. 
For $\alpha\in(0,1]$, we thus use
$$
 (1+\alpha)^n\leq (1+\alpha)^n+ (1-\alpha)^n \leq 2(1+\alpha)^{n-1},\qquad n\in\mathbb{N},
$$
and
$$
 \sum_{i=1}^mn_i=\TT-1,\qquad T\in\mathcal{T}_m,
$$
to obtain
\begin{equation}\label{primbound}
 W^{ev}\big(g(1+\alpha),0,\tfrac{k}{2}\big)  
 \leq W^{ev} (g,\alpha,k)
 \leq W^{ev} \big(g(1+\alpha),0,\tfrac{k}{1+\alpha}\big).
\end{equation}
We will use this in our analysis of $Z^{di}(g,\alpha)$ in Section \ref{Sec:Dilute5}.

\section{Critical behaviour of loop models}
\label{sec:5}

In this section, we use results from Sections \ref{sec:3} and \ref{sec:4} together with a transfer matrix formalism to examine the critical behaviour of the dense and dilute loop models. 
This approach is first applied to the pure CDT model and readily extended to the dense loop model by shifting the coupling constant. 
It is subsequently used to gain insight into the dilute loop model. 
Focus here is on the critical behaviour of the models; algebraic aspects of the transfer matrix formalism 
will be discussed elsewhere.

\subsection{Dense loop model}\label{subsub:CritDe}

It follows from \eqref{LoopTree1}, \eqref{W} and \eqref{Wde} that 
\begin{equation}\label{Zde} 
 Z^{de}(g,\alpha)=W(g(1+\alpha^2))=\frac{1- \sqrt{1-4g(1+\alpha^2)}}{2g(1+\alpha^2)} .
\end{equation}
This expression implies that the critical coupling for the dense loop model, $g_c^{de}(\alpha)$, 
and the corresponding value of the partition function, $Z_c^{de}(\alpha):=Z^{de}(g_c^{de}(\alpha),\alpha)$, are given by
$$
    g_c^{de}(\alpha) = \frac{1}{4(1+\alpha^2)}, \quad\quad\quad Z_c^{de}(\alpha)=2.
$$

\noindent
\textbf{Remark}. 
The correspondence between causal triangulations on the sphere and rooted trees can be viewed as a particular instance of Schaeffer's bijection \cite{schaefferbijection}, where each pair of triangles sharing a space-like edge form a quadrangle \cite{durhuus2010spectral}. In the dense loop model, such a quadrangle can then be decorated in two ways: (i) with two vertical arcs (both intersecting the shared space-like edge) or (ii) with two horizontal arcs (not intersecting the space-like edge). We emphasise that there are no compatibility constraints between the decorated quadrangles. The expression \eqref{Zde} now follows by assigning the weight $g\alpha^2$ to the quadrangle in case (i) and $g$ in case (ii).\\

It also follows that the behaviour of $Z^{de}(g,\alpha)$ near the critical point is the same as for pure CDT:
\begin{align}\label{deCrit2}
 Z^{de}(g,\alpha) \sim Z_c^{de}(\alpha)- c_\alpha\sqrt{g_c^{de}(\alpha) - g},
\end{align}
where $c_\alpha=4\sqrt{1+\alpha^2}$ is $g$-independent.

By differentiating the expression (\ref{equ:zeroCoupWz}) for $W(g)$, one can express the average value of $|T|$ for $T\in\mathcal{T}$ 
as
$$
 \frac{1}{W(g)}\sum_{T\in\mathcal{T}}|T|\,g^{|T|-1}=1+\frac{gW'(g)}{W(g)}=\frac{2g}{\sqrt{1-4g}\big(1-\sqrt{1-4g}\big)},
$$
which is seen to have a square-root divergence at $g_c=\frac{1}{4}$.
This reflects the fact that, close to the critical point, 
the large triangulations yield the dominant contribution to $W(g)$. A precise definition of the limiting distribution of large 
triangulations in the form of a measure on causal triangulations of infinite size (or radius) is given in \cite{durhuus2010spectral}. For such 
infinite triangulations $C$ with central vertex $x$, the ball of radius $R$ around $x$ is defined as 
$$
 B(x,R):= \{v\in V(C)\,|\, d_C(x,v)\leq R\},
$$
where $d_C$ denotes the graph distance on $C$. The Hausdorff dimension $d_H$ of $C$ is then defined as the polynomial growth 
rate of the number of vertices, $|B(x,R)|$, as a function of $R$, that is,
\begin{equation}\label{equ:hausdorff}
 d_H(C):= \lim_{R\to\infty} \frac{\ln |B(x,R)|}{\ln R},
\end{equation}
provided the limit exists. For pure CDT, it is shown in \cite{durhuus2010spectral} that $d_H=2$ almost surely. 

As a consequence of (\ref{Zde}) and \eqref{deCrit2}, the same arguments apply to the present case, 
characterised by $Z^{de}(g,\alpha)$ with critical coupling $g_c^{de}(\alpha)$,
and yield the same value for the Hausdorff dimension:
$$
 d_H^{de}=2\quad(\mathrm{almost\ surely}).
$$
Hence, the coupling of the dense loop model to CDT does not influence the statistical behaviour of the underlying triangulations. 
This is analogous to what is seen for the Ising model coupled to a random planar tree, where a relation similar to \eqref{deCrit2} can be 
derived \cite{durhuus2012generic}.

\subsection{Pure CDT transfer matrix}
\label{sec:5.2}

To analyse the critical behaviour of the dilute loop model in Section \ref{Sec:Dilute5},
we will combine results on the corresponding labelled tree model obtained in 
Sections \ref{sec:3} and \ref{sec:4} with insight gained by applying a transfer matrix formalism.
It is convenient to develop first the similar approach in the pure CDT model, where transfer matrices have been previously discussed 
in \cite{malyshev2001two,hernandez2013ising,di2000integrable} and implicitly already in \cite{ambjorn1998non}. 
Our focus and methods for investigating the transfer matrix eigenvalues are, however, different and do not depend on the explicit 
evaluation of the eigenvalues. 

The transfer matrix formalism views the triangulated disk as a concatenation of triangulated annuli. 
A compatibility condition
is then enforced along the shared boundaries where two annuli are concatenated. In the pure CDT case, the only constraint
is that the boundary lengths agree, that is, the matching boundaries must contain the same number of space-like edges.
As the boundary length can take on any positive integer value, it is natural to let the Hilbert space of
square summable complex sequences,
$$
 l_2(\mathbb{N}):=\big\{(x_n)_{n\in\mathbb{N}}\,|\,\sum_{n=1}^\infty|x_n|^2< \infty;\,x_i\in\mathbb{C},\,i\in\mathbb{N}\big\},
$$
encode the degrees of freedom along the space-like boundaries. Here, $n$ labels the length of a space-like boundary component of a triangulated annulus. 
We shall view the transfer matrix $\tm$, whose matrix elements $\tm_{r,s}$ with respect to the standard orthonormal basis of $l_2(\mathbb{N})$ are given in \eqref{matr-el},  as an operator on $l_2(\mathbb{N})$. In the following, we use Dirac notation where $|w\rangle$ is a sequence in $l_2(\mathbb{N})$ with coordinates 
$w_n$, $n\in\mathbb{N}$.

To take into account the distinguished vertex on each cycle, we only consider annulus triangulations where
the first triangle (relative to the distinguished vertex) is forward-directed.
With $s$ and $r$ denoting the number of lower, respectively upper, space-like edges of the annulus, we set
\begin{equation}\label{matr-el}
  \tm_{r,s}(g):=\binom{r+s-1}{r}g^{\frac{r+s}{2}},
\end{equation}
where the combinatorial factor counts the number of ways the $r$ backward-directed and $(s-1)$ forward-directed 
triangles can be concatenated to the right (viewed outwardly) of the distinguished forward-directed triangle.
As in \cite{ambjorn1998non,malyshev2001two}, the weight $g^{\frac{r+s}{2}}$ encodes that a weight $g^{\frac{1}{2}}$ has been attributed to each space-like edge 
of the annulus, ensuring that a shared space-like edge between a pair of concatenated annuli is assigned the familiar weight $g$. 

Although $\tm_{r,s}(g)$ is not symmetric in $r$ and $s$, the operator $\tm(g)$ is {\em symmetrisable}, admitting a factorisation in terms of a 
diagonal operator $\dm$ and a symmetric operator $\km(g)$,
$$
    \tm(g) =\dm\km(g),
$$
whose matrix elements are given by
$$
    \dm_{r,s} := \frac{\delta_{r,s}}{r}, \qquad \km_{r,s}(g) 
    :=\frac{(r+s-1)!}{(r-1)!(s-1)!}\,g^{\frac{r+s}{2}}.
$$

\begin{prop} The operator $\km(g)$ is trace-class for $g\in \mathbb D$, it is positive definite for $g\in(0,\frac 14)$, 
and the function $h\mapsto\km(h^2)$ is analytic on $\{h\in\mathbb{C}\,|\,|h|<\frac{1}{2}\}$.
\end{prop} 

\begin{proof} Note first that, for fixed $s$, the sequence $(\km_{r,s}(g))_{r\in\mathbb N}$ is square summable provided $|g|<1$. 
Hence, the operator $\km(g)$ is well defined for $|g|<1$ on the dense subspace $V$ of sequences with only finitely many non-vanishing entries.  
For $n\in\mathbb N$, let $P_n$ be the orthogonal projection in $l_2(\mathbb N)$ onto the subspace spanned by vectors $|w\rangle$ whose entries $w_r$ vanish for $r>n$. 
Then, the operator $K_n(g):=P_n\km(g)P_n$ is of finite rank and hence is bounded on $l_2(\mathbb N)$, with matrix elements 
$$
 K_n(g)_{r,s}=\begin{cases}  \km_{r,s}(g),&\quad r,s\leq n,\\[.15cm]
  0,&\quad\text{otherwise}.\end{cases}
$$
Moreover, $K_n(g)$ is positive semidefinite for $g\in[0,1)$  since the identity
$$ 
\frac{(r+s-1)!}{(r-1)!(s-1)!} = \sum_{k\geq 1} k\binom{s}{k}\binom{r}{k}
$$
implies 
\begin{equation}\label{wKw}
\langle w | K_n(g)| w \rangle = \sum_{k=1}^n k \big|\sum_{s=k}^n \binom{s}{k} w_s g^{\frac{s}{2} }\big|^2 \geq 0.
\end{equation} 
Its trace norm is given by 
\begin{equation}
\label{Kn1}
\Vert K_n(g)\Vert_1 = \mathrm{tr} K_n(g)
 =\sum_{s=1}^n \sum_{k=1}^s s \binom{s}{k}\binom{s-1}{k-1} g^s
 \leq\sum_{s=1}^\infty s\binom{2s-1}{s} g^s
 =\frac{g}{(1-4g)^{\frac 32}},
\end{equation}
where the last equality is valid for $g\in(0,\frac 14)$.
It follows, in particular, that the operator norms $\Vert K_n(g)\Vert$ are uniformly bounded in $n$ for any fixed $g\in[0,\frac 14)$. Since $\langle v| K_n(g)| w\rangle \to \langle v| \km(g)| w\rangle$ as $n\to\infty$ for all $v,w\in V$, it follows that $\km(g)$ extends to a bounded operator on $l_2(\mathbb N)$ that equals the weak limit of $(K_n(g))_{n\in\mathbb N}$ for $g$ in this interval. Due to \eqref{Kn1}, it now follows from Theorem 10 in Section 2.4 of \cite{gelfand1964generalized} that $\km(g)$ is trace-class with trace norm given by 
$$
\Vert \km(g)\Vert_1 = \mathrm{tr} \,\km(g) = \frac{g}{(1-4g)^{\frac 32}}\,.
$$
Letting $n\to\infty$ in \eqref{wKw}, it follows that $\langle w| \km(g)| w\rangle$  
 vanishes if and only if 
$$
\sum_{s=k}^\infty \binom{s}{k} w_s g^{\frac{s}{2}} =0
$$
for all $k\geq 1$. Evidently, this sum is convergent for $|g|<1$ and 
$$
 \sum_{s=k}^{\infty} \binom{s}{k} w_s g^{\frac{s}{2}}=\frac{g^{\frac{k}{2}}}{k!}\,f^{(k)}\big(g^{\frac{1}{2}}\big),
$$
where $f^{(k)}(g^{\frac{1}{2}})$ is the $k$'th derivative of the analytic function
$$
 f(z) = \sum_{s=1}^{\infty} w_s z^s,\qquad |z|<1,
$$
evaluated at $z=g^{\frac 12}$. It follows that, for $g\neq 0$, $\langle w| \km(g)| w\rangle$  
 vanishes if and only 
if $f(z)=0$, that is, $w_s=0$ for all $s$. This shows that $\km(g)$ is positive definite.

Let now $g\in\mathbb D$ be arbitrary and write $g= |g|e^{i\theta}$. Then,
$$
\km(g) = U(\theta)\km(|g|)U(\theta),
$$
where $U(\theta)$ is the diagonal unitary operator with matrix elements 
$$
U_{r,s}(\theta) = e^{\frac{i r\theta}{2}}\delta_{r,s}\,.
$$
It follows that $\km(g)$ is well-defined and trace-class for all $g\in\mathbb D$, with 
$$
\Vert \km(g)\Vert_1 = \mathrm{tr} \,|\km(g)| = \Vert \km(|g|)\Vert_1\,.
$$
It is evident that $g \to \langle v | \km(g)| w\rangle$ is analytic on $\mathbb D$ for all $v,w\in V$, so $\km(g)$ is 
analytic on $\mathbb D$ as a function of $\sqrt{g}$, in the sense of Kato (see \cite{reed1978iv} Section XII.2).   
\end{proof}

\noindent Since $\dm$ is a bounded operator, it follows that the transfer operator $\tm(g)$ is trace-class for $g\in(0,\frac{1}{4})$ with
$$
    \mathrm{tr}\big(\tm(g)\big) 
    =\sum_{s=1}^{\infty} \tm_{s,s}(g) 
    =\frac{1-\sqrt{1-4g}}{2\sqrt{1-4g}}.
$$
Since the matrix elements of $\km(g)$ are positive for $g\in (0,\frac{1}{4})$, it is a consequence of the Perron-Frobenius theorem 
(see e.g.~Theorem XIII.43 in \cite{reed1978iv}) that $\km(g)$ has a non-degenerate positive largest eigenvalue, which equals the 
operator norm $\Vert\km(g)\Vert$, and the corresponding normalised eigenvector has positive entries only. 
By the Kato-Rellich theorem, see Theorem XII.8 in \cite{reed1978iv}, this eigenvalue is an analytic function of $g$ on $(0, \frac{1}{4})$,
and one can multiply the corresponding normalised eigenvector by a phase such that the entries of the product are analytic functions of $g$ on $(0, \frac{1}{4})$.
Since $\dm $ is diagonal and positive definite, these statements also hold for $\dm^{\half} \km(g) \dm^{\half}.$

\subsection{Analysis of pure CDT model}
\label{sec:5.2a}

For $m\neq0$, the fixed-height pure CDT partition function is given by $Z_m(g)=W_m(g)$ and is expressible as a matrix element:
\begin{equation}\label{matr-el2}
 Z_m(g) = \big\langle v(g)\big| \tm^{m-1}(g)\big| v(g)\big\rangle,
 \qquad m\in\mathbb{N},
\end{equation}
where $\tm^{m-1}_{r,s}(g)\equiv\big(\tm^{m-1}(g)\big)_{r,s}$ and where $v(g)\in l_2(\mathbb N)$ has entries
$$
 v_n(g):= g^{\frac n2},\qquad n\in\mathbb{N}.
$$
This expression for $Z_m(g)$ corresponds to a sum over all possible causal sphere triangulations of height $m$,
as depicted on the left in Figure \ref{figure: causalcylinder}.

We now define the {\em time-periodic partition functions} by 
\begin{equation}\label{PerPart} 
 Z_m^{per}(g) 
 :=\mathrm{tr}\big(\tm^{m-1}(g)\big), \qquad m\ge 2.
\end{equation}
Note that $Z_m^{per}(g)$ can be expressed as a sum over triangulations of 
a \emph{cylinder} of height $(m-1)$ whose boundary cycles, $S_1$ and $S_m$, have equal (but arbitrary) length, 
and with a weight $g^{\frac{1}{2}}$ associated to each elementary triangle. 
Alternatively, the weights can be assigned to the space-like edges; however, this requires the space-like edges
within the bulk be assigned the weight $g$ and those on the boundary cycles the weight $g^{\frac{1}{2}}$.

\begin{figure}
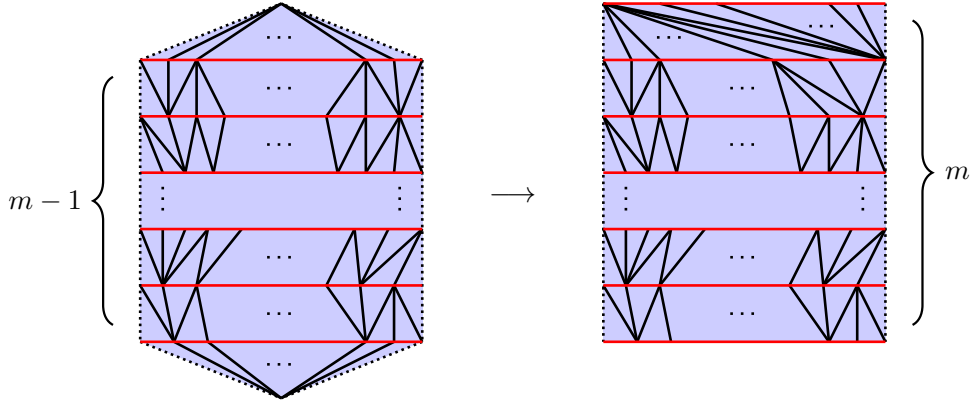

    \centering
    \begin{align*}
        \scalebox{1.5}{\raisebox{-1.75cm}{\triangulationexamplebirAlt}}
        \quad\quad
        \longrightarrow
        \quad\quad
        \scalebox{1.5}{\raisebox{-1.25cm}{\triangulationexampleikiAlt}}
    \end{align*}
    \caption{Triangulation contributing to $Z_m(g)$ and the corresponding triangulation contributing to $Z^{per}_{m+1}(g)$. 
    In both figures, the dashed edges indicate that the leftmost and rightmost time-like edges are identified.}
    \label{figure: causalcylinder}
\end{figure}

Since the trace of a positive trace-class operator $\om$ equals the sum of its eigenvalues \cite{gelfand1964generalized},
it follows that $\mathrm{tr}(\om^n) \leq(\mathrm{tr}\,\om)^n$, $n\in\mathbb{N}$,
and hence that
$$
 Z_m^{per}(g)\leq\big(\mathrm{tr}\,\tm(g)\big)^{m-1}.
$$
\begin{lem} \label{lem:Z-Zper}
For each $m\geq 1$ and $g \in [0,\frac{1}{4})$,
\begin{equation}\label{Z-Zper}
 Z_m(g)\leq Z_{m+1}^{per}(g).
\end{equation}
\end{lem}
\begin{proof}
Let $\bar{C}$ be a causal sphere triangulation of height $m$, with notation as in Section \ref{subsec:CT}. 
Then, (i) remove the $|S_1|$ backward-directed triangles with space-like edges in $S_1$, 
and (ii) place them to the right (viewed upwardly) of the $|S_m|$ forward-directed triangles with space-like edges in $S_m$, thereby creating an annulus $A_m$
between $S_m$ and the new outer boundary $S_{m+1}$. This is illustrated in Figure \ref{figure: causalcylinder} and
yields a triangulation, $\bar{C}'$, of a cylinder of height $m$, of the form contributing to $Z_{m+1}^{per}(g)$.
In fact, as a consequence of the conservation of triangulation area under $\bar{C}\to \bar{C}'$, 
the contribution of $\bar{C}'$ to $Z_{m+1}^{per}(g)$ is the same as the contribution of $\bar{C}$ to $Z_{m}(g)$. 
Finally, the map $\bar{C}\to \bar{C}'$ is evidently injective but not surjective, since the top annulus of $\bar{C}'$ is restricted to be of standard form. 
This establishes the inequality (\ref{Z-Zper}).
\end{proof}

\noindent
It follows from the cyclicity of the trace in (\ref{PerPart}) that 
$$
 Z_m^{per}(g) =\mathrm{tr}\big((\dm^{\frac 12}\km(g)\dm^{\frac 12})^{m-1}\big), \qquad m\ge 2.
$$
The following proposition describes the behaviour of  the largest eigenvalue as $g$ approaches the critical point. Although it does not provide any additional information for understanding the partition function of pure CDT, the method is general and will be used for the analysis of the dilute loop model in Section \ref{Sec:Dilute5}.
\begin{prop}\label{lambdaCrit1}
The largest eigenvalue $\lambda_1(g)$ of \,$\dm^{\frac 12}\km(g)\dm^{\frac 12}$\! satisfies
$$ 
\lambda_1(g) \upto 1 \qquad \mbox{\rm as}\quad g\upto \tfrac 14.
$$
\end{prop}

\begin{proof}
All matrix elements of $\dm^{\frac 12}\km(g)\dm^{\frac 12}$ are strictly increasing positive functions of $g$. 
By the variational principle (see Theorem XIII.1 in \cite{reed1978iv}), it follows that $\lambda_1(g)$ is strictly increasing for $g\in(0,\frac{1}{4})$.
We also note that $\lim_{g\downto0}\lambda_1(g)=0$.
If unlikely to cause confusion, arguments of functions may be omitted in the following. 

If $\lambda_1\leq c_1$ for some constant $c_1<1$, then
$\big(1-k\,\dm^{\frac 12}\km\dm^{\frac 12}\big)^{-1}$ is a bounded operator for $k\in[1,c_1^{-1})$ and
\begin{align*} 
 \sum_{m=1}^{\infty} k^m Z_{m+1}^{per}
 &=\sum_{m=1}^\infty \mathrm{tr}\Big(\big(k\,\dm^{\frac 12}\km\dm^{\frac 12}\big)^m\Big)
 \nonumber\\
 &=k\,\mathrm{tr} \Big(\dm^{\frac 12}\km\dm^{\frac 12}\big(1-k\,\dm^{\frac 12}\km\dm^{\frac 12}\big)^{-1}\Big)
 \nonumber\\
 &\leq k\,\big\Vert\big(1-k\,\dm^{\frac 12}\km\dm^{\frac 12}\big)^{-1}\big\Vert\,\mathrm{tr}(\tm)
 \nonumber\\[.15cm]
 &<\infty
\end{align*} 
for all $g\in(0,\frac 14)$. By \eqref{Z-Zper}, this implies
$$
 W(g,0,k)=1+\sum_{m=1}^{\infty} k^m Z_m(g) <\infty,
$$
in contradiction to Proposition \ref{critk}, which says that the critical coupling for $W(g,0,k)$ satisfies $g_c(k)<\frac 14$ for $k>1$. 
This shows that $\lim_{g\upto\frac 14}\lambda_1(g) \geq 1$. 

Now, suppose this limit is {\em strictly} greater than $1$. Then, there would exist $g_0<\frac 14$ such that $\lambda_1(g_0)=1$.
Let $|w^{(1)}(g)\rangle$ denote the normalised (analytic) eigenvector of \,$\dm^{\frac 12}\km\dm^{\frac 12}$\, with positive coordinates for $g$ in an interval $I$ around $g_0$. 
For $I$ sufficiently small, there exists a constant $c>0$ such that  
$$
 \big\langle v\big|\dm^{\frac 12}\big|w^{(1)}\big\rangle\big\langle w^{(1)}\big|\dm^{-\frac 12}\big|v\big\rangle \geq c
$$
for $g \in I$. As above, we suppress the $g$ dependence and denote $|v(g)\rangle$ 
and $|w^{(1)}(g)\rangle$ by $| v \rangle$ and $|w^{(1)}\rangle$, respectively. Let $\lambda_2\geq\lambda_3\geq\ldots$ denote the remaining eigenvalues 
(multiplicities included) of $\dm^{\frac 12}\km\dm^{\frac 12}$
and $\{|w^{(n)}\rangle\,|\,n\in\mathbb{N}\}$ a corresponding orthonormal set of eigenvectors.
For $g<g_0$, the partition function satisfies
\begin{align*}
    Z(g)-1&=\sum_{m=1}^{\infty}\big\langle v\big| \dm^{\frac 12}\big(\dm^{\frac 12} \km\dm^{\frac 12}\big)^{m-1}\dm^{-\frac 12}\big| v\big\rangle
    \\[.1cm]
    &=\sum_{m,n= 1}^{\infty}\big\langle v\big| \dm^{\frac 12}\big(\dm^{\frac 12} \km\dm^{\frac 12}\big)^{m-1}\big|w^{(n)}\big\rangle 
      \big\langle w^{(n)}\big|  \dm^{-\frac 12}\big| v\big\rangle
    \\[.1cm]
    &= \big\langle v\big| \dm^{\frac 12}\big|w^{(1)}\big\rangle\big\langle w^{(1)}\big| \dm^{-\frac 12}\big|v\big\rangle\sum_{m=1}^{\infty} \lambda_1^{m-1} 
      +\sum_{m=1}^{\infty}\sum_{n=2}^{\infty} \lambda_n^{m-1}\big\langle v\big| \dm^{\frac 12}\big|w^{(n)}\big\rangle\big\langle w^{(n)}\big|\dm^{-\frac 12}\big|v\big\rangle 
    \\[.1cm]
    &\geq\frac{c}{1-\lambda_1} + \sum_{n=2}^{\infty} \frac{\big\langle v\big|\dm^{\frac 12} \big|w^{(n)}\big\rangle\big\langle w^{(n)}\big| \dm^{-\frac 12}\big|v\big\rangle}{1-\lambda_n},
\end{align*}
and it follows that
\begin{align}\label{divcalc}
    Z(g)&\geq 1+\frac{c}{1-\lambda_1} 
    - \sum_{n=2}^{\infty}\frac{\big|\big\langle v\big| \dm^{\frac 12} \big|w^{(n)}\big\rangle\big\langle w^{(n)}\big| \dm^{-\frac 12}\big|v\big\rangle\big|}{1-\lambda_n} 
   \geq1+\frac{c}{1-\lambda_1} - \frac{\big\Vert \dm^{\frac 12} v\big\Vert\big\Vert\dm^{-\frac 12}v\big\Vert }{1-\lambda_2}. 
\end{align}
Note that, although $\dm^{-\frac12}$ is unbounded, $\dm^{-\frac 12} v$ belongs to $l_2(\mathbb{N})$, as the entries of $v(g)$ 
decay exponentially. It follows from \eqref{divcalc} that $Z(g)$ diverges as $g$ approaches $g_0$ from below, in contradiction to the 
fact that $Z(g)$ is analytic on $\mathbb D$. 
It follows that $\lim_{g\upto\frac 14}\lambda_1(g)=1$, and since $\lambda_1(g)$ is strictly increasing for $g\in(0,\frac{1}{4})$, this concludes the proof.
\end{proof}

\noindent
This completes our transfer matrix analysis of the pure CDT model.

\subsection{Dense loop model revisited}
\label{Sec:DenseRev}

Generalising (\ref{Zde}), the fixed-height dense loop model partition function is expressible in terms of the similar pure CDT partition function as
$$
 Z_m^{de}(g,\alpha)=Z_m\big(g(1+\alpha^2)\big).
$$ 
Analogues of the pure CDT transfer matrix $\tm(g)$ and the related operator $\km(g)$ similarly satisfy 
$\tm^{de}(g,\alpha)=\tm(g(1+\alpha^2))$ and $\dm^{\frac 12}\km^{de}(g,\alpha)\dm^{\frac 12}=\dm^{\frac 12}\km(g(1+\alpha^2))\dm^{\frac 12}$.
Both $\tm^{de}(g,\alpha)$ and $\km^{de}(g,\alpha)$ are positive definite and trace-class for $g\in(0,\frac{1}{4(1+\alpha^2)})$, with
$$
     \mathrm{tr}\big(\tm^{de}(g,\alpha)\big) 
    =\frac{1-\sqrt{1-4g(1+\alpha^2)}}{2\sqrt{1-4g(1+\alpha^2)}},\qquad
     \mathrm{tr}\big(\km^{de}(g,\alpha)\big) 
    =\frac{g(1+\alpha^2)}{\big(1-4g(1+\alpha^2)\big)^{\frac{3}{2}}}.
$$
Moreover, as counterpart to Proposition \ref{lambdaCrit1}, we have that, for $\alpha\in[0,1]$, 
the largest eigenvalue $\lambda_1(g,\alpha)$ of $\dm^{\frac 12}\km^{de}(g,\alpha)\dm^{\frac 12}$
approaches $1$ from below as $g$ approaches $g_c ^{de}{(\alpha)}=\frac{1}{4(1+\alpha^2)}$ from below.

\subsection{Dilute loop model}
\label{Sec:Dilute5}

Recalling \eqref{LoopTree2m} and \eqref{Wmeven}, the transfer matrix for the dilute loop model has entries given by
$$
 \tm^{di}_{r,s}(g,\alpha)= \binom{r+s-1}{r}
  \big[(1+\alpha)^{r}+(1-\alpha)^{r}\big]^{\frac 12} \big[(1+\alpha)^{s}+(1-\alpha)^{s}\big]^{\frac 12} g^{\frac{r+s}{2}},
$$
normalised such that, in analogy with \eqref{matr-el2},
\begin{equation*}
 Z_m^{di}(g,\alpha)
 =\big\langle v(g,\alpha)\big|\big(\tm^{di}(g,\alpha)\big)^{m-1}\big|v(g,\alpha)\big\rangle,\qquad m\in\mathbb{N},
\end{equation*}
where the coordinates of the vector $v(g,\alpha)\in l^2(\mathbb{N})$ are defined by 
$$
 v_n(g,\alpha) := \big[(1+\alpha)^n+(1-\alpha)^n\big]^{\frac{1}{2}}g^{\frac{n}{2}},\qquad n\in\mathbb{N}.
$$
Recalling that $Z^{di}_0(g,\alpha)=1$, the full partition function is thus given by
\begin{equation*} 
 Z^{di}(g,\alpha)=1+\sum_{m=1}^\infty\big\langle v(g,\alpha)\big|\big(\tm^{di}(g,\alpha)\big)^{m-1}\big|v(g,\alpha)\big\rangle.
\end{equation*}

The transfer matrix $\tm^{di}(g,\alpha)$ is not symmetric but can be symmetrised by the same diagonal operator as $\tm(g)$, that is, 
$$
 \tm^{di}(g,\alpha)=2\dm\km^{di}(g,\alpha),
$$
where
$$
 \km_{r,s}^{di}(g,\alpha) =\frac{1}{2}\frac{(r+s-1)!}{(r-1)!(s-1)!}
  \big[(1+\alpha)^{r}+(1-\alpha)^{r}\big]^{\frac 12} \big[(1+\alpha)^{s}+(1-\alpha)^{s}\big]^{\frac 12} g^{\frac{r+s}{2}}.
$$
We also note that
$$
 \km^{di}(g,0)=\km(g).
$$

Essentially repeating the arguments in the analysis of the pure CDT model in Section \ref{sec:5.2a}, 
we see that $\km^{di}(g,\alpha)$ and $\dm^{\frac 12}\km^{di}(g,\alpha)\dm^{\frac 12}$ are 
positive definite trace-class operators on $l^2(\mathbb N)$ for $g\in(0,\frac{1}{4(1+\alpha)})$, $\alpha \in [0,1]$, and the Perron-Frobenius theorem 
applies to show that they have non-degenerate largest eigenvalues (equal to the operator norm) with eigenvectors that have positive coordinates only. 
We note that
$$ 
 \mathrm{tr}\big(\km^{di}(g,\alpha)\big)
 =\frac{1}{2}\big[\mathrm{tr}\big(\km(g(1+\alpha))\big)+\mathrm{tr}\big(\km(g(1-\alpha))\big)\big]
 = \frac{g(1+\alpha)}{2\big(1-4g(1+\alpha)\big)^{\frac{3}{2}}} + \frac{g(1-\alpha)}{2\big(1-4g(1-\alpha)\big)^{\frac{3}{2}}}
$$ 
and
$$ 
 \mathrm{tr}\big(\tm^{di}(g,\alpha)\big)
 =\mathrm{tr}\big(\tm(g(1+\alpha))\big)+\mathrm{tr}\big(\tm(g(1-\alpha))\big)
 = \frac{1-\sqrt{1-4g(1+\alpha)}}{2\sqrt{1-4g(1+\alpha)}} + \frac{1-\sqrt{1-4g(1-\alpha)}}{2\sqrt{1-4g(1-\alpha)}}.
$$ 
For $g>\frac{1}{4(1+\alpha)}$, it is seen that
$\km^{di}(g,\alpha)$ and $\tm^{di}(g,\alpha)$ are not bounded operators. By mimicking the analysis of $\dm^{\frac 12}\km(g)\dm^{\frac 12}$, together with the obvious inequality 
$$
 \big|\km^{di}_{r,s}(g,\alpha)\big| \leq\km_{r,s}\big(|g|(1+|\alpha|)\big),
$$
we see that $\dm^{\frac 12}\km^{di}(g,\alpha)\dm^{\frac 12}$ is analytic in $(\sqrt{g},\alpha)$ for $|\alpha|<1$ and $|g|<\frac{1}{4(1+|\alpha|)}$, 
taking values in the set of trace-class operators on $l_2(\mathbb N)$.

A slight variation of the argument establishing Lemma \ref{lem:Z-Zper} yields
\begin{equation}
\label{di-Z-Zper}
 Z_m^{di}(g,\alpha) \leq  \mathrm{tr}\,\big(\tm^{di}(g,\alpha)^m \big),\qquad m\in\mathbb{N},
\end{equation}
for $g\in(0,\frac{1}{4(1+\alpha)})$ and $\alpha\in [0,1]$.
This allows us to establish the following counterpart to Proposition \ref{lambdaCrit1}.

\begin{prop}
For every $\alpha\in[0,1]$, the largest eigenvalue $\lambda_1^{di}(g,\alpha)$ of \,$\dm^{\frac 12}\km^{di}(g,\alpha)\dm^{\frac 12}$\! is a strictly increasing function of $g$. 
As $g$ approaches $\tfrac{1}{4(1+\alpha)}$ from below, its limit  $\bar\lambda_1^{di}(\alpha)$ satisfies 
$$ 
 \bar\lambda_1^{di}(\alpha)\leq1.
$$
\end{prop} 

\begin{proof} 
By the variational principle for eigenvalues \cite{reed1978iv} and the fact that the matrix elements of $\km ^{di}(g,\alpha)$ are strictly increasing functions of $g >0$, one sees that $\lambda_1^{di}(g,\alpha)$ has a positive derivative with respect to $g$, and is therefore strictly increasing in $g$. Suppose the limit in question is greater than $1$ for some fixed value of $\alpha$. Then, there exists a $g_0 < \frac{1}{4(1+\alpha)}$ such 
that $\lambda_1^{di}(g_0,\alpha) =1$. A calculation similar to the one leading to \eqref{divcalc} then implies that 
$$
 W^{ev}(g,\alpha) =  1+\sum_{m=1}^\infty\big\langle v(g,\alpha)\big|\big(\tfrac{1}{2}\tm^{di}(g,\alpha)\big)^{m-1}\big|v(g,\alpha)\big\rangle \geq\frac{c'}{1-\lambda_1^{di}(g, \alpha)} - B(g, \alpha),\qquad g\in(0,g_0),
$$
where $c'>0$, while $B$ is bounded for $g$ close to $g_0$.
This shows that $W^{ev}(g,\alpha)$ diverges as $g$ approaches $g_0$. This, however, contradicts the upper bound in \eqref{primbound} with 
$k=1$, which implies that $W^{ev}(g,\alpha)$ is bounded for $g\leq g_0$.  
\end{proof} 

\noindent
Since $\tm^{di}(g,0) =2\tm(g)$ and $\tm^{di}(g,1)=\tm(2g)$, 
Proposition \ref{lambdaCrit1} implies that
\begin{equation}\label{lambda01}
 \bar\lambda_1^{di}(0)=1,\qquad 
 \bar\lambda_1^{di}(1)=\tfrac{1}{2}.
\end{equation}
Our main result on the critical behaviour of the dilute loop model is the following.
\begin{thm}\label{di-crit-beh}
For $\alpha$ real and sufficiently small, the critical coupling $g_c^{di}(\alpha)$ for the partition function $Z^{di}(g,\alpha)$ 
is determined by the equation
\begin{equation}\label{gcdi}
\lambda_1^{di}\big(g_c^{di}(\alpha),\alpha\big)=\tfrac{1}{2},
\end{equation}
and there exist $C_1(\alpha),C_2(\alpha)>0$ such that
\begin{equation}\label{di-crit}
\frac{C_1(\alpha)}{g_c^{di}(\alpha)-g} \leq Z^{di}(g,\alpha) \leq \frac{C_2(\alpha)}{g_c^{di}(\alpha)-g}
\end{equation}
for $g$ close to $g_c^{di}(\alpha)$.
 \end{thm}
 
\begin{proof}

Noting that $\lambda_1^{di}(g,\alpha)$ is a continuous function, it follows from the first identity in \eqref{lambda01} 
that $\bar\lambda_1^{di}(\alpha) >\frac{1}{2}$ for $\alpha$ sufficiently small. Since $\lambda_1 ^{di} (g,\alpha)$ is strictly increasing in $g$, the intermediate value theorem implies that the value of $g_c ^{di}(\alpha)$ satisfying \eqref{gcdi} is unique and strictly smaller than $\frac{1}{4(1+\alpha)}$, for $\alpha$ small. 
Thus, $\dm^{\frac 12} \km^{di}(g,\alpha)\dm^{\frac 12}$ is analytic in a neighbourhood of $g_c^{di}(\alpha)$ for 
$\alpha$ small and fixed. 

Mimicking the calculation leading to \eqref{divcalc} implies that, for $g\in(0,g_c^{di}(\alpha))$,
$$
 Z^{di}(g,\alpha) \geq \frac{C'_1(\alpha)}{1-2\lambda_1^{di}(g,\alpha)} - B(g,\alpha),
$$
where $C'_1(\alpha)>0$, while $B(g,\alpha)$ is bounded for $g$ close to $g_c^{di}(\alpha)$.  
As previously noted, $\lambda_1^{di}(g,\alpha)$ is an analytic function of $g$ in $(0, g_c ^{di} (\alpha))$. The lower bound in \eqref{di-crit} follows from the positivity of 
the $g$-derivative of $\lambda_1^{di}(g,\alpha)$. 
To obtain the upper bound, we apply \eqref{di-Z-Zper} to write 
$$
 Z^{di}(g,\alpha) 
 \leq 1+\sum_{m=1}^{\infty} \mathrm{tr}\big(\tm^{di}(g,\alpha)\big)^{m} 
  = 1+\sum_{m=1}^{\infty} 2^{m}\mathrm{tr}\big( \dm^{\frac{1}{2}}\km^{di}(g,\alpha)\dm^{\frac{1}{2}}\big)^{m}
  =1+\sum_{n=1}^\infty\frac{2\lambda_n^{di}(g,\alpha)}{1-2\lambda_n^{di}(g,\alpha)}
$$
for $g\in(0,g_c^{di}(\alpha))$,
where $\lambda_1^{di}(g,\alpha)>\lambda_2^{di}(g,\alpha)\geq \lambda_3^{di}(g,\alpha)\geq\dots$ denote the eigenvalues of 
$\dm^{\frac 12}\km^{di}(g,\alpha)\dm^{\frac 12}$. 
Separating out the first summand, the remaining terms yield a contribution that is bounded for $g$ close to $g_c^{di}(\alpha)$, 
and the upper bound follows as before. This establishes $g_c^{di}(\alpha)$ as the critical coupling and concludes the proof.
\end{proof}

\noindent
Theorem \ref{di-crit-beh} shows that the critical behaviour of $Z^{di}(g,\alpha)$ is different from that of the pure 
CDT model, the exponent characterising the singularity being shifted from $\frac 12$ to $-1$. For $\alpha=0$, a more detailed analysis of the corresponding height coupled trees with $k=2$ carried out in \cite{UnelInPrep1} reveals that $Z^{di}(g,0)$ has a simple pole at $g_c ^{di}(0)$ and that the Hausdorff dimension equals 1. For $\alpha>0$, such a detailed analysis is not yet available, but Theorem \ref{di-crit-beh} provides strong evidence that the singularity is still a simple pole and that the Hausdorff dimension of the disk triangulations underlying the dilute loop configurations equals 1:
$$
 d_H^{di}(\alpha\ll1)=1\quad(\mathrm{almost\ surely}).
$$

We have not determined whether the behaviour described by Theorem \ref{di-crit-beh} persists for larger values of $\alpha$. 
By the second identity in \eqref{lambda01}, we note that $g_c^{di}(1) =\frac 18$; 
in this case, the critical behaviour of $Z^{di}(g,\alpha)$ is the same as that of the pure CDT model. 
Writing the matrix elements of $\km^{di}(g,\alpha)$ in the form
$$
 \km_{r,s}^{di}(g,\alpha)=\frac{1}{2}\frac{(r+s-1)!}{(r-1)!(s-1)!}
  \Big[1+\Big(\frac{1-\alpha}{1+\alpha}\Big)^{\!r}\Big]^{\frac{1}{2}}\Big[1+\Big(\frac{1-\alpha}{1+\alpha}\Big)^{\!s}\Big]^{\frac{1}{2}}\big(g(1+\alpha)\big)^{\frac{r+s}{2}},
$$
we see that, for fixed $g(1+\alpha)\in(0,\frac{1}{4})$, they are decreasing functions of $\alpha\in[0,1]$. It follows that the limiting value $\bar\lambda_1^{di}(\alpha)$ is a decreasing 
function of $\alpha$. Hence, there are two possibilities: 
either 
(i) $\bar\lambda_1^{di}(\alpha)>\frac{1}{2}$ for all $\alpha\in[0,1)$, in which case the critical behaviour from Theorem \ref{di-crit-beh} 
extends to $\alpha\in[0,1)$, 
or 
(ii) there exists $\alpha_0<1$ such that $\bar\lambda_1^{di}(\alpha)>\frac{1}{2}$ for $\alpha\in[0,\alpha_0)$, but
$\bar\lambda_1^{di}(\alpha)=\frac{1}{2}$ for $\alpha\in(\alpha_0,1]$.
If case (ii) eventuates, the critical behaviour of $Z^{di}(g,\alpha)$ for $\alpha\in[\alpha_0,1)$, 
in particular at the transition point $\alpha=\alpha_0$, would be an interesting subject of study.

\section{Discussion} 
\label{sec:6}

We have introduced a dense and a dilute loop model on causal triangulations and studied their critical behaviour 
through examinations of their partition functions and of eigenvalue properties established using a transfer matrix formalism. 
Both models admit a description in terms of labelled planar trees, but only the dilute loop model experiences an effective height coupling and
(for some values of $\alpha$) a change in the Hausdorff dimension. Although the labelled tree correspondences were important in our analysis, it is conceivable that this feature is, in fact, a limitation when it comes to exhibiting
nontrivial couplings of matter to geometry.

A natural generalisation of the dilute loop model that does not readily admit a tree correspondence, is obtained by assigning a separate weight $\gamma$ to arcs 
whose endpoints are on the time-like edges of an elementary triangle, as in the second diagrams in Figure \ref{figure:dilute1}. 
With $t(L)$ denoting the number of such arcs in a loop configuration $L\in{\mathcal L}_m^{di}$, the fixed-height partition functions are given by
\begin{equation}\label{GenPart}
 Z_m^{di}(g,\alpha,\gamma) := \sum_{L\in{\mathcal L}_m^{di}} g^{|L|}\alpha^{s(L)}\gamma^{t(L)}.
\end{equation}
As in the discussion of $Z_m^{di}(g,\alpha)$ defined in (\ref{equ:diPf}),
it follows that, for given $\alpha,\gamma\in[0,1]$, there exists a critical coupling $g_c^{di}(\alpha,\gamma)$ 
such that $Z_m^{di}(g,\alpha,\gamma)$ is finite for $g<g_c^{di}(\alpha,\gamma)$ and divergent for $g>g_c^{di}(\alpha,\gamma)$.
Since $Z_m^{di}(g,\alpha,\gamma)\leq Z_m^{di}(g,\alpha,1)=Z_m^{di}(g,\alpha)$, we see that
$$ 
 g^{di}_{c}(\alpha, \gamma) \ge g^{di}_c(\alpha, 1)=g^{di}_c(\alpha). 
$$
Attributing a unity of length to each arc, one may view $s(L)+t(L)$ as the total length of $L$.
For $\gamma=\alpha$, the weight of a loop configuration can thus be expressed in terms of area and length, 
in the sense noted for the dense loop model at the end of Section \ref{subsec:DLM}. 
On the other hand, the tree correspondence of Proposition \ref{diluteloop-trees} is seen to introduce long-distance couplings on the trees, 
thereby limiting its usefulness in this case. Related to this, it is not possible to perform the summation over loop labellings and triangulations independently
in (\ref{GenPart}), as was done in the dense and dilute models. 
This intimate connection between matter and geometry is an interesting yet challenging aspect of this new model. A fully-packed version of the dilute loop model, where empty triangles are excluded, can be obtained as a limiting case of the generalised dilute loop model described above and will likewise exhibit long-distance tree couplings. 

It would also be interesting to study the dense and dilute loop models with loop fugacities different from $1$ and to
explore other types of loop models on causal triangulations. For example, one may adapt
the so-called fused Temperley-Lieb loop models \cite{ZJ07,PRT14,MDPR14} based on \cite{BR89,FR02}, and loop models whose underlying algebraic structures are given by Birman-Wenzl-Murakami algebras \cite{BW89,Mur87}.
We anticipate that one may describe and analyse such models, including partition functions and correlation functions more generally,
using an extension of the transfer matrix formalism employed in Section \ref{sec:5}.

\subsection*{Acknowledgements}

BD and M{\"U} acknowledge financial support from Villum Fonden via the QMATH Centre of Excellence (Grant no.~10059).
XP was supported by an Australian Postgraduate Award from the Australian Government.
JR was supported by the Australian Research Council under the Discovery Project scheme, project numbers DP160101376 and DP200102316.
XP and JR thank QMATH for their hospitality during a visit to the University of Copenhagen in February and March 2020.

\end{document}